\documentclass[sigconf, nonacm]{acmart}

\usepackage[linesnumbered,ruled,vlined]{algorithm2e}
\SetKwRepeat{Do}{do}{while}% add do-while sentences
\SetCommentSty{mycommfont}
\usepackage{enumitem}
\usepackage{xspace}
\usepackage{bm}
\usepackage{mathrsfs}
\usepackage{booktabs}
\usepackage{colortbl}
\usepackage{bm}
\usepackage{amsthm}
\usepackage{tabularx,subfigure,multirow, graphicx}

\usepackage{xspace}
\usepackage{colortbl} % 用于为表格添加颜色
\usepackage{xcolor}   % 提供颜色定义功能
\usepackage[a-2b]{pdfx}%enable pdfa
\usepackage[pdfa]{hyperref} % add hyper reference to citations and references
\newcommand{\stitle}[1]{\vspace{1mm} \noindent{\bf #1}}

\definecolor{newtext}{rgb}{0, 0, 0.61}

\definecolor{grayA}{gray}{0.95}  % 最接近白色
\definecolor{grayB}{gray}{0.85}  % 中等对比度
\definecolor{grayC}{gray}{0.75}  % 高对比度
\setlist[enumerate]{label=\arabic*.,leftmargin=5mm,itemsep=5pt}
\setlist[itemize]{leftmargin=5mm,itemsep=5pt}

\newtheorem{theorem}{\bf Theorem}[section]

\newtheorem{example}{\bf Example}

\theoremstyle{remark}

\theoremstyle{definition}

\newcommand\vldbdoi{XX.XX/XXX.XX}
\newcommand\vldbpages{XXX-XXX}
\newcommand\vldbvolume{18}
\newcommand\vldbissue{12}
\newcommand\vldbyear{2025}
\newcommand\vldbauthors{Xiaoyao Zhong, 
	Haotian Li,
	Jiabao Jin, 
	Mingyu Yang,
	Deming Chu,
	Xiangyu Wang,
	Zhitao Shen,
	Wei Jia,
	George Gu,
	Yi Xie,
	Xuemin Lin,
	Heng Tao Shen,
	Jingkuan Song, Peng Cheng}
\newcommand\vldbtitle{\shorttitle} 
\newcommand\vldbavailabilityurl{https://github.com/antgroup/vsag}
\newcommand\vldbpagestyle{empty} 

\newcommand\vsag{\textsl{VSAG}\xspace}

\begin{document}

\title{VSAG: An Optimized Search Framework for Graph-based Approximate Nearest Neighbor Search}\thanks{Corresponding Author: Peng Cheng.}

\author{Xiaoyao Zhong, Haotian Li}
\affiliation{%
	\institution{Ant Group, Shanghai, China}
}
\email{zhongxiaoyao.zxy@antgroup.com}
\email{tianlan.lht@antgroup.com}

\author{Jiabao Jin, Mingyu Yang}
\affiliation{%
	\institution{Ant Group, Shanghai, China}
}
\email{jinjiabao.jjb@antgroup.com}
\email{	yangming.ymy@antgroup.com}

\author{Deming Chu, Xiangyu Wang}
\affiliation{%
	\institution{Ant Group, Shanghai, China}
}
\email{chudeming.cdm@antgroup.com}
\email{wxy407827@antgroup.com}

\author{Zhitao Shen, Wei Jia}
\affiliation{%
	\institution{Ant Group, Shanghai, China}
}
\email{zhitao.szt@antgroup.com}
\email{jw94525@antgroup.com}

\author{George Gu, Yi Xie}
\affiliation{%
	\institution{Intel Corporation, Shanghai, China}
}
\email{george.gu@intel.com}
\email{ethan.xie@intel.com}

\author{Xuemin Lin}
\affiliation{%
	\institution{Shanghai Jiaotong University}
	\city{Shanghai}
	\country{China}
}
\email{xuemin.lin@gmail.com}

\author{Heng Tao Shen, Jingkuan Song}
\affiliation{%
	\institution{Tongji University, Shanghai, China}
}
\email{shenhengtao@hotmail.com}
\email{jingkuan.song@gmail.com}

\author{Peng Cheng}
\affiliation{%
	\institution{Tongji University}
	\institution{East China Normal University}
	\city{Shanghai}
	\country{China}
}
\email{cspcheng@tongji.edu.cn}

\begin{abstract}
Approximate nearest neighbor search (ANNS) is a fundamental problem in vector databases and AI infrastructures. Recent graph-based ANNS algorithms have achieved high search accuracy with practical efficiency. Despite the advancements, these algorithms still face performance bottlenecks in production, due to the random memory access patterns of graph-based search and the high computational overheads of vector distance. In addition, the performance of a graph-based ANNS algorithm is highly sensitive to parameters, while selecting the optimal parameters  is cost-prohibitive, e.g., manual tuning requires repeatedly re-building the index.
This paper introduces \vsag, an open-source framework that aims to enhance the in production performance of graph-based ANNS algorithms.
\vsag has been deployed at scale in the services of Ant Group, and it incorporates three key optimizations: \textit{(i) efficient memory access}: it reduces L3 cache misses with pre-fetching and cache-friendly vector organization; \textit{(ii) automated parameter tuning}: it automatically selects performance-optimal parameters without requiring index rebuilding;
\textit{(iii) efficient distance computation}: it leverages modern hardware, scalar quantization, and smartly switches to low-precision representation to dramatically reduce the distance computation costs.
We evaluate \vsag on real-world datasets.
The experimental results show that \vsag achieves the state-of-the-art performance and provides up to $4\times$ speedup over HNSWlib (an industry-standard library) while ensuring the same accuracy. 
\end{abstract}

\maketitle

\pagestyle{\vldbpagestyle}
\begingroup\small\noindent\raggedright\textbf{PVLDB Reference Format:}\\
\vldbauthors. \vldbtitle. PVLDB, \vldbvolume(\vldbissue): \vldbpages, \vldbyear.\\
\href{https://doi.org/\vldbdoi}{doi:\vldbdoi}
\endgroup
\begingroup
\renewcommand\thefootnote{}\footnote{\noindent
This work is licensed under the Creative Commons BY-NC-ND 4.0 International License. Visit \url{https://creativecommons.org/licenses/by-nc-nd/4.0/} to view a copy of this license. For any use beyond those covered by this license, obtain permission by emailing \href{mailto:info@vldb.org}{info@vldb.org}. Copyright is held by the owner/author(s). Publication rights licensed to the VLDB Endowment. \\
\raggedright Proceedings of the VLDB Endowment, Vol. \vldbvolume, No. \vldbissue\ %
ISSN 2150-8097. \\
\href{https://doi.org/\vldbdoi}{doi:\vldbdoi} \\
}\addtocounter{footnote}{-1}\endgroup

\ifdefempty{\vldbavailabilityurl}{}{
\vspace{.3cm}
\begingroup\small\noindent\raggedright\textbf{PVLDB Artifact Availability:}\\
The source code, data, and/or other artifacts have been made available at \url{\vldbavailabilityurl}.
\endgroup
}

\section{Introduction}

{

	At Ant Group~\cite{antgroup}, we have observed an increasing demand to manage large-scale high-dimensional vectors across different departments.
	This demand is fueled by two factors.
	First, the advent of Retrieval-Augmented Generation (RAG) for large language models (LLMs)~\cite{LLM-RAG-NIPS-2020, gao2024RAGsurvey} requires vector search to address issues such as hallucinations and outdated information. 
	Second, the explosive growth of unstructured data (e.g., documents, images, and videos), requires efficient analysis and storage methods.
	Many systems transform these unstructured data into embedding vectors for efficient retrieval, e.g., Alipay's facial-recognition payment~\cite{alipay}, Google's image search~\cite{google}, and YouTube's video search~\cite{youtube}.

	Approximate nearest neighbor search (ANNS) is the foundation for these AI and LLM applications. Due to the curse of dimensionality~\cite{Curse-of-dim-1998}, exact nearest neighbor search becomes prohibitively expensive as dimensionality grows. ANNS, however, trades off a small degree of accuracy for a substantial boost in efficiency, establishing itself as the gold standard for large-scale vector retrieval.

	Recently, graph-based ANNS algorithms (e.g., HNSW~\cite{HNSW-PAMI-2020} and VAMANA~\cite{Diskann-NIPS-2019}) successfully balance high recall with practical runtime performance. 
	These methods typically construct a graph, where each node is a  vector and each edge connects nearby vector pairs.
	During a query, an approximate $k$-nearest neighbor search starts from a random node and greedily moves closer to the query vector $x_q$, thereby retrieving its $k$ nearest neighbors.

	Despite their success, existing graph-based ANNS solutions still face considerable performance challenges. 
	First, they incur \emph{random memory-access overhead}, since graph traversals with arbitrary jumps often lead to frequent cache misses and elevated costs. 
	Second, repeated \emph{distance computations} across candidate vectors can dominate total runtime, especially when vectors are high-dimensional. 
	Finally, performance is highly \emph{sensitive to parameter settings} (e.g., maximum node degree and candidate pool size), yet adjusting these parameters generally requires rebuilding the index, which can take hours or even days.
	We use a small example from our experiments to illustrate these issues:

	Modern production systems~\cite{Milvus-SIGMOD-2021,faiss:johnson2019billion} typically employ vector quantization to reduce the distance computation cost.
	Therefore, we set our baseline as \textbf{HNSW} with \textbf{SQ4} quantization~\cite{2012sq}.
	We conduct 1,000 vector queries on the GIST1M.
	In what follows, we report the performance limitations of graph-based ANNS algorithms, using the experimental evidence of the baseline:
	\textit{(i) high memory access cost:} 
	each query needs over 2{,}959 random vector accesses (total 1.4\,MB), causing a 67.42\% L3 cache miss rate. 
	The memory-access operations consume 63.02\% of the search time.
	\textit{(ii) high parameter tuning cost}:
	if we use the optimal parameters instead of the manually selected values, the QPS can increase from 1,530 to 2,182, by 42.6\%.
	However, brute-force tuning of parameters takes more than 60 hours, which is prohibitively expensive.
	\textit{(iii) high distance computation cost}: distance computations still take 26.12\% of the search time despite using SQ4 quantization.

}

	\begin{table}[t!]
		\caption{\small Comparison to Existing Algorithms (GIST1M).}\vspace{-3ex}
		\label{tab:intro:algo_comp}
		\resizebox{\linewidth}{!}{
			\begin{tabular}{l|l|l|>{\columncolor[gray]{0.9}}l}
				\toprule 
				\textbf{Metric } & \textbf{IVFPQFS~\cite{PQ-fast-scan-VLDB-2015}} & \textbf{HNSW~\cite{HNSW-PAMI-2020}} & \textbf{VSAG (this work)} \\
				\midrule
				Memory Footprint & 3.8G & 4.0G & 4.5G \\
				Recall@10 (QPS=2000) & 84.57\% & 59.46\% & 89.80\% \\ 
				QPS \hspace{2.4em}  (Recall@10=90\%) & 1195 & 511.9 & 2167.3 \\
				Distance Computation Cost & 0.71ms  & 1.62ms & 0.1ms \\
				L3 Cache Miss Rate & 13.98\% & 94.46\% & 39.23\% \\
				Parameter Tuning Cost & >20h & >60h & 2.92h \\
				Parameter Tuning & manual & manual & auto \\
				\bottomrule 
			\end{tabular}
		}\vspace{-6ex}
	\end{table}

\stitle{Contributions.} 
This paper presents \vsag, an open-source framework for enhancing the in-production efficiency of graph-based ANNS algorithms.
The optimizations of \vsag are in three-fold.
\textit{(i) Efficient Memory Access}: during graph-based search, it pre-fetches the neighbor vectors, and creates a continuous copy of the neighbor vectors for some vertices.
This cache-friendly design can reduce L3 cache misses.
\textit{(ii) Automated Parameter Tuning:} \vsag can automatically tune parameters for environment (e.g., prefetch depth), index (e.g., max degree of graph), and query (e.g., candidate size).
Suppose there are 3 index parameters and 5 choices for each parameter.
The tuned index parameters of \vsag offer similar performance to that of brute-force tuning, which needs to enumerate all $5^{3}$ combinations of parameters leading to a total tuning time of $5^{3}$ times of index construction time.
On the contrast, the tuning cost of \vsag is only 2-3 times of index construction time. \textit{(iii) Efficient Distance Computation}: \vsag provides various approximate distance techniques, such as scalar quantization. All distance computation is well optimized with modern hardware, and a selective re-ranking strategy is used to ensure retrieval accuracy.

Table \ref{tab:intro:algo_comp} compares the performance of \vsag with existing works on GIST1M.
The results show that \vsag alleviates the performance challenges in memory access, parameter tuning, and distance computation, thus providing higher QPS with the same recall rate.

In summary, we make the following contributions:

\begin{enumerate}
	\item We enhance the memory access of \vsag in \S\ref{sec:tech1_cache_miss_opt}. The L3 cache miss rate of \vsag is much less than that of other graph-based ANNS works.
	
	\item We propose automatic parameter tuning for \vsag in \S\ref{sec:tech2_auto_param_opt}.
	It automatically selects performance-optimal parameters that are comparable to grid search without requiring index rebuilding.

	\item We accelerate \vsag in distance computation in \S\ref{sec:tech3_distance_opt}.
	Compared with other graph-based ANNS works, \vsag requires much less time for distance computation.
	
	\item {We evaluate the algorithms on real datasets with sizes up to 100 million vectors in \S\ref{sec:experimental}.}
	The results show that \vsag can achieve the SOTA performance and outperform HNSWlib by up to 4$\times$ in QPS under the same recall guarantee.   
\end{enumerate}

\section{Overview of VSAG Framework}
\label{sec:framework}

\begin{table}[t!]
	\centering
	{\small\scriptsize
		\caption{\small Symbols and Descriptions.}\vspace{-4ex} \label{tab:symbol}
		\begin{tabular}{l|l}\toprule
			{\bf Symbol} & {\bf Description} \\ \hline 
			$D$   & the base dataset \\
			$G$   & the graph index \\
			$L$   & the labels of edges \\
			$\tau_l,\tau_h$ & the distance function with low precision and high precision \\
			
			$x,x_b, x_q, x_n$   & a normal, base, query, and neighbor vector\\
			$NN_k(x), ANN_k(x)$ & the $k$ nearest and approximate $k$ nearest neighbors of $x$\\
			
			$ef_s, ef_c$   & the candidate pool size in search and construction phase\\
			
			$\alpha_s, \alpha_c$ & the pruning rate used in search and construction phase \\
			
			$m_s, m_c$ & the maximum degree of graph used in search and construction phase \\
			
			$\omega$   & prefetch stride\\
			$\nu$   & prefetch depth\\
			
			$\delta$   & redundancy ratio\\
			
			\bottomrule
		\end{tabular}
	}\vspace{-6ex}
\end{table}

\begin{figure*}[t]\centering
	\includegraphics[width=\linewidth]{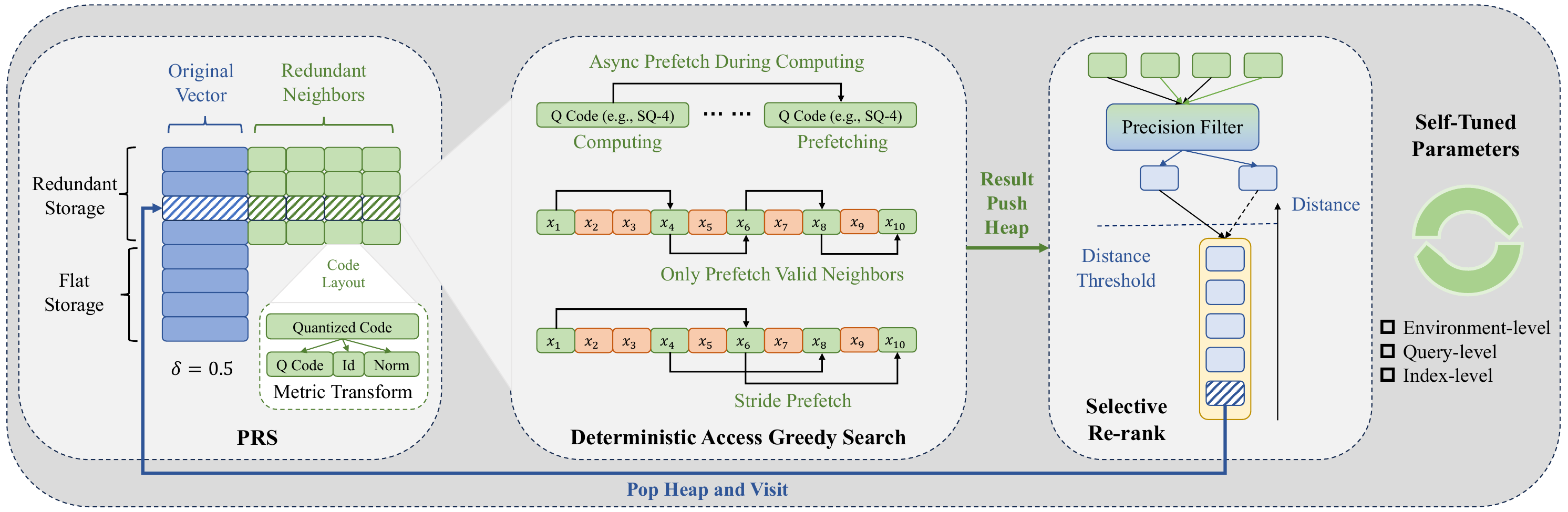}\vspace{-2ex}
	\caption{The Search Framework of \vsag.}%
	\label{fig:framework}\vspace{-2ex}
\end{figure*}

This section provides an overview of our \vsag framework, which includes memory access optimization, automatic parameter tuning, and distance computation optimization. As shown in Figure~\ref{fig:framework}, \vsag integrates these optimizations into different phases of the search process. First, the \textit{PRS} manages the storage of vector data and graph indexes. It stores low-precision quantized codes of vectors alongside high-precision original vectors to support distance computation optimization. This redundant storage is also utilized to balance resource usage and reduce memory access costs. Second, during each hop of the search process when exploring a vector, we employ \textit{Deterministic Access Greedy Search} to accelerate memory access by minimizing cache misses. Finally, the visited results are pushed into a heap. We apply \textit{Selective Re-rank} to combine low- and high-precision distances, ensuring effective search performance. The search process then pops the nearest unfinished point from the heap and proceeds to the next hop. Throughout    the entire life-cycle of the index, \vsag uses a smart auto-tuner to select parameters that deliver the best search performance across varying environments and retrieval requirements. We will now detail the three main optimizations used in \vsag. The symbols used in this paper are listed in Table~\ref{tab:symbol}.

\subsection{Memory Access Optimization}

\stitle{Deterministic Access.} Distance computations for neighboring vectors often incur random memory access patterns in graph-based algorithms, leading to significant cache misses. \vsag addresses this by integrating software prefetching~\cite{2020prefetchAndMemoryAccess} (i.e., \texttt{\_mm\_prefetch}) through a \textit{Deterministic Access} strategy (see \S\ref{sec:deterministic_access}). \vsag strategically inserts prefetch instructions during and before critical computations, proactively loading target data into L3/higher-level caches. This prefetch-pipeline overlap ensures data availability before subsequent computation phases begin, effectively mitigating cache miss penalties. Furthermore, the \vsag framework effectively mitigates suboptimal prefetch operations through batch processing and reordering of the access sequence.

\stitle{PRS.} \vsag introduces a \textit{Partial Redundant Storage (PRS)} design (see \S\ref{sec:prs}), which provides a flexible and high-performance storage foundation to optimize both distance computations and memory access while balancing storage and computational resource usage. In production environments constrained by fixed hardware configurations, such as 4C16G (i.e., equipped with 4 CPU cores and 16GB of memory) and 2C8G (i.e., equipped with 2 CPU cores and 8GB of memory), most algorithms frequently exhibit resource utilization imbalances between computational and memory subsystems. During computational processes, CPUs frequently encounter idle cycles caused by cache misses, which hinders their ability to achieve optimal utilization, and thereby limits the system's QPS. 

To address this challenge, the PRS framework \textit{Redundantly Storing Vectors}  (see \S\ref{sec:redundant_storing_vector}) that embeds compressed neighbor vectors at each graph node. This architectural design enables batched distance computations while leveraging more efficient \textit{Hardware-based Prefetch}~\cite{chen1995hardwarePrefetch} (see \S\ref{sec:hardware_prefetch}) to maintain high cache hit rates. By incorporating advanced quantization methods~\cite{jegou2010pq_ivfadc, 2024rabitq}, PRS achieves high vector storage compression ratios, thereby maintaining acceptable storage overhead despite data redundancy.  

\stitle{Balance of Resources.} In particular, the system uses a parameter called the redundancy ratio $\delta$ to control  the \textit{Balance of Computational Efficiency and Memory Utilization} (see \S\ref{sec:balance_cpu_memory}). In compute-bound scenarios with high-throughput demands, \vsag adaptively increases the redundancy ratio to mitigate cache contention, thus minimizing CPU idle cycles during memory access while preserving storage efficiency. In contrast, in memory-constrained low-throughput scenarios, the framework strategically reduces redundancy ratio to optimize the index-memory footprint. Under memory-constrained conditions, this optimization enables deployment on reduced instance tiers, thereby curtailing compute wastage.

\subsection{Automatic Parameter Tuning}
\vsag addresses parameter selection complexity through a tripartite classification system with specialized optimization strategies: \textit{environment-level}, \textit{query-level}, and \textit{index-level} parameters. Environment-level parameters (e.g., prefetch stride $\omega$) exclusively influence query-per-second (QPS) performance without recall rate impacts, thus incurring the lowest tuning overhead. Query-level parameters (e.g., candidate set size $ef_s$~\cite{malkov2018hnsw}) exhibit moderate tuning costs by jointly affecting QPS and recall, requiring adjustment based on query vector distributions. Index-level parameters (e.g., maximum degree $m_c$~\cite{malkov2018hnsw}) demand the highest tuning investment due to their tripartite impact on QPS, recall, and index construction time – parameter validation necessitates multiple index rebuilds.

\begin{itemize}[leftmargin=*, itemsep=1pt]
	\item \textit{Environment-level Parameters} (see \S\ref{sec:elp_tuner}): \vsag employs a grid search to identify optimal configurations for peak QPS performance through systematic parameter space exploration.
	
	\item \textit{Query-level Parameters} (see \S\ref{sec:qlp_tuner}): \vsag implements multi-granular tuning strategies, including \textit{fine-grained adaptive optimization} that dynamically adjusts parameters based on real-time query difficulty assessments.
	
	\item \textit{Index-level Parameters} (see \S\ref{sec:ilp_tuner}): \vsag introduces a novel mask-based index compression technique that encodes multiple parameter configurations into an unified index structure. During searches, edge-label filtering dynamically emulates various construction parameters, thereby reducing index-level parameters to query-level equivalents while keeping a single physical index.
\end{itemize}

\subsection{Distance Computation Optimization}

Distance computation is a main overhead in vector retrieval, and its cost increases significantly with the growth of vector dimensions.
\textit{Quantization methods} (see \S\ref{sec:minimize_computation_overhead}) can effectively accelerate distance computation. For example, under identical instruction set architectures, \texttt{AVX512}~\cite{avx512} can process $4\times$ as many \texttt{INT8} data per instruction compared to \texttt{FLOAT32} values. However, naive quantization approaches often result in significant accuracy degradation.
\vsag uses \textit{Selective Re-rank} (see \S\ref{sec:selective_rerank}) to improve efficiency without sacrificing the search accuracy.
Furthermore, specific distance metrics (i.e., euclidean distance) can be strategically decomposed and precomputed, effectively reducing the number of required instructions during actual search operations.

\section{Memory Access Optimization}
\label{sec:tech1_cache_miss_opt}

Graph-based algorithms suffer from random memory access patterns that incur frequent cache misses. The fundamental strategy for mitigating cache-related latency lies in effectively utilizing vector computation intervals to prefetch upcoming memory requests into cache. In \vsag, three primary optimization strategies emerge for maximizing cache utilization efficiency:

\begin{itemize}[itemsep=1pt]
	\item Leveraging software prefetching to improve cache hit rates.
	\item Optimizing search patterns to enhance the effectiveness of software prefetching.
	\item Optimizing the memory layout of indexes to efficiently utilize hardware prefetching.
\end{itemize}

\subsection{Software-based Prefetch: Making Random Memory Accesses Logically Continuous}
\label{sec:software_prefetch}

As shown in Figure~\ref{fig:tech1-software-prefetch}, when computing vector distances, the vector is loaded sequentially from a segment of memory. The CPU fetches data from memory in units of cache lines~\cite{veidenbaum1999cacheline}. Consequently, multiple consecutive cache fetch operations are triggered for a single distance computation. 

\begin{example}
	Take standard 64-byte cache line architectures as an example. The 960-dimensional vector of GIST1M stored as float32 format necessitates  $960 \times 4 / 64 = 60$  cache line memory transactions, demonstrating significant pressure on memory subsystems.
\end{example}

This passive caching mechanism creates operational inefficiencies by extensively fetching data only upon cache misses, resulting in synchronous execution bottlenecks. As illustrated in Figure~\ref{fig:tech1-software-prefetch}, the orange timeline shows how the regular ANNS algorithm serializes the computation and memory access phases: Each cache line fill (analogous to blocking I/O) stalls computation until completion. The accumulated latency from successive cache line transfers introduces a significant constant factor in complexity analysis, particularly in memory-bound scenarios with poor data locality.

\begin{figure}[t!]\centering\vspace{-4ex}
	\includegraphics[width=0.9\linewidth]{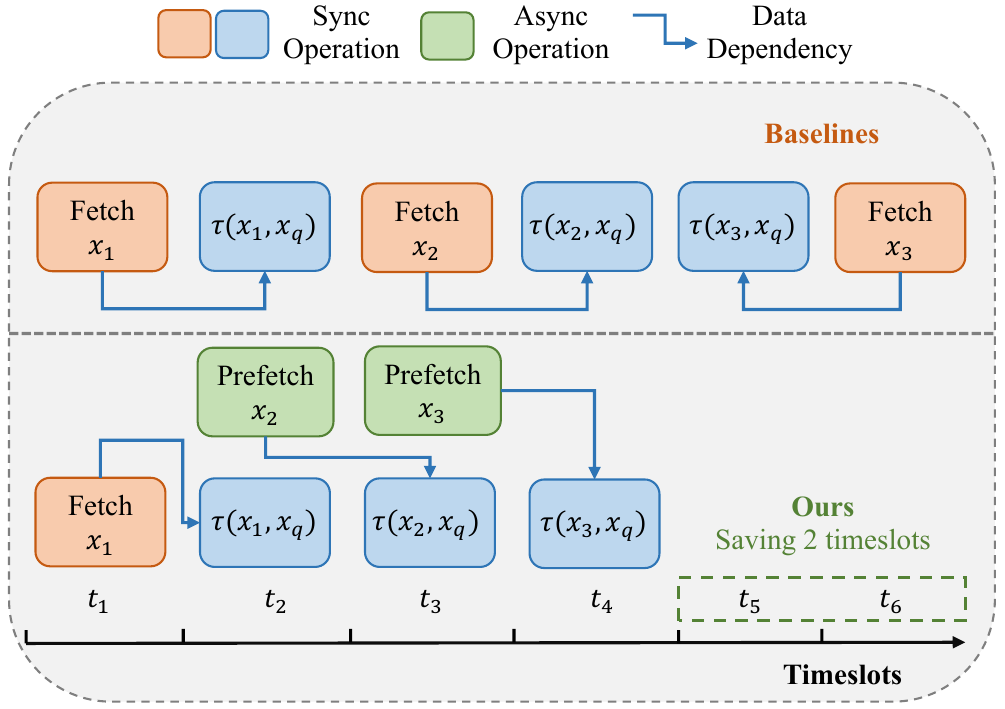}\vspace{-2ex}
	\caption{\small Passive Memory Access and Software-based Prefetch.
	}\vspace{-3ex}
	\label{fig:tech1-software-prefetch}
	\vspace{-1ex}
\end{figure}

\noindent\textbf{Software-based Prefetch.} Modern CPUs support software prefetch instructions, which can asynchronously load data into different levels of the cache~\cite{2020prefetchAndMemoryAccess}. By leveraging the prefetch instruction to preload data, we can achieve a near-sequential memory access pattern from the CPU level. This indicates that data is preloaded into the cache before the CPU requires it, thereby preventing disruptions in the computational flow caused by random memory address loads. More specifically, prior to computing the distance for the current neighbor, the vector for the next neighbor can be prefetched. As detailed in Figure~\ref{fig:tech1-software-prefetch}, the green flow represents the use of prefetching. From the perspective of the CPU, the majority of distance computations make use of data that has already been cached. Furthermore, because prefetching operates asynchronously, it does not obstruct ongoing computations. The synergy of asynchronous prefetching and immediate data access optimizes the utilization of CPU computational resources, thereby substantially enhancing search performance.

\subsection{Deterministic Access Greedy Search: Advanced Prefetch Validity}

In \S\ref{sec:software_prefetch}, the software-based prefetch mechanism initiates the fetching of next-neighbor vector upon completion of each neighbor computation. However, this method results in redundant operations because previously visited neighbors that do not require distance computations still generate prefetch time cost. Example \ref{example:prefect}  illustrates the inherent limitations of previous prefetch schemes:
\begin{example}\label{example:prefect}
	In Figure \ref{fig:tech1-batch-prefetch}(a), $x_2$ and $x_3$ have already been visited, and the distances do not need to be recomputed. This renders the previous prefetching of these vectors ineffective. Additionally, when computing $x_4$, the prefetch may also fail due to the prefetching gap being too short.
\end{example}

In this section, we present two dedicated strategies to address the aforementioned prefetching challenges.

\begin{figure}[t!]\centering\vspace{-4ex}
	\scalebox{0.45}[0.45]{\includegraphics{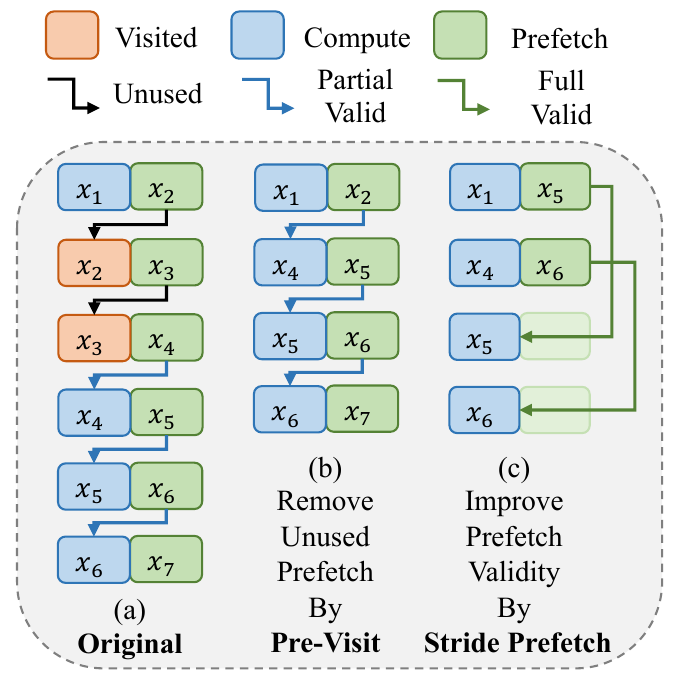}}\vspace{-2ex}
	\caption{\small Three Different Prefetching Strategies}\vspace{-3ex}
	\label{fig:tech1-batch-prefetch}
	
	\vspace{-1ex}
\end{figure}

\subsubsection{Deterministic Access} 
\label{sec:deterministic_access}
In contrast to prefetching during edge access checks, \vsag exclusively prefetches only those edges that have not been accessed. The mechanism begins by batch-processing all neighbor nodes to determine their access status. Following this, unvisited neighbors are logically grouped, and prefetching is performed collectively. This strategic approach ensures that each prefetched memory address corresponds exclusively to computation-essential data, thereby enhancing prefetching efficiency and minimizing redundant memory operations.

\subsubsection{Stride Prefetch}
\label{sec:jumping_prefetch}
Batch processing ensures that each prefetch retrieves data intended for future use. However, prefetch effectiveness varies due to the asynchronous nature of prefetching and the absence of a callback mechanism to confirm prefetch completion. Optimal performance occurs when the required data reside in the cache precisely when the computation flow demands it. Premature prefetching risks cache eviction, while delayed prefetching negates performance gains. This necessitates balancing prefetch timing with computation duration. To address this, stride prefetching dynamically aligns hardware computation throughput with software prefetching rates, maximizing prefetch utility. The key parameter, the prefetch stride $\omega$~\cite{vanderwiel2000stride}, determines how many computation steps occur before each prefetch. Adjusting $\omega$ is crucial, and in \S\ref{sec:elp_tuner}, we propose an automated strategy to select its optimal value.

\begin{algorithm}[t!]\small
	\DontPrintSemicolon
	\KwIn{graph $G$, labels $L$, base dataset $D$, initial nodes $I$, query point $x_q$, low- and high- precision distance functions $\tau_l$ and $\tau_h$,
		search parameters $k$, $ef_s$, $m_s$, $\alpha_s$, $\omega$, $\nu$}
	\KwOut{ $ANN_k(x_q)$ and their high-precision distances $T$}
	candidate pool $C$ $\leftarrow$ maximum-heap with size of $ef_s$\;
	
	visited set $V \leftarrow \emptyset$\;
	
	insert $(x_i,\tau_l(x_i, x_q)),\forall x_i\in I$ into $C$ \;
	
	\While{$C$ has un-expended nodes} {
		$x_i \leftarrow $ closest un-expended nodes in $C$ \;
		$N \leftarrow $ empty list \;
		\For{$j \in G_{i}$} {  \tcp{Only retrieve Id(i.e., $x_j.id = j$)}
			\If{ $j \not \in V$
				\textbf{\textup{and}} $L_{j} \le \alpha_s$
				\textbf{\textup{and}} $|N| < m_s$} {
				$N\leftarrow N\cup \{j\}$ \;
				$V\leftarrow V\cup \{j\}$ \;
			}
		}
		
		\For{$k \in [0, min(\omega, |N|))$} {
			prefetch $\nu$ cache lines start from $D_{N_k}$  \;
		}
		\For{$k \in [0, |N|)$} {
			\If{$k + \omega < |N|$} {
				prefetch $\nu$ cache lines start from $D_{N_{k + \omega}}$ \;
			}
			$x_j \leftarrow D_{N_k}$ \tcp{Memory Access}
			
			insert $(x_j, \tau_l(x_j, x_q))$ into $C$ and keep $|C| \le ef_s$\;
		}
	}
	
	$ANN_k(x_q), T \leftarrow$ selective re-rank $C$ with $\tau_l$ and $\tau_h$ \;
	
	\Return{$ANN_k(x_q), T$} \;
	\caption{Deterministic Access Greedy Search}
	\label{algo:greedy}
\end{algorithm}

\begin{example}
	As detailed in Figure~\ref{fig:tech1-batch-prefetch}(b), the adjusted pattern demonstrates that during batch processing, the \textbf{Deterministic Access} strategy eliminates the need to access $x_2$ and $x_3$. Consequently, the search logic progresses from $x_1$ directly to $x_4$. This sequence modification enables the prefetch mechanism to target $x_4$ while computing $x_1$. Figure~\ref{fig:tech1-batch-prefetch}(c) further reveals the temporal characteristics of asynchronous prefetching: The data loading process requires two vector computation cycles to populate the cache line. When computation for $x_1$ initiates, only the $x_4$ vector can be prefetched. After the computations of $x_1$ and $x_4$, the \textbf{Stride Prefetch} strategy ensures timely cache population of $x_6$ data, which is immediately available for subsequent computation.
\end{example}

\noindent\textbf{Deterministic Access Greedy Search.} 
The cache-optimized search algorithm is formalized in Algorithm~\ref{algo:greedy}. The graph index $G$ constitutes an oriented graph that maintains base vectors along with their neighbors. The labels $L$ of edges in $G$ are used for automatic index-level parameters tuning (see \S\ref{sec:tech2_auto_param_opt}). We use $G_i$ and $L_i$ to indicate the out-edges and labels of $x_i$. The low- and high-precision distance functions $\tau_l$ and $\tau_h$ are used to accelerate distance computation while maintaining search accuracy, and they are employed in the selective re-ranking process (see \S\ref{sec:tech3_distance_opt}). The complete algorithm explanation is provided in Appendix A.

\subsection{PRS: Flexible Storage Layout Boosting Search Performance}
\label{sec:prs}

While incorporating well-designed prefetch patterns into search processes can theoretically improve performance, the inherent \textbf{limitations of Software-based Prefetch} prevent guaranteed memory availability for all required vectors. This phenomenon can be attributed to multiple fundamental constraints:
(a) Prefetch instructions remain advisory operations rather than mandatory commands. Even when optimal prefetch patterns are implemented, their actual execution cannot be assured. 
(b) Cache line contention represents another critical challenge. In multi-process environments, aggressive prefetch strategies may induce L3 cache pollution through premature data loading.
(c) The intrinsic cost disparity between prefetch mechanisms further compounds these issues. Software-based prefetching intrinsically carries higher operational costs and demonstrates inferior efficiency compared to hardware-implemented alternatives.

\subsubsection{Hardware-based Prefetch} 
\label{sec:hardware_prefetch}
Hardware-based prefetching relies on hardware mechanisms that adaptively learn from cache miss events to predict memory access patterns. The system employs a training buffer that dynamically identifies recurring data access sequences to automatically prefetch anticipated data into the cache hierarchy. Compared to software-controlled prefetching, this hardware approach demonstrates better runtime efficiency while functioning transparently at the architectural level. The training mechanism shows particular effectiveness for \textit{Sequential Memory Access} patterns~\cite{braun2019sma}, where it can rapidly detect and exploit sequential memory access characteristics. This optimization proves particularly beneficial for space-partitioned index structures like inverted file-based index~\cite{PQ-PAMI-2014}, where vectors belonging to the same partition maintain contiguous storage allocation. Conversely, graph-based indexing architectures exhibit irregular access patterns with poor spatial locality, resulting in inefficient \textit{Random Memory Access}~\cite{braun2019sma}. The inherent randomness of access sequences prevents the training buffer from establishing effective pattern recognition models.

\subsubsection{Redundantly Storing Vectors}
\label{sec:redundant_storing_vector}
\vsag integrates the benefits of space-partitioned indexes into graph-based indexing algorithms through redundant vector storage. By co-locating neighbor lists with their corresponding vectors within each node's data structure, it achieves \textit{Sequential Memory Access}. This design ensures that neighbor retrieval operations only require sequential access within contiguous memory regions, thereby fully leveraging hardware prefetching~\cite{chen1995hardwarePrefetch} capabilities.

\begin{example}
	
	As illustrated in Figure \ref{fig:framework}, consider 10 vectors stored contiguously in memory. Even when accessing $x_1$ through $x_5$ where $x_2$ and $x_3$ are not immediately required, hardware prefetchers can still proactively load $x_4$ into cache. This behavior stems from the memory locality created by storing adjacent vectors ($x_1$ to $x_5$) in consecutive memory addresses. The consistent memory layout and predictable access patterns effectively compensate for software-based prefetching inefficiencies through hardware optimizations.
\end{example}

\begin{figure}[t!]\centering\vspace{-3ex}
	{\includegraphics[width=.8\linewidth]{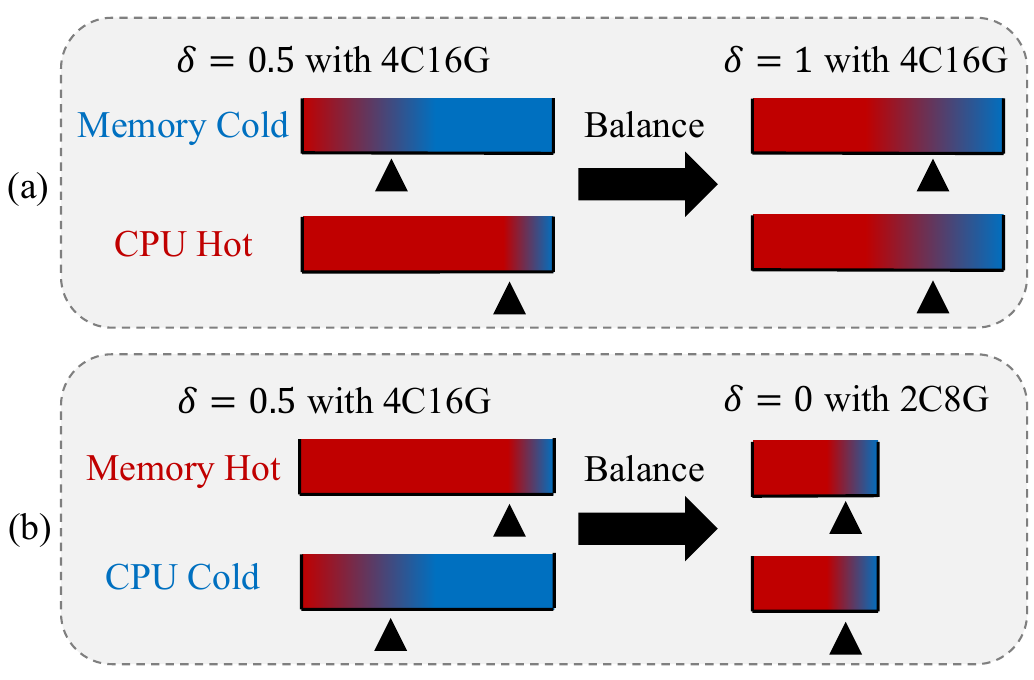}}\vspace{-2ex}
	\caption{\small Adjust Redundant Ratio $\delta$}\vspace{-3ex}
	\label{fig:PRS}
\end{figure}

\subsubsection{Balance of Computational Efficiency and Memory Utilization.}
\label{sec:balance_cpu_memory}
To address the computational-memory resource imbalance caused by fixed instance specifications (e.g., 4C16G, 2C8G) in industrial applications, we propose PRS to take advantage of hardware prefetching to reduce CPU idle time. A dynamically tunable \textit{redundancy ratio} $\delta$ controls the proportion of redundantly stored neighbor vectors in graph indexes, balancing prefetch efficiency and CPU utilization. When $\delta=1$, full redundancy maximizes hardware prefetch benefits, achieving peak CPU utilization at the cost of higher memory consumption. In contrast, $\delta=0$ eliminates redundancy to minimize memory usage but sacrifices prefetch efficiency. This flexibility enables workload-aware resource optimization as shown in Example \ref{exp:RPS}.
\begin{example}
	For high-throughput/high-recall scenarios (Fig. \ref{fig:PRS} (a)), increasing $\delta=1$ prioritizes CPU efficiency to meet demanding targets with fewer compute resources. In memory-constrained environments with moderate throughput requirements (Fig. \ref{fig:PRS} (b)), reducing $\delta=0$ alleviates memory pressure while allowing instance downsizing (e.g., 4C16G to 2C8G), which maintains service quality through controlled compute scaling and reduces infrastructure costs.
	\label{exp:RPS}
\end{example}

\section{Automatic Parameter Tuning}
\label{sec:tech2_auto_param_opt}

\subsection{Parameter Classification}
We observe that specific parameters in the retrieval process exert a significant influence~\cite{2024vdtuner} on the search performance (e.g., prefetch depth $\nu$, construction ef $ef_c$, search ef $ef_s$, and candidate pool size $m_c$). These parameters can be systematically categorized into three distinct types for discussion: Environment-Level Parameters (ELPs), Query-Level Parameters (QLPs), and Index-Level Parameters (ILPs).

\noindent\textbf{ELPs} primarily affect the efficiency (QPS) of the retrieval process and are directly related to the retrieval environment (e.g., prefetch stride $\omega$, prefetch depth $\nu$). Most of these parameters are associated with the execution speed of system instructions rather than the operational flow of the algorithm. For instance, prefetch-related parameters mainly influence the timing of asynchronous operations.

\noindent\textbf{QLPs} inherently influence both retrieval efficiency (QPS) and effectiveness (Recall) simultaneously. These parameters operate on prebuilt indexes and can be dynamically configured during the retrieval phase (e.g., search parameter $ef_s$, selective reranking strategies). In particular, efficiency and effectiveness exhibit an inherent trade-off: For a static index configuration, achieving higher QPS inevitably reduces Recall performance.

\noindent\textbf{ILPs} define an index's core structure and performance. Tuning the construction-related parameters (e.g., $m_c$, $\alpha_c$) requires rebuilding indexes. It is a process far more costly than adjusting QLPs or EIPs. Crucially, ILPs impact both efficiency and effectiveness simultaneously.

\noindent\textit{\underline{Hardness.}}
The complexity of tuning these three parameter categories shows a progressive increase. ELPs focus solely on algorithmic efficiency, resulting in a straightforward single-objective optimization problem. In contrast, QLPs require balancing both efficiency and effectiveness criteria, thereby forming a multi-objective optimization challenge. The most demanding category, ILPs, necessitates substantial index construction time in addition to the aforementioned factors, leading to exponentially higher tuning expenditures. In subsequent sections, we present customized optimization strategies for each parameter category.

\subsection{Search-based Automatic ELPs Tuner}
\label{sec:elp_tuner}

The proposed algorithm utilizes multiple ELPs that exhibit substantial variance in optimal configurations in different testing environments and methodologies, as demonstrated by our comprehensive experimental analysis (see \S \ref{sec:experimental}). As elaborated in \S \ref{sec:jumping_prefetch}, the stride prefetch mechanism operates through two crucial parameters: the prefetch stride $\omega$ and the prefetch depth $\nu$. The parameter $\omega$ governs the prefetch timing based on the dynamic relationship between the CPU computational speed and the prefetch latency within specific deployment environments. In particular, smaller $\omega$ values are required to initiate earlier prefetching when faced with faster computation speeds or slower prefetch latencies. Meanwhile, the $\nu$ parameter is primarily determined by the CPU's native prefetch patterns combined with vector dimensionality characteristics.

These environment-sensitive parameters exclusively influence algorithmic efficiency metrics (e.g., QPS) while maintaining consistent effectiveness outcomes (e.g., Recall), thereby enabling independent optimization distinct from core algorithmic logic. The \vsag tuning framework implements an optimized three-step procedure:
\begin{enumerate}[itemsep=0pt]
	\item Conduct an exhaustive grid search for all combinations of environment-dependent parameter.
	\item Evaluate performance metrics using sampled base vectors.
	\item Select parameter configurations that maximize retrieval speed while maintaining operational stability.
\end{enumerate}
\vspace{-2ex}

\subsection{Fine-Grained Automatic QLPs Tuner }
\label{sec:qlp_tuner}
\textbf{Observation.} There are significant differences in the parameters required for different queries to achieve the target recall rate, with a highly skewed distribution. Specifically, 99\% of queries can achieve the target recall rate with small query parameters, while 1\% of queries require much larger parameters. Experimental results show that assigning personalized optimal retrieval parameters to each query can improve retrieval performance by 3–5 times.

To address the observation of QLPs, we propose a \textbf{Decision Model} for query difficulty classification, enabling personalized parameter tuning while maintaining computational efficiency. We introduce a learning-based adaptive parameter selection framework through a GBDT classifier~\cite{2016xgboost} with early termination capabilities, demonstrating superior performance in fixed-recall scenarios. The model architecture addresses two critical constraints: discriminative power and computational efficiency. Our feature engineering process yields the following optimal feature set:

\begin{itemize}[leftmargin=5mm,itemsep=1pt]
	\item Cardinality of scanned points
	\item Distribution of distances among current top-5 candidates
	\item Temporal distance progression in recent top-5 results
	\item Relative distance differentials between top-K candidates and optimal solution
\end{itemize}

The candidate size is initially set to a relatively large value (e.g., $ef_s$ = 300). During the search process, after traversing certain hops in the graph, we employ the decision tree along with the current features to determine whether the candidate size should be reduced, which effectively prevents overly simple queries from undergoing extensive and unnecessary search process.

\subsection{Labeling-based Automatic ILPs Tuner }
\label{sec:ilp_tuner}
\textbf{Impact of ILPs.} The retrieval efficiency and effectiveness of graph-based indexes are fundamentally determined by their structural properties. Modifications to ILPs (e.g., maximum degree $m_c$~\cite{jayaram2019diskann, malkov2018hnsw}, pruning rate $\alpha_c$~\cite{jayaram2019diskann}) induce concurrent alterations in both graph topology and search dynamics~\cite{yang2025revisitingindexconstructionproximity}, creating non-linear interactions between retrieval speed and accuracy. 
Rebuilding the index while tuning ILPs is computationally expensive(typically taking 2-3 hours per parameter configuration for million-scale datasets).

\subsubsection{Compression-Based Construction Framework}
We empirically observe that indexes constructed with varying ILPs exhibit substantial edge overlap. As formalized in Theorem~\ref{theo:max_degree}, when maintaining a fixed pruning rate $\alpha$, a graph built with lower maximum degree $m_c$ constitutes a strict subgraph of its higher-degree counterpart. 

\begin{theorem} (Maximum Degree-Induced Subgraph Hierarchy)
	Let $G^a$ and $G^b$ denote indexes constructed with maximum degrees $m_c = a$ and $m_c = b$ respectively, under identical initialization conditions and fixed $\alpha$. Given equivalent greedy search outcomes $\mathrm{ANN}_k(x)$ for all points $x$ under both configurations, then $a < b \implies G^a \subset G^b$ where $\forall (x_i, x_j) \in E(G^a), (x_i, x_j) \in E(G^b)$.
	\label{theo:max_degree}
\end{theorem}
\begin{proof}
	Please refer to Appendix B.
\end{proof}

The lifecycle of the index in \vsag is divided into \textbf{three distinct phases: building, tuning, and searching}. In contrast, traditional indexing frameworks typically consist of only two phases: building and searching. They select edges to construct the index structure using a set of fixed index-level parameters. \textit{The reason behind our design is that we believe it is insufficient to rely solely on the raw data for selecting edges during the construction phase.} These parameters may need to be dynamically adjusted based on varying user requirements and query complexities.

\vsag employs an \textit{index compression strategy} that constructs the index once using relaxed ILPs while labeling edges with specific ILPs. Then, \vsag can dynamically select edges by leveraging labels, shifting the edge selection process from the index construction phase to the tuning phase. It enables adaptive edge selection without index reconstruction. Thus, it achieves equivalent search performance to maintaining multiple indexes with different parameter configurations, while incurring only the overhead of constructing a single index. The parameter labeling mechanism maintains topological flexibility while enhancing storage efficiency through a compressed index representation. Due to space limitations, please refer to Appendix B for the detailed construction algorithm of \vsag. Then, we illustrate the labeling algorithm, which assigns label to each edge.

\begin{algorithm}[t!]\small
	\DontPrintSemicolon
	\KwIn{ 
		base points $x_i$, 
		approximate nearest neighbors $\mathrm{ANN}_k(x_i)$ and related distances $T_i$ sorted by distance in ascending order, 
		pruning rates $A$ sorted in ascending order, 
		maximum degree $m_c$, 
		reverse insertion position $r$
	}
	\KwOut{ neighbors list $G_i$ and labels list $L_i$ of point $x_i$ }

	{initialize neighbors list $G_{i,r:} \leftarrow \mathrm{ANN}_k(x_i)_{r:}$} \;
	
	{initialize neighbors labels list $L_{i,r:}$ with $0$} \;
	
	{$count \leftarrow r$} \;
	
	\ForEach{$\alpha_c \in A$ \textbf{\textup{and}} $count < m_c$} {
		\ForEach{$j \in G_{i,r: }$ \textbf{\textup{and}} $count < m_c$} {
			\If{$L_{i,j} \neq 0$} {
				continue \;
			}
			
			$is\_pruned\leftarrow$ False \;
			
			\ForEach{$x_k \in G_{i,0: j}$}  {
				\If{$ 0 < L_{i,k} \le \alpha_c$ \textbf{\textup{and}} $\alpha_c \cdot \tau(x_j, x_k) \le {T_{i,j} } $} {
					$is\_pruned \leftarrow$ True \;
					{break} \;
				}
			} 
			
			\If{\textbf{\textup{not}} is\_pruned} {
				$L_{i,j} \leftarrow \alpha_c$ \;
				
				{$count \leftarrow count +1$} \;
			}
		}
	}
	
	shrink $|G_i| \le m_c$ and $|L_i| \le m_c$ by removing $x_j \ s.t.\ L_{i,j} = 0$ \;

	\Return{$G_i, L_i$} \;
	\caption{Prune-based Labeling}
	\label{algo:pruning}
\end{algorithm}

\noindent{\textbf{Prune-based Labeling Algorithm.}} The pruning strategy of \vsag is shown in Algorithm \ref{algo:pruning}. First, $\mathrm{ANN}_k(x_i)$ and their distances $T_i$ are sorted in ascending order of distance, and the pruning rates $A$ are also sorted in ascending order. The reverse insertion position $r$ indicates whether this pruning process occurs during the insertion of a reverse edge. We initialize the out-edges of $x_i$ by $ANN_k(x_i)_{r:}$ (Line 1). Each edge is then assigned a label of $0$ (Line 2). Note that when $r > 0$, it indicates that the current pruning occurs during the reverse edge addition phase, and only the labels of edges within the interval $[r:|G_i|]$ need to be updated. Otherwise, all labels should be updated. We use $count$ to record the number of neighbors that have non-zero labels (Line 3). When $\text{count} = r$, it means we have already collected all the neighbors we need. At this point, the algorithm should terminate (Lines 4-5). 

Next, each $\alpha_c$ is examined in ascending order (Line 4). For each unlabeled neighbor $x_j$ (Lines 5-7), neighbor $x_k$ with smaller distance is used to make pruning decision (Lines 8-9). The pruning decision requires satisfying two conditions (Lines 10-12): (a) The neighbor $x_k$ exists in the graph constructed with the $\alpha_c$ (i.e., $0 < L_{i, k} \leq \alpha_c$). (b) The pruning condition is satisfied (i.e., $\alpha_c \cdot \tau(x_j, x_k) \le \tau(x_i,x_j)$). Here, we accelerate the computation of $\tau(x_i, x_j)$ by using the cached result $T_{i, j}$. If no neighbor can prune $x_j$ with $\alpha_c$, it is assigned the label $L_{i, j} \leftarrow \alpha_c$ (Lines 13-15). Finally, the algorithm returns the neighbor set $G_i$ and labels set $L_i$ of $x_i$ (Lines 16-17).

Building upon the parameter analysis of $m_c$, we formally establish the subgraph inclusion property for graphs constructed with varying pruning rates $\alpha_c$ in Algorithm~\ref{algo:pruning} through Theorem \ref{theo:alpha}.

\begin{theorem} (Subgraph Inclusion Property with Varying Pruning Rate $\alpha_c$)
	Fix all ILPs except $\alpha_c$, and let $G^a$ and $G^b$ be indexes constructed by Algorithm 3 in Appendix using pruning rates $\alpha_c = a$ and $\alpha_c = b$ respectively, where $a < b$. Suppose that for every data point $x_i$, the finite-sized approximate nearest neighbor (ANN) sets $\mathrm{ANN}_k(x_i)$ retrieved during construction remain identical under both $\alpha_c$ values. Then $G^a$ forms a subgraph of $G^b$, i.e., all edges in $G^a$ satisfy $(x_i, x_j) \in E(G^b)$. 
	\label{theo:alpha}
\end{theorem}
\begin{proof}
	Please refer to Appendix C.2.
\end{proof}

This theorem reveals a monotonic relationship between pruning rates and graph connectivity. When reducing $\alpha_c$, any edge pruned under this stricter parameter setting would necessarily be eliminated under larger $\alpha_c$ values. Then, we can characterize each edge by its preservation threshold $\alpha_e$: the minimal pruning rate required to retain the edge during construction. Consequently, all graph indexes constructed with pruning rates $\alpha_c \geq \alpha_e$ will contain this edge. This threshold-based perspective permits efficient compression of multiple parameterized graph structures into a unified index, where edges are annotated with their respective $\alpha_e$ values.

\begin{figure}[t!]\centering
	{\includegraphics[width=.8\linewidth]{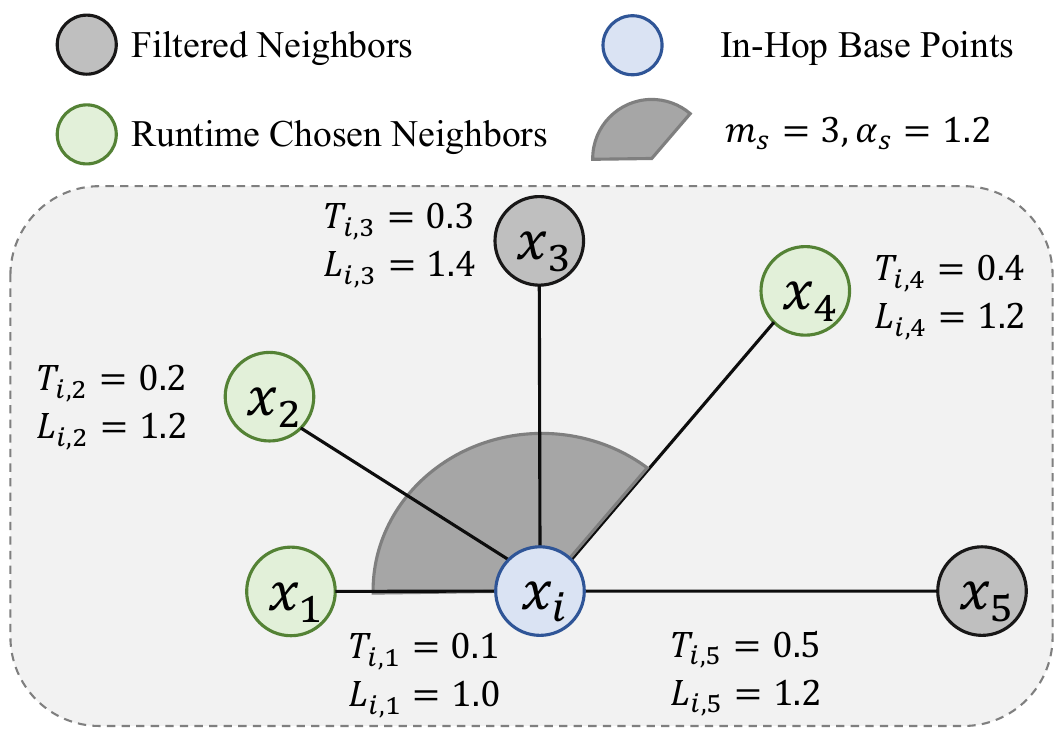}}\vspace{-2ex}
	\caption{\small Runtime Adjust ILPs ($m_c$ and $\alpha_c$) by Tuning $m_s$ and $\alpha_s$}\vspace{-4ex}
	\label{fig:example_auto_ILP}
\end{figure}

\begin{example}
	As shown in Figure~\ref{fig:example_auto_ILP}, we illustrate the state of the labeled graph generated by Algorithm~\ref{algo:pruning} and the runtime edge selection process. Suppose that during the greedy search in the graph using Algorithm~\ref{algo:greedy}, we need to explore the in-hop base point $x_i$ (Line 5 of Algorithm~\ref{algo:greedy}). The node $x_i$ has 5 neighbors sorted by distance $T_{i,j}$ in ascending order (i.e., $x_1, \dots, x_5$). The distances are $T_{i,1}, \dots, T_{i,5}$, and the corresponding labels are $L_{i,1}, \dots, L_{i,5}$.
	
	Given a relaxed ILP with $m_s = 3$ and $\alpha_s = 1.2$, we visit the neighbors of $x_i$ in ascending order of distance. We then filter out neighbors that do not satisfy the pruning condition (Line 8 of Algorithm~\ref{algo:greedy}):
	\begin{itemize}
		\item $x_2$ is filtered out because $L_{i,2} = 1.4 > \alpha_s$.
		\item $x_5$ is filtered out because we have found $m_s = 3$ valid neighbors.
	\end{itemize}
	
	Thus, in this search hop, we will visit $x_1$, $x_3$, and $x_4$. As proven in Theorem~\ref{theo:alpha} and Theorem~\ref{theo:max_degree}, the search process is equivalent to searching in a graph constructed with $m_c = 3$ and $\alpha_c = 1.2$. In other words, we can dynamically adjust the ILPs  (i.e., $m_c, \alpha_c$) by tuning the relaxed QLPs  (i.e., $m_s, \alpha_s$). This approach saves significant costs associated with rebuilding the graph.
	
	Please refer to Appendix D for  details of tuning ILPs after \vsag is constructed with labels.
\end{example}

\section{Distance Computation Acceleration}
\label{sec:tech3_distance_opt}

Recent studies~\cite{Fast-Index-arxiv-2024,ADSampling:journals/sigmod/GaoL23,yang2024ddc} illustrate that the exact distance computation takes the majority of the time cost of graph-based ANNS. Approximate distance techniques, such as scalar quantization, can accelerate this process at the cost of reduced search accuracy. \vsag adopt a two-stage approach that first performs an approximate distance search followed by exact distance re-ranking. \S\ref{sec:distance_computation_analysis} analyzes the distance computation scheme, with subsequent sections detailing optimization strategies for \vsag component.

\subsection{Distance Computation Cost Analysis}
\label{sec:distance_computation_analysis}
\vsag employs low-precision vectors during graph traversal operations while reserving precise distance computations exclusively for final result reranking. The dual-precision architecture effectively minimizes distance computation operations (DCO)~\cite{yang2024ddc} overhead while preserving search accuracy through precision-aware hierarchical processing.

If we only consider the cost incurred by distance computation, the total distance computation cost can be expressed as follows:

\vspace{-2ex}
\begin{equation}
	\label{eq:distance_cost_model}
	cost = cost_{lp} + cost_{hp} = n_{lp} \cdot t_{lp} + n_{hp} \cdot t_{hp}\notag
\end{equation}
\vspace{-2ex}

Here, distance computation cost $cost$ consists of two components: the computation cost for low-precision vectors $cost_{lp}$ and the computation cost for high-precision vectors $cost_{hp}$. Each component is determined by the number of distance computations ($n_{lp}$ or $n_{hp}$) and the cost of a single distance computation ($t_{lp}$ or $t_{hp}$).

The optimization of $n_{lp}$ is closely related to the specific algorithm workflow, while $t_{hp}$ is primarily determined by the computational cost of \texttt{FLOAT32} vector operations - both of which remain relatively constant. Consequently, the \vsag framework focuses primarily on optimizing the parameters $t_{lp}$ and $n_{hp}$.

\vsag optimizes the overall cost in three ways:  
\begin{itemize}[leftmargin=*, itemsep=2pt]
	
	\item The combination of quantization techniques, hardware instruction set SIMD, and memory-efficient storage (\S \ref{sec:minimize_computation_overhead}) achieves exponential reduction in low-precision distance computation ($t_{lp}$).  
	\item Enhanced quantization precision by parameter optimization (Appendix 
	G) mitigates candidate inflation from precision loss, achieving sublinear growth in required low-precision computations ($n_{hp}$).  
	\item Selective re-ranking with dynamic thresholding (\S 
	\ref{sec:selective_rerank}) establishes an accuracy-efficiency equilibrium, restricting high-precision validation ($n_{hp}$) to a logarithmically scaled candidate subset.
	
\end{itemize}

\subsection{Minimizing Low-Precision Computation Overhead}
\label{sec:minimize_computation_overhead}

\textbf{SIMD and Quantization Methods.} Modern CPUs employ \texttt{SIMD} instruction sets (\texttt{SSE/AVX/AVX512}) to accelerate distance computations through vectorized operations. These instructions process 128-bit, 256-bit, or 512-bit data chunks in parallel, with vector compression techniques enabling simultaneous processing of multiple vectors. For example, \texttt{AVX512} can compute one distance for 16-dimensional \texttt{FLOAT32} vectors per instruction, but when compressing vectors to 128 bits, it achieves 4x acceleration by processing four vector pairs concurrently. Product Quantization (PQ)~\cite{jegou2010pq_ivfadc} enables high compression ratios for batch processing through \texttt{SIMD}-loaded lookup tables. While PQ-Fast Scan~\cite{PQ-fast-scan-VLDB-2015} excels in partition-based searches through block-wise computation, its effectiveness diminishes in graph-based searches due to random vector storage patterns and inability to filter visited nodes, resulting in wasted \texttt{SIMD} bandwidth. In contrast, Scalar Quantization (SQ)~\cite{2012sq} proves more suitable for graph algorithms by directly compressing vector dimensions (e.g., \texttt{FLOAT32$\to$ INT8/INT4}) without requiring lookup tables. As demonstrated in \vsag, SQ achieves the optimal balance between compression ratio and precision preservation while fully utilizing \texttt{SIMD} acceleration capabilities, making it particularly effective for memory-bound graph traversals.

\noindent\textbf{Distance Decomposition.}  
\vsag optimizes Euclidean distance computations by decoupling static and dynamic components. The system precomputes and caches invariant vector norms during database indexing, then combines them with real-time dot product computations during queries. This decomposition reduces operational complexity while preserving mathematical equivalence, as shown by the reformulated Euclidean distance:\vspace{-3ex}

$$
\|x_b - x_q\|^2 = \|x_b\|^2 + \|x_q\|^2 - 2x_b \cdot x_q,
$$\vspace{-3ex}

The computational optimization strategy can be summarized as follows: Only the Inner Product term $x_b \cdot x_q$ requires real-time computation during search operations, while the squared query norm $\|x_q\|^2 $ can be pre-computed offline before initiating the search process. By storing just one additional \texttt{FLOAT32} value per database vector $x_b$ (specifically the precomputed $||x_b||^2$), we can effectively transform the computationally expensive Euclidean distance computation into an equivalent inner product operation. This space-time tradeoff reduces the subtraction CPU instruction in distance computation, which saves one CPU clock cycle.

\begin{table}[t!]
	
	\small
	\begin{center}
		{\small \vspace{-4ex}
			\caption{Dataset Statistics.}\vspace{-3ex} \label{tab:datasets}
			\begin{tabular}{l|l|l|l|l}\toprule
				{\bf Dataset\quad} & {\bf Dim } & {\bf  \#Base\quad} & {\bf  \#Query} & {\bf  Type\quad} \\  \midrule
				GIST1M & 960 & 1,000,000 & 1,000 & Image \\
				
				SIFT1M & 128 & 1,000,000 & 10,000 & Image \\
				
				TINY & 384 & 5,000,000 & 1,000 & Image \\
				
				GLOVE-100 & 100 & 1,183,514 & 10,000 & Text \\
				
				WORD2VEC & 300 & 1,000,000 & 1,000 & Text \\
				
				OPENAI & 1536 & 999,000 & 1,000 & Text \\
				
				ANT-INTERNET & 768 & 9,991,307 & 1,000 & Text \\
				
				MSMARCO & 1024 & 113,519,750 & 1,000 & Text\\
				\bottomrule
			\end{tabular}
		}\vspace{-4ex}
	\end{center}
\end{table}

\subsection{Selective Re-rank}
\label{sec:selective_rerank}
Quantization methods can significantly enhance retrieval efficiency, but quantization errors may lead to substantial recall rate degradation. While re-ranking with full-precision vectors can mitigate this performance loss. However, applying exhaustive re-ranking to all candidates is inefficient. The \vsag framework addresses this challenge through selective re-ranking, effectively compensating for approximation errors in distance computation without compromising system performance. A straightforward approach is to select only the candidates with small low-precision distances for re-ranking. The optimal number of candidates requiring re-ranking varies significantly depending on query characteristics, quantization error distribution, and search requirement $k$. To address this dynamic requirement,  \vsag implements DDC~\cite{yang2024ddc} scheme that can automatically adapt re-ranking scope based on error-distance correlation analysis.

\section{Experimental Study}
\label{sec:experimental}

\subsection{Experimental Setting}
\stitle{Datasets.} 
Table~\ref{tab:datasets} presents the datasets used in our experiments, and they are widely adopted in existing works~\cite{ann-benchmakrs} and benchmarks~\cite{ExRaBitQ-arxiv-2024}.
For each dataset, we report the vector dimensions (Dim), the number of base vectors (\#Base), the number of query vectors (\#Query), and the dataset type (Type).
All vectors are stored in \texttt{float32} format.

\begin{figure*}[t!]\centering\vspace{-4ex}
	\includegraphics[width=0.7\textwidth]{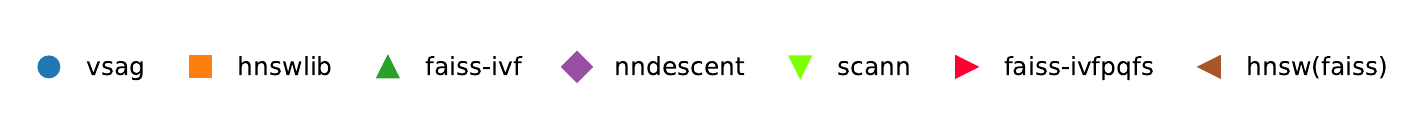}\vspace{-5ex}\\
	\subfigure[][{\scriptsize GIST1M \textit{Top-10}}]{
		\scalebox{0.13}[0.13]{\includegraphics{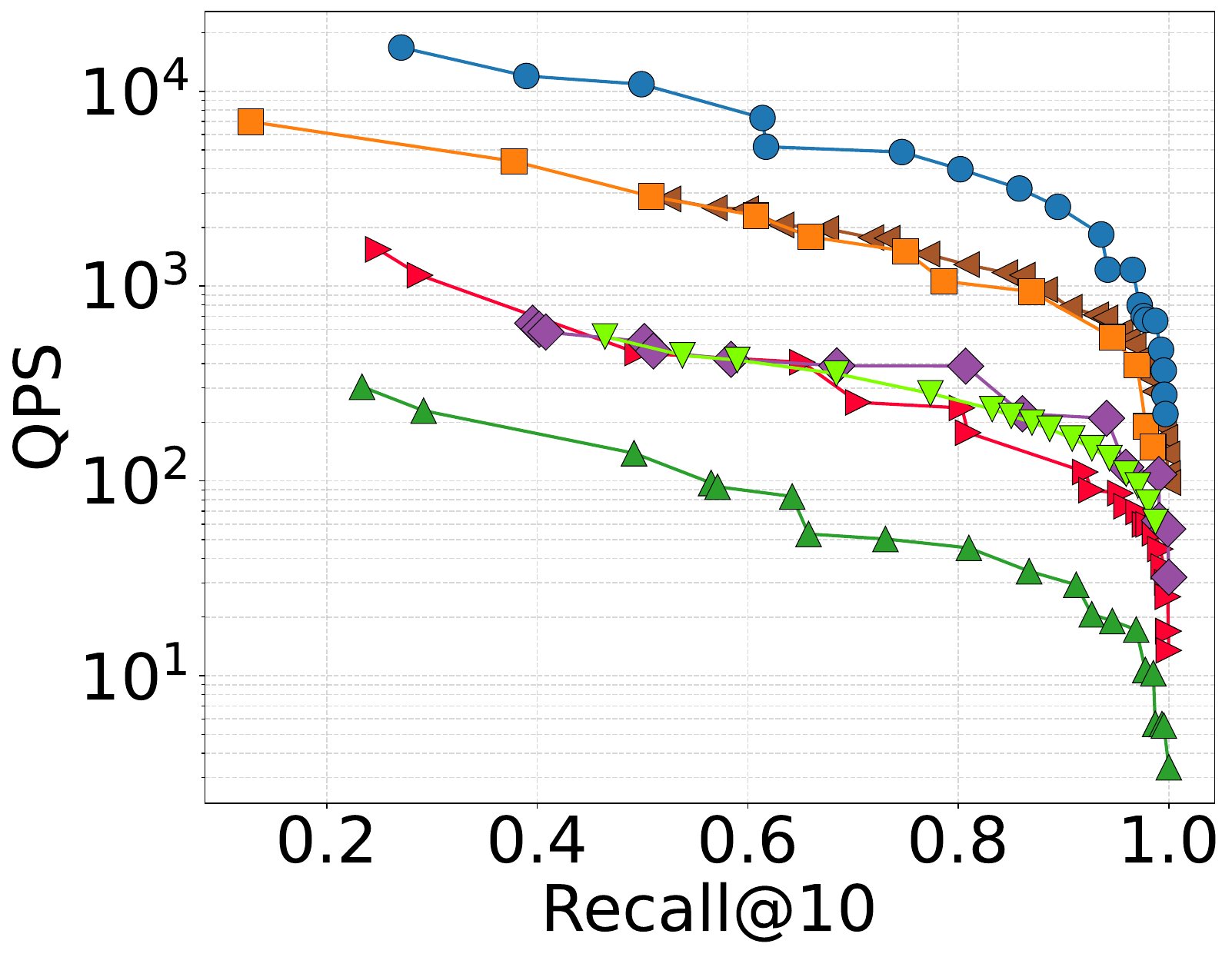}}
		\label{subfig:gist_top_10}}
	\subfigure[][{\scriptsize GIST1M \textit{Top-100}}]{
		\scalebox{0.13}[0.13]{\includegraphics{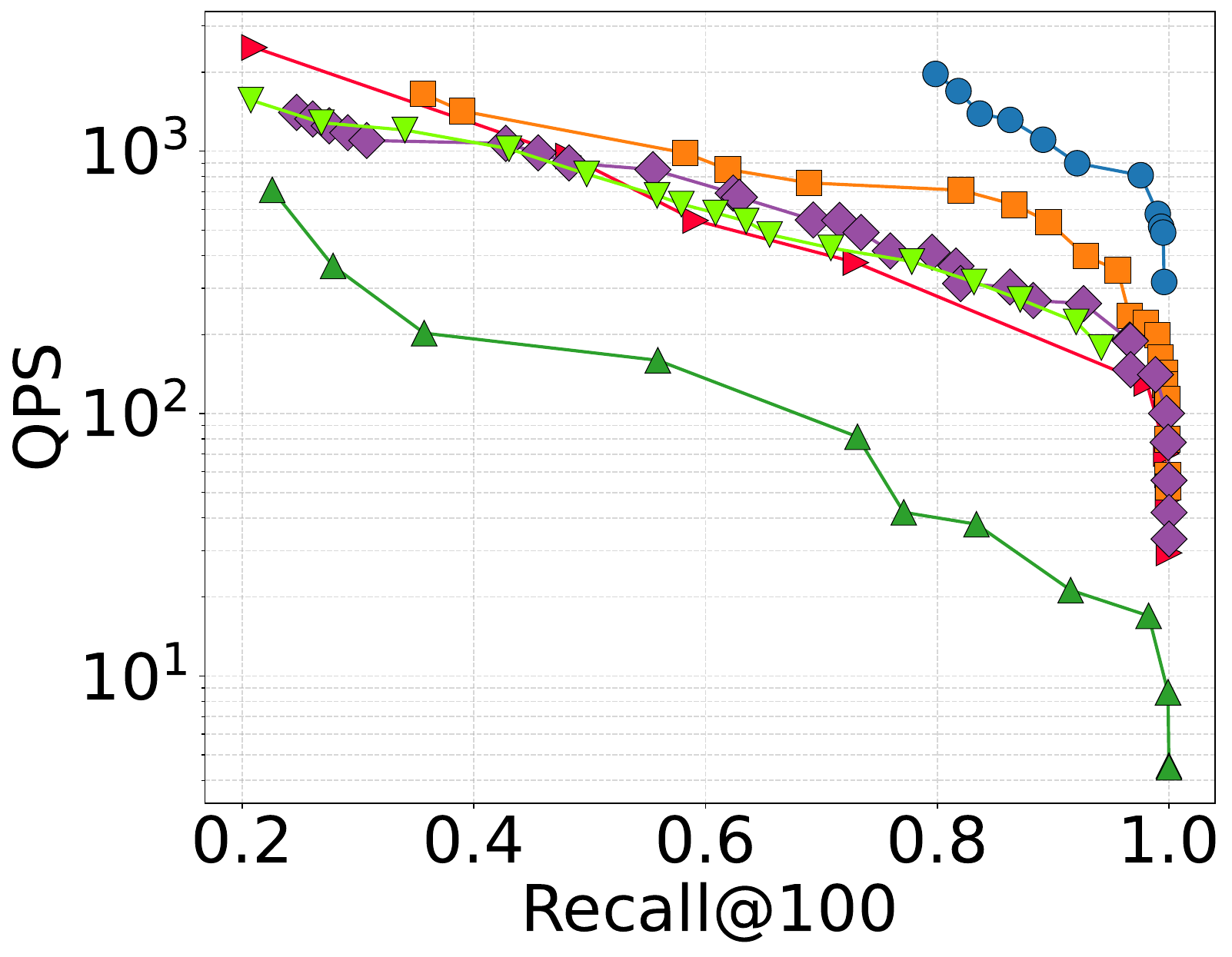}}
		\label{subfig:gist_top_100}}
	\subfigure[][{\scriptsize SIFT1M \textit{Top-10}}]{
		\scalebox{0.13}[0.13]{\includegraphics{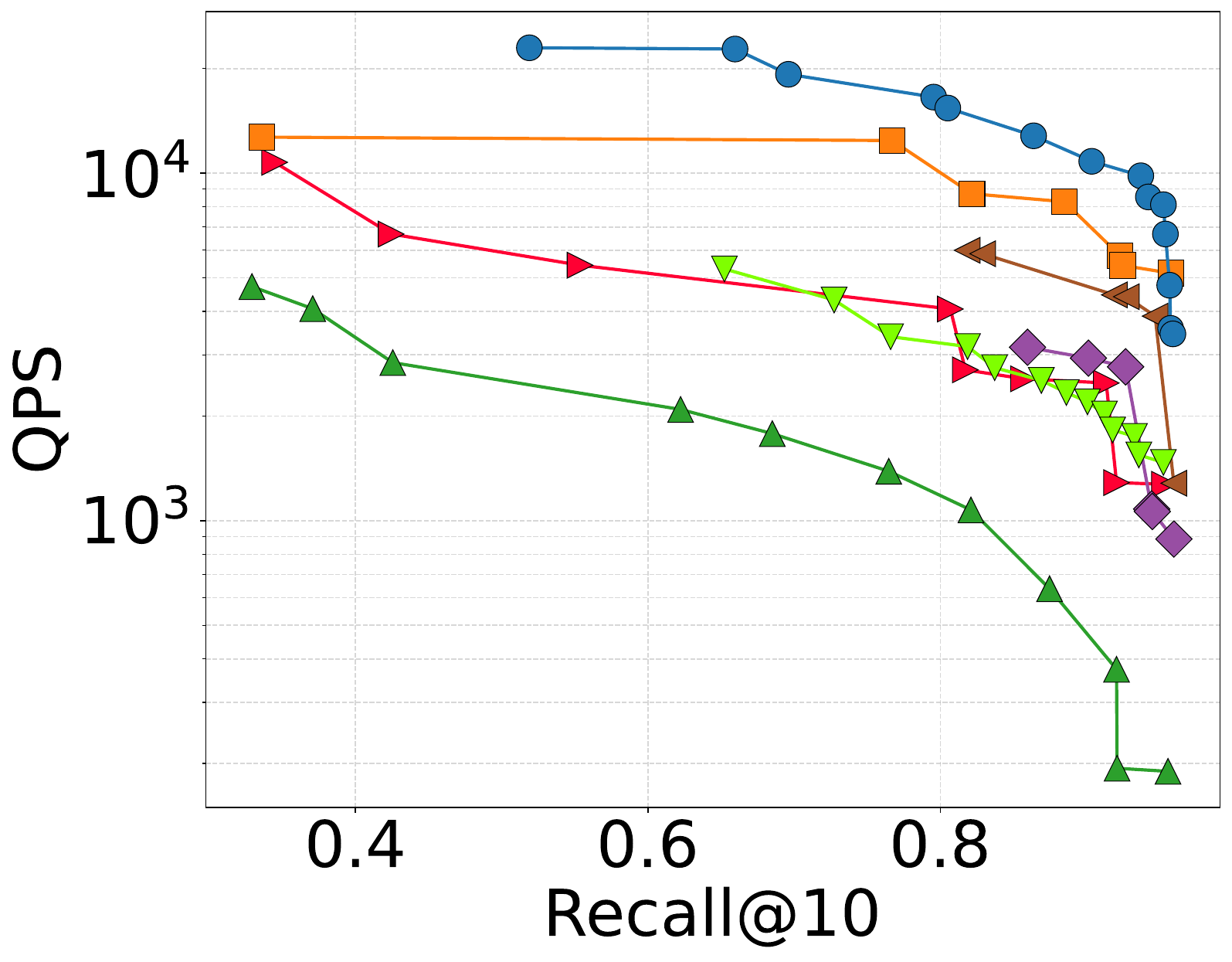}}
		\label{subfig:sift_top_10}}
	\subfigure[][{\scriptsize SIFT1M \textit{Top-100}}]{
		\scalebox{0.13}[0.13]{\includegraphics{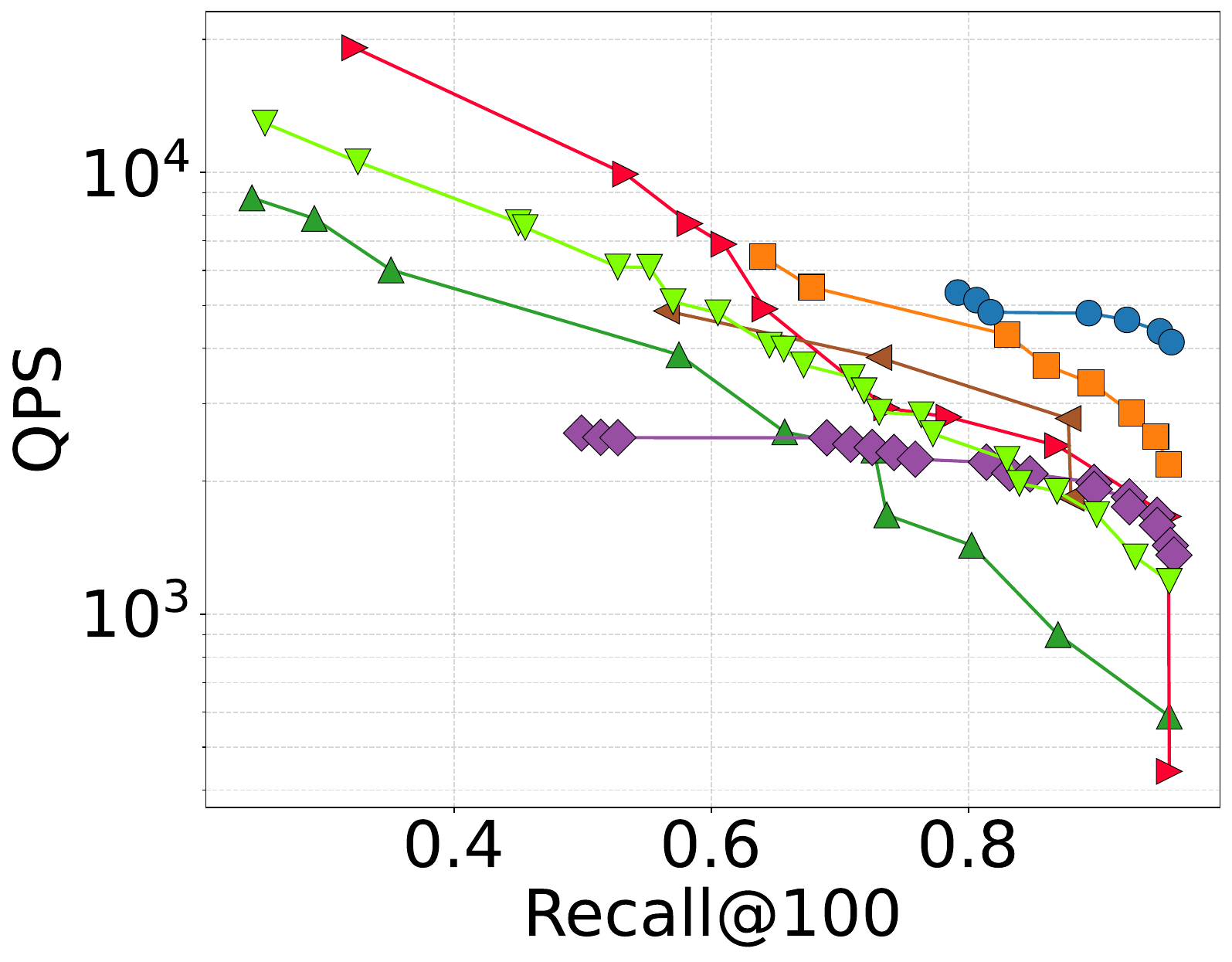}}
		\label{subfig:sift_top_100}}\\\vspace{-2ex}
	\subfigure[][{\scriptsize GLOVE-100 \textit{Top-10}}]{
		\scalebox{0.13}[0.13]{\includegraphics{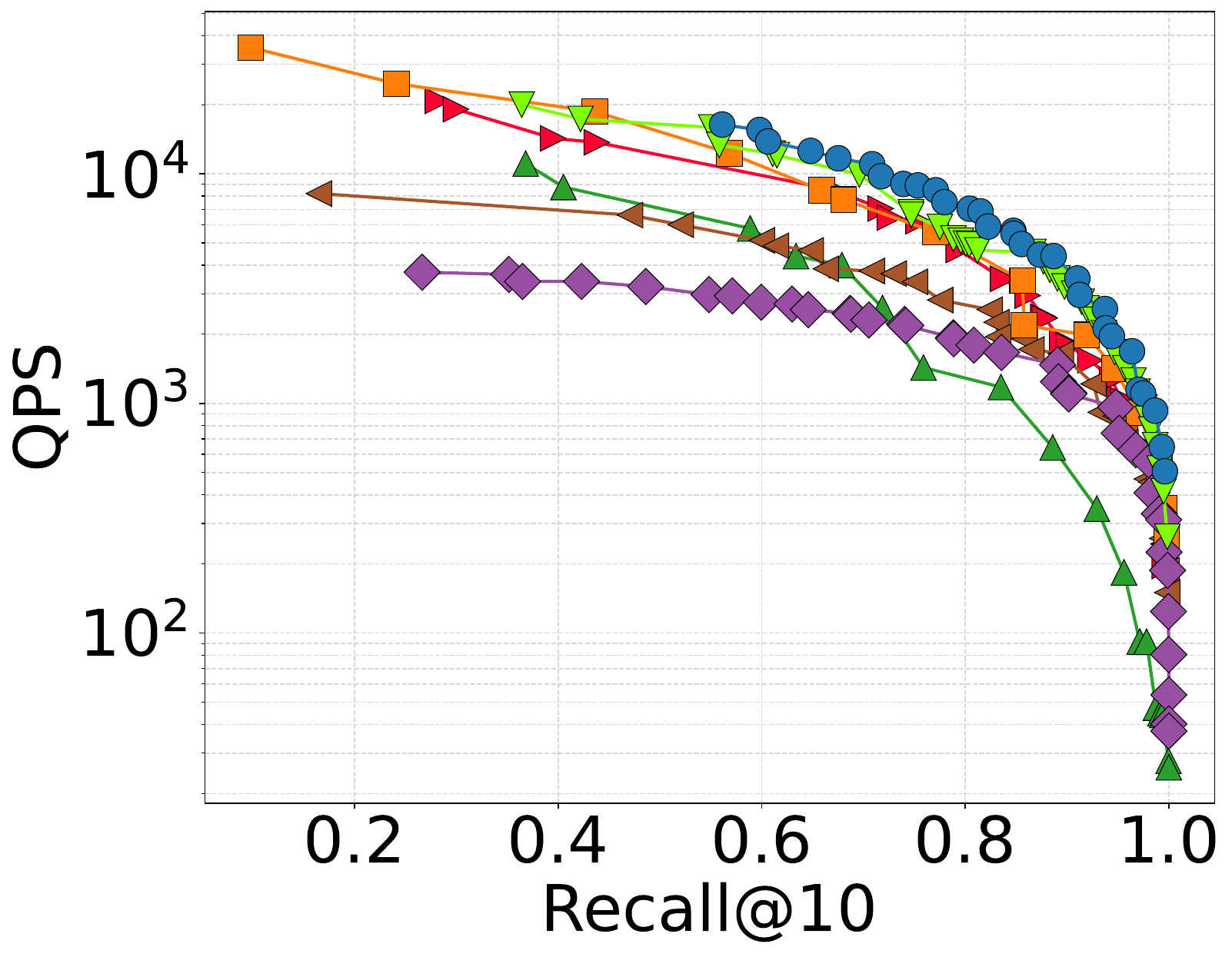}}
		\label{subfig:glove_top_10}}
	\subfigure[][{\scriptsize GLOVE-100 \textit{Top-100}}]{
		\scalebox{0.13}[0.13]{\includegraphics{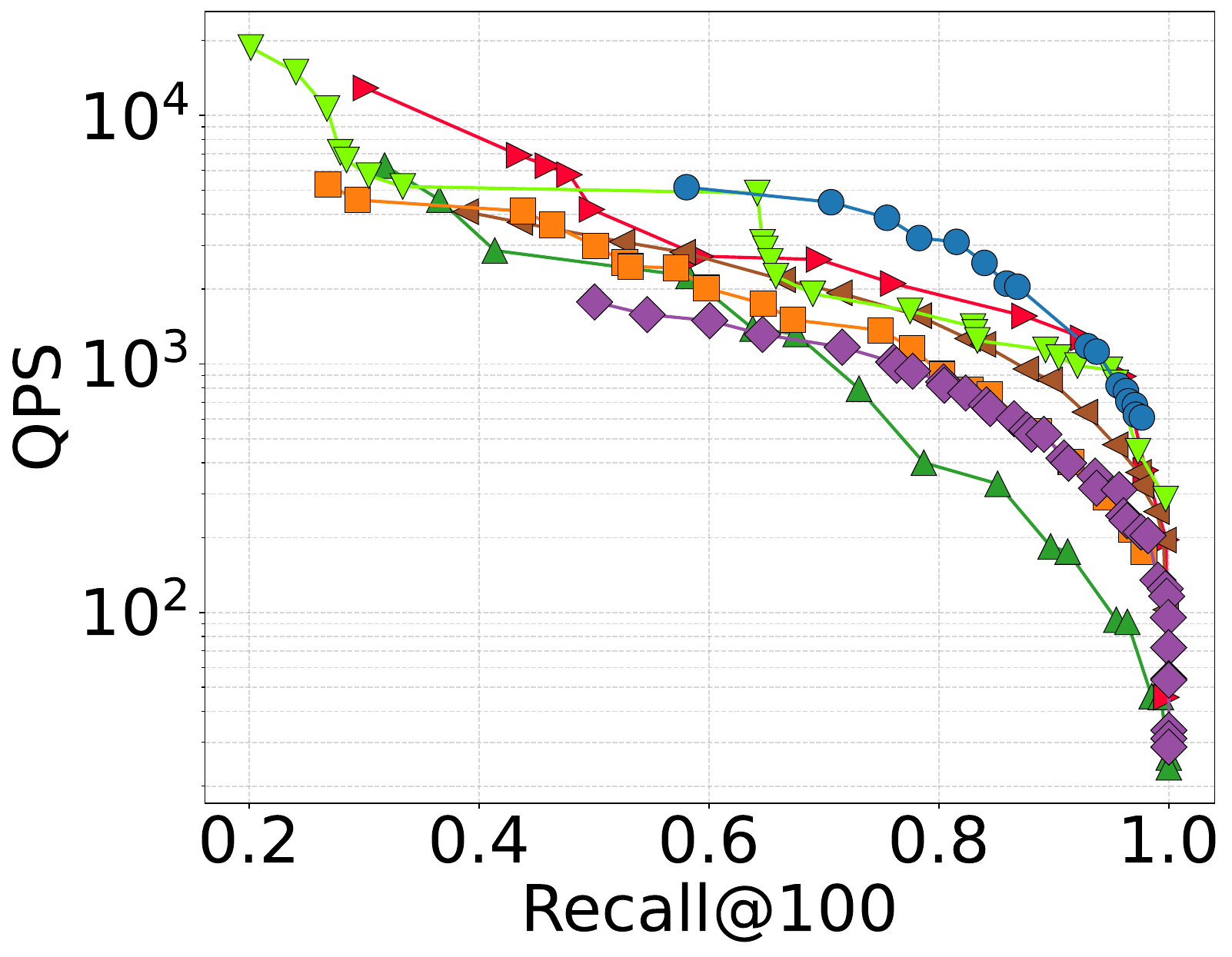}}
		\label{subfig:glove_top_10}}
	\subfigure[][{\scriptsize TINY \textit{Top-10}}]{
		\scalebox{0.13}[0.13]{\includegraphics{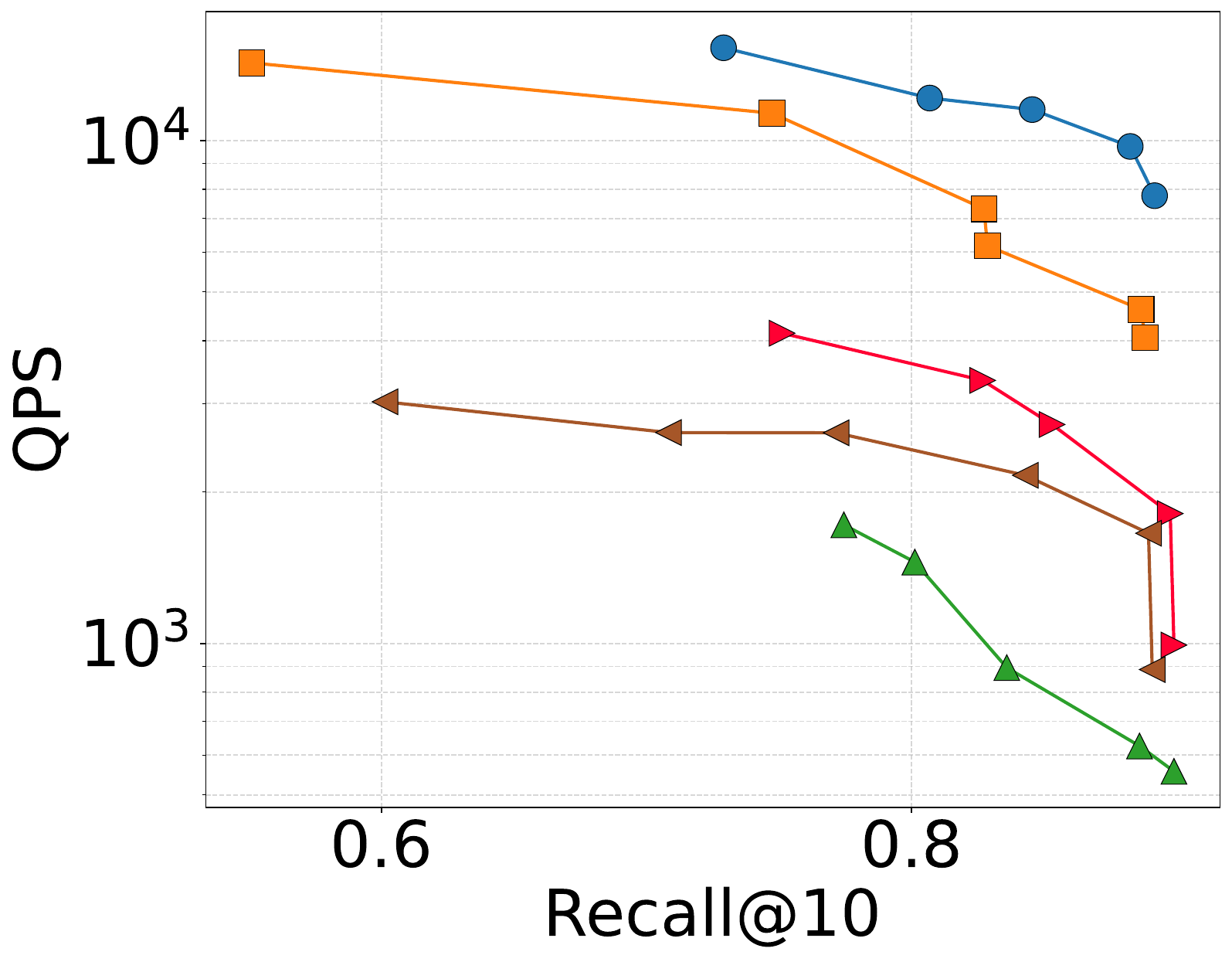}}
		\label{subfig:tiny_top_10}}
	\subfigure[][{\scriptsize TINY \textit{Top-100}}]{
		\scalebox{0.13}[0.13]{\includegraphics{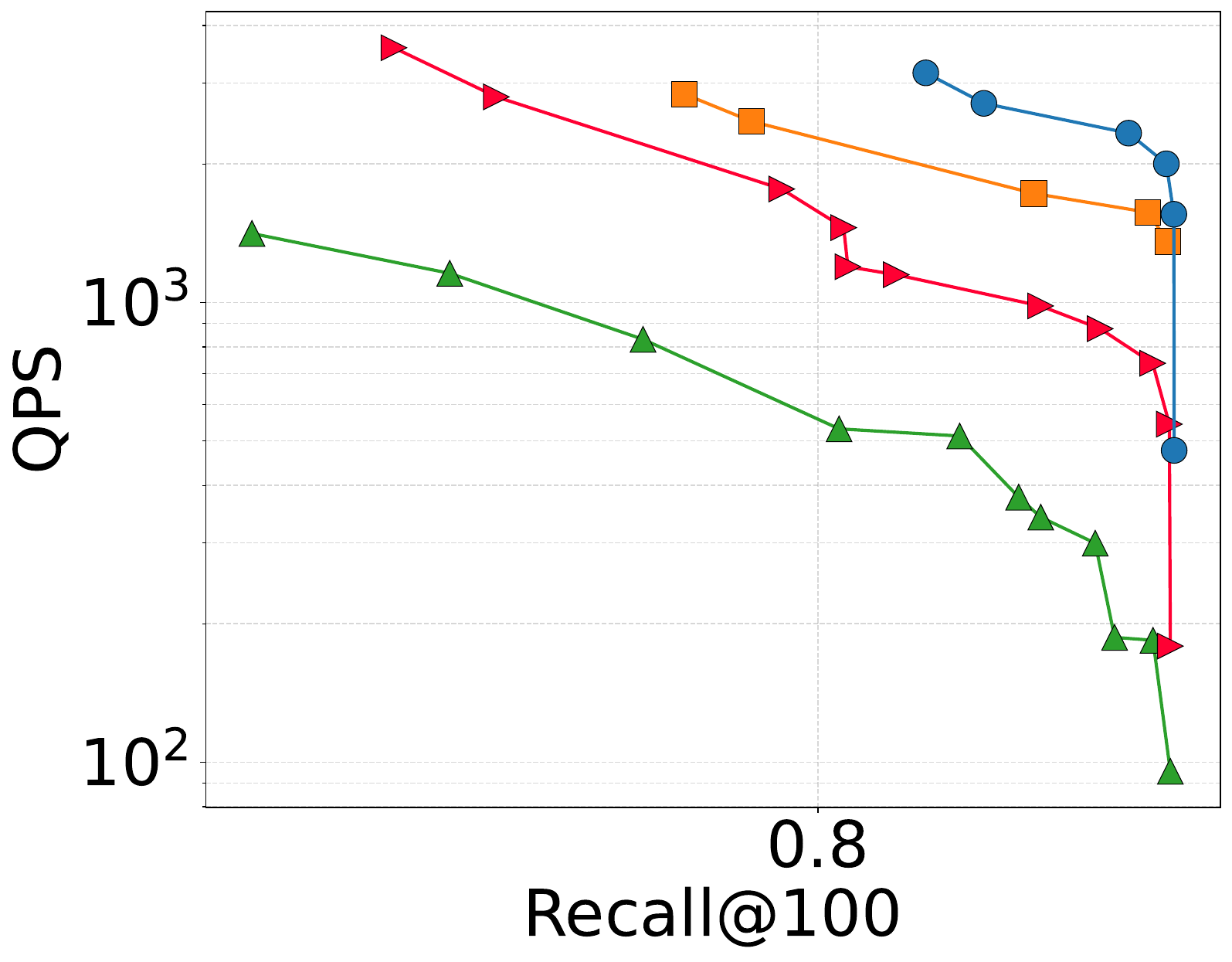}}
		\label{subfig:tiny_top_100}}\\\vspace{-2ex}
	\subfigure[][{\scriptsize WORD2VEC \textit{Top-10}}]{
		\scalebox{0.13}[0.13]{\includegraphics{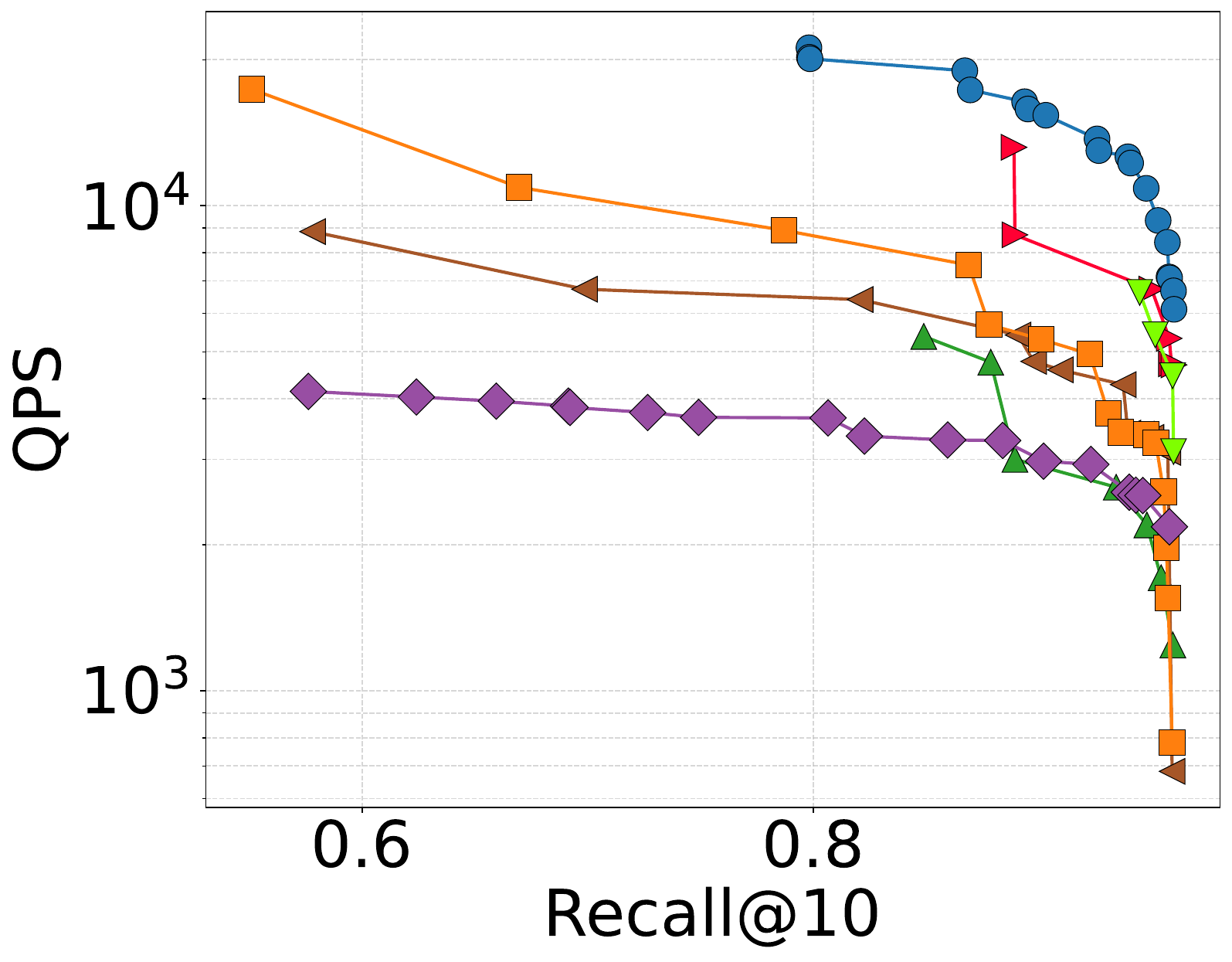}}
		\label{subfig:word2vec_top_10}}
	\subfigure[][{\scriptsize WORD2VEC \textit{Top-100}}]{
		\scalebox{0.13}[0.13]{\includegraphics{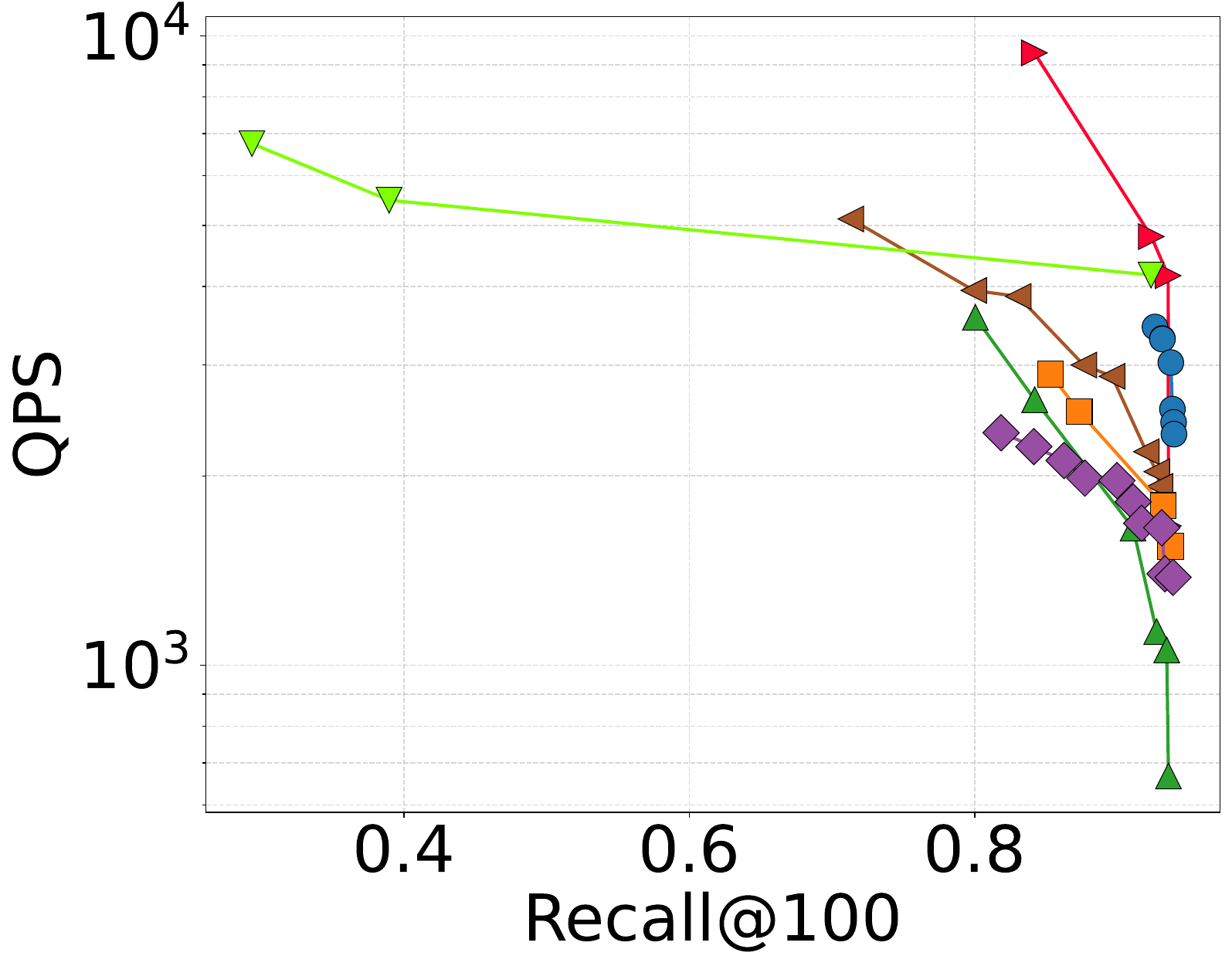}}
		\label{subfig:word2vec_top_100}}
	\subfigure[][{\scriptsize OPENAI  \textit{Top-10}}]{
		\scalebox{0.13}[0.13]{\includegraphics{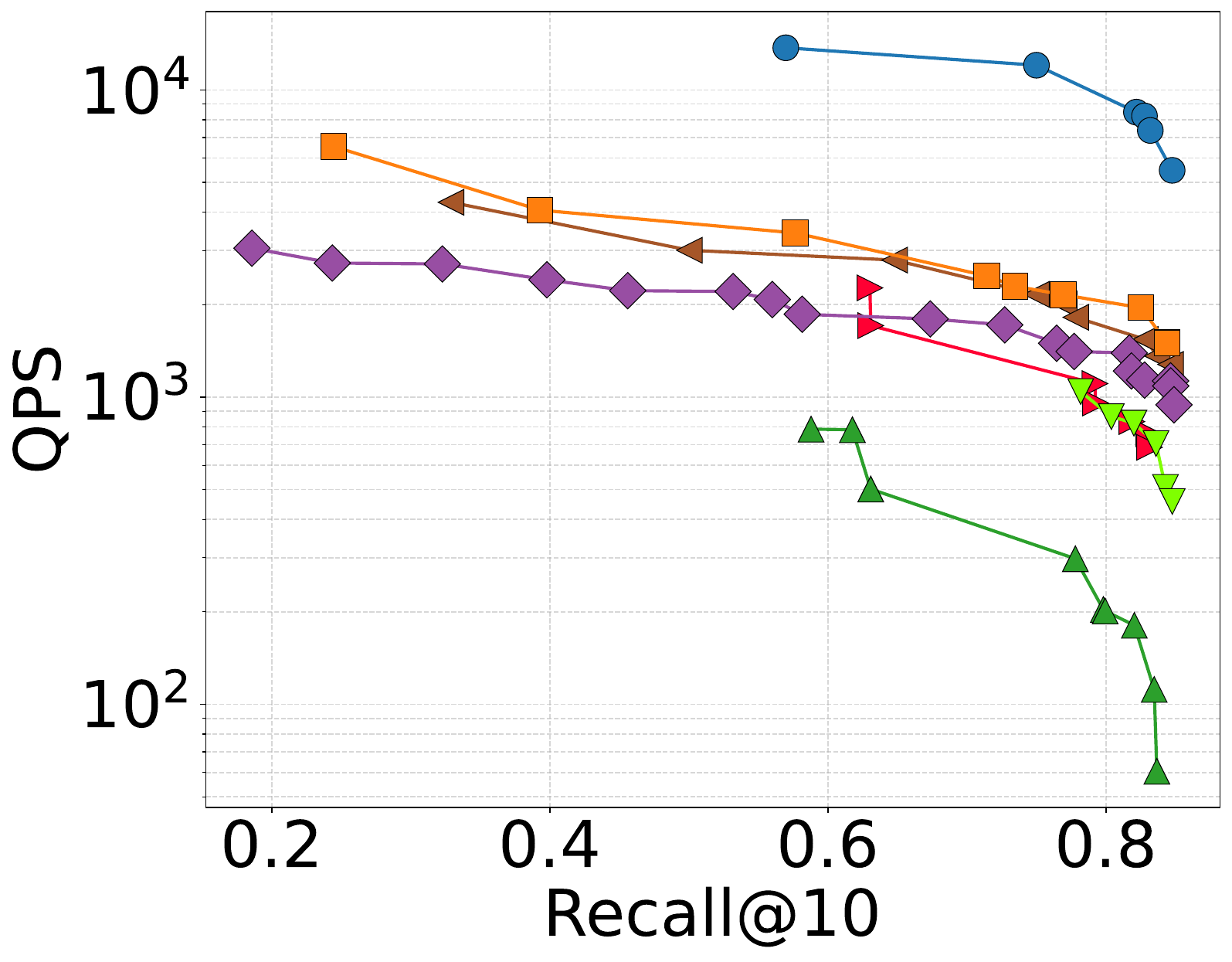}}
		\label{subfig:openai_top_10}}
	\subfigure[][{\scriptsize OPENAI  \textit{Top-100}}]{
		\scalebox{0.13}[0.13]{\includegraphics{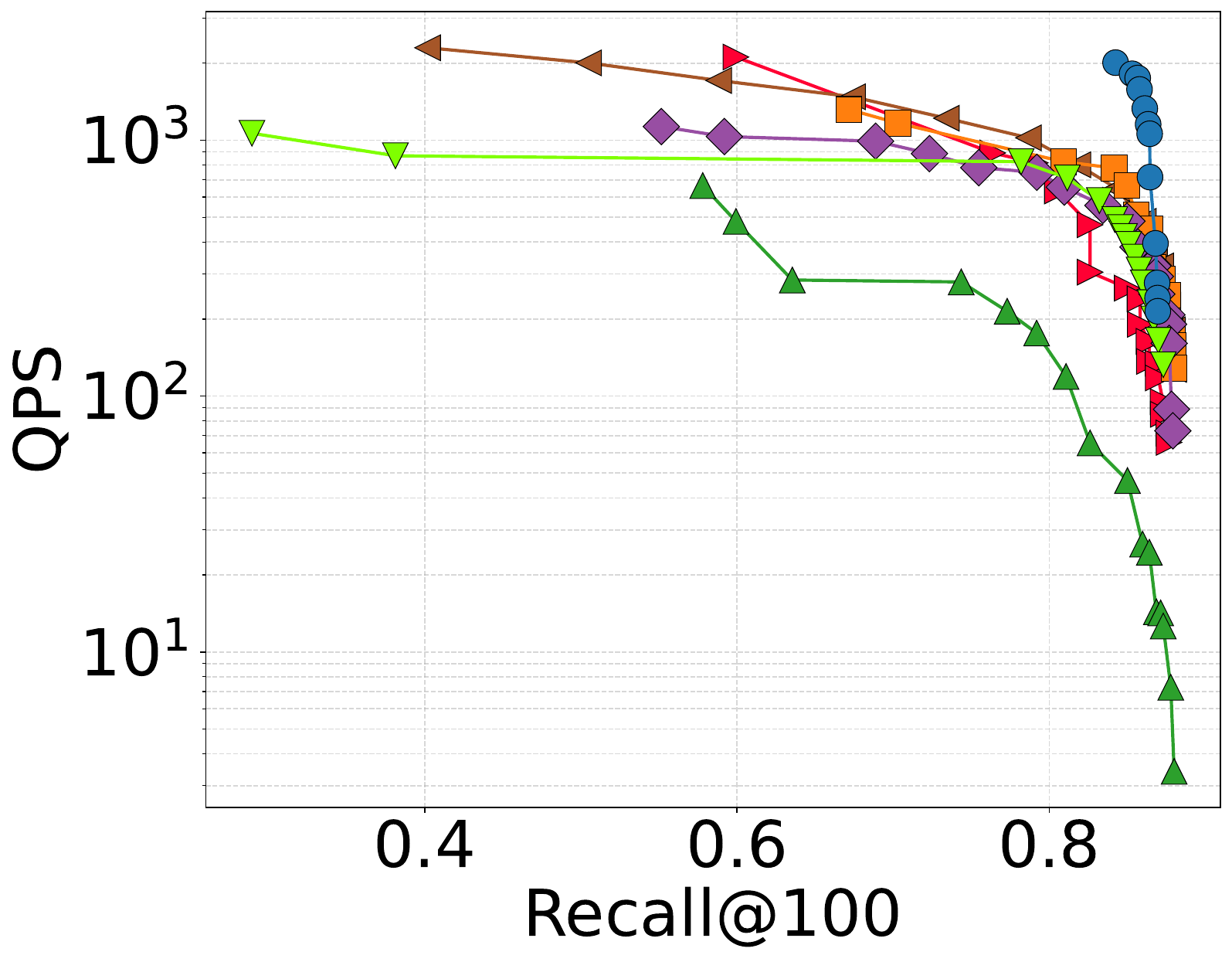}}
		\label{subfig:openai_top_100}}\vspace{-4ex}
	\caption{\small Overall Performance.}\vspace{-3ex}
	\label{fig:overall_performance}
\end{figure*}

\stitle{Algorithms.}
We compare \vsag with three graph-based methods (\textit{hnswlib,~hnsw(faiss),~ nndescent}) and three partition-based methods (\textit{faiss-ivf,~faiss-ivfpqfs,~scann}).
All methods are widely adopted in industry, and Faiss is the most popular vector search library.
\begin{itemize}[leftmargin=*, itemsep=1pt]
	\item \textit{hnswlib}~\cite{hnswlib}: the most popular graph-based index.
	
	\item \textit{hnsw(faiss)}~\cite{faiss:johnson2019billion}: the HNSW implementation in Faiss.
	
	\item \textit{nndescent}~\cite{NN-decent-WWW-2011}: a graph-based index that achieves efficient index construction by iteratively merging the neighbors of base vectors.
	
	\item \textit{faiss-ivf}~\cite{faiss:johnson2019billion}: the most popular partition-based method.
	
	\item \textit{\textit{faiss-ivfpqfs.}}~\cite{PQ-fast-scan-VLDB-2015}: an IVF implementation with PQ (Product Quantization)~\cite{PQ-PAMI-2014} and FastScan~\cite{PQ-fast-scan-VLDB-2015} optimizations.
	
	\item \textit{\textit{scann.}}~\cite{Learned-PQ-2020-ICML-guoruiqi}: The partition-based method developed by Google, which is highly optimized for Maximum Inner Product Search (MIPS) through anisotropic vector quantization.
	
\end{itemize}

\stitle{Performance Metrics.}
We evaluate algorithms with \textit{Recall Rate} and \textit{Queries Per Second (QPS)}~\cite{malkov2018hnsw, jayaram2019diskann}.
The recall rate is defined as the percentage of actual KNNs among the vectors retrieved by the algorithm, i.e., $ \text{Recall}@k = \frac{| \text{ANN}_k(x_q) \cap \text{NN}_k(x_q) |}{k} $.
If not specified, the QPS is the queries per second when $ k = 10 $.
Each algorithm is evaluated on a dedicated single core.

\stitle{Parameters.}
Table \ref{tab:parameter_settings} reports the parameter configurations of algorithms.
Unless otherwise specified, we report the best performance of an algorithm among all combinations of parameter settings.

\begin{table}[t!]
	\centering\vspace{-3ex}
	\small
	\scriptsize
	\caption{Parameter Settings of Algorithms.}\vspace{-4ex} \label{tab:parameter_settings}
	\resizebox{\linewidth}{!}{
		\begin{tabular}{l|l}
			\toprule
			\textbf{Algorithm} & \textbf{Construction Parameter Settings} \\
			\midrule
			\vsag & 
			\texttt{maximum\_degree} $\in \{8, 12, 16, 24, 32, 36, 48, \textbf{64}\} $ \\
			& \texttt{pruning\_rate} $\in \{\textbf{1.0}, 1.2, 1.4, 1.6, 1.8, 2.0\} $ \\
			\midrule
			\textit{hnswlib}, \textit{hnsw(faiss)} & 
			\texttt{maximum\_degree} $ \in \{4, 8, 12, 16, 24, 36, 48, \textbf{64}, 96\}$\\
			& \texttt{candidate\_pool\_size} $ \in \{\textbf{500}\}$\\
			\midrule
			\textit{nndescent} & 
			\texttt{pruning\_prob} $\in \{\textbf{0.0}, 1.0\} $ \\
			& \texttt{leaf\_size} $\in  \{24, \textbf{36}, 48\} $ \\
			& \texttt{n\_neighbors} $\in  \{10, 20, \textbf{40}, 60\} $ \\
			& \texttt{pruning\_degree\_multiplier} $\in  \{0.5, 0.75, 1.0, 1.5, \textbf{2.0}, 3.0\} $ \\
			\midrule
			\textit{faiss-ivf}, \textit{faiss-ivfpqfs} & 
			\texttt{n\_clusters} $\in \{32, 64, 128, 256, 512, 1024, 2048, 4096, 8192\} $ \\
			\midrule
			\textit{scann} & 
			\texttt{n\_leaves} $\in  \{100, 600, 1000, 1500, \textbf{2000}\} $ \\
			& \texttt{avq\_threshold} $\in  \{0.15, 0.2, \textbf{0.55}\} $ \\
			& \texttt{dims\_per\_block} $\in  \{1, 2, 3, \textbf{4}\} $ \\
			\bottomrule
		\end{tabular}
	}
	\vspace{-4ex}
\end{table}

\stitle{Environment.}
The experiments are conducted on a server with an Intel(R) Xeon(R) Platinum 8163 CPU @ 2.50GHz and 512GB memory, except that Section \ref{sec:exp:overall} is evaluated on an AWS \texttt{r6i.16xlarge} machine with hyperthreading disabled, and Section \ref{sec:exp:scalability} is evaluated on a server with 4 AMD EPYC 7T83 64-Core Processors and 2TB memory.
We implement \vsag in C++, and compile it with g++ 10.2.1, \texttt{-Ofast} flag, and  AVX-512 instructions enabled.
For baselines, we use the implementations from the official Docker images.

\subsection{Overall Performance}
\label{sec:exp:overall}
Figure~\ref{fig:overall_performance} evaluates the recall (Recall@10 and Recall@100) vs. QPS performance of algorithms.
We report the best performance of an algorithm under all possible parameter settings.
Across all datasets, \vsag can achieve higher QPS with the same recall rate.
In addition, \vsag can provide a higher QPS increase on high-dimensional vector datasets, such as GIST1M and OPENAI.
In particular, \vsag outperforms \textit{hnswlib} by 226\% in QPS on GIST1M when fixing $Recall@10=90\%$, it also provides $\sim$400\% higher QPS than \textit{hnswlib} on OPENAI when fixing $Recall@10=80\%$.
The reason is that \vsag adopts quantization methods, which can provide significant QPS increase without sacrificing the search accuracy, and they are especially effective for high-dimensional data~\cite{2024rabitq}.

\begin{table*}[t!]
	\centering
	\caption{Ablation Study of \vsag's Strategies.}\vspace{-2ex}
	\label{tab:cache_opt}
	\resizebox{\linewidth}{!}{
		\begin{tabular}{@{}l*{5}{cc}@{}}
			\toprule
			\multirow{2}{*}{Strategy} & 
			\multicolumn{2}{c}{Recall@10} & 
			\multicolumn{2}{c}{QPS} & 
			\multicolumn{2}{c}{L3 Cache Load} & 
			\multicolumn{2}{c}{L3 Cache Miss Rate} & 
			\multicolumn{2}{c}{L1 Cache Miss Rate} \\
			\cmidrule(lr){2-3} 
			\cmidrule(lr){4-5} 
			\cmidrule(lr){6-7} 
			\cmidrule(lr){8-9} 
			\cmidrule(lr){10-11}
			& GIST1M & SIFT1M & GIST1M & SIFT1M & GIST1M & SIFT1M & GIST1M & SIFT1M & GIST1M & SIFT1M \\
			\midrule
			Baseline & 90.7\% & 99.7\% & 510 & 1695 & 198M & 112M & \cellcolor{grayB}93.89\% & \cellcolor{grayB}77.88\% & \cellcolor{grayA}39.37\% & \cellcolor{grayA}17.55\%\\
			Above~+~1.Quantization & 89.8\% & 98.4\% & 1272 & 2881 & 125M & 79M &\cellcolor{grayB} 67.42\% & \cellcolor{grayB} 52.09\% & \cellcolor{grayA}19.44\% & \cellcolor{grayA}11.56\%\\
			Above~+~2. Software-based Prefetch & 89.8\% & 98.4\% & 1490 & 3332 & 120M & 53M & \cellcolor{grayB}71.71\% & \cellcolor{grayB}53.86\% & \cellcolor{grayA}16.98\% & \cellcolor{grayA}9.58\% \\
			Above~+~3. Stride Prefetch & 89.8\% & 98.4\% & 1517 & 3565 & 118M & 50M & \cellcolor{grayB}64.57\% & \cellcolor{grayB}19.26\% & \cellcolor{grayA}17.18\% & \cellcolor{grayA}9.66\%\\
			Above~+~4. ELPs Auto Tuner & 89.8\% & 98.4\% & 2052 & 4946 & 43M & 49M & \cellcolor{grayB}45.88\% & \cellcolor{grayB}32.65\% & \cellcolor{grayA}16.44\% & \cellcolor{grayA}10.11\%\\
			Above~+~5. Deterministic Access & 89.8\% & 98.4\% & 2167 & 5027 & \cellcolor{grayC}65M & \cellcolor{grayC}72M & \cellcolor{grayB}39.23\% & \cellcolor{grayB}20.98\% & \cellcolor{grayA}15.43\% & \cellcolor{grayA}9.91\% \\
			Above~+~6. PRS ($\delta=0.5$) & 89.8\% & 98.4\% & 2255 & 4668 & \cellcolor{grayC}55M & \cellcolor{grayC}63M & 55.75\% & 50.74\% & \cellcolor{grayA}15.20\% & \cellcolor{grayA}10.17\%\\
			Above~+~7. PRS ($\delta=1$) & 89.8\% & 98.4\% & 2377 & 4640 & \cellcolor{grayC}46M & \cellcolor{grayC}55M & 71.62\% & 74.73\% & \cellcolor{grayA}14.69\% & \cellcolor{grayA}9.26\%\\
			\bottomrule
		\end{tabular}
	}
\end{table*}

\vspace{-1ex}
\subsection{Scalability Performance}
\label{sec:exp:scalability}
\begin{figure}[t!]\centering
	\vspace{-1ex}
	\includegraphics[width=0.4\textwidth]{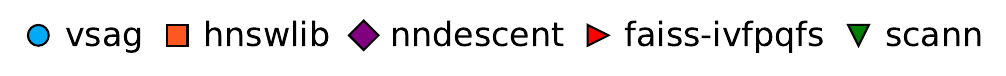}
	\hfill\\
	\vspace{-2ex}
	\subfigure[][{\scriptsize ANT-INTERNET ($\sim$10M)}]{
		\scalebox{0.3}[0.3]{\includegraphics{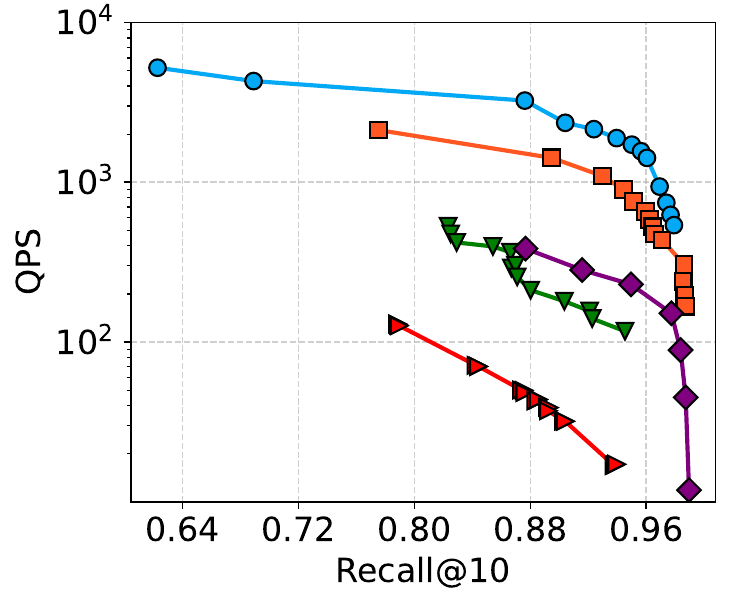}}
		\label{subfig:scale_antinternet}}
	\subfigure[][{\scriptsize MSMARCO ($\sim$100M)}]{
		\scalebox{0.3}[0.3]{\includegraphics{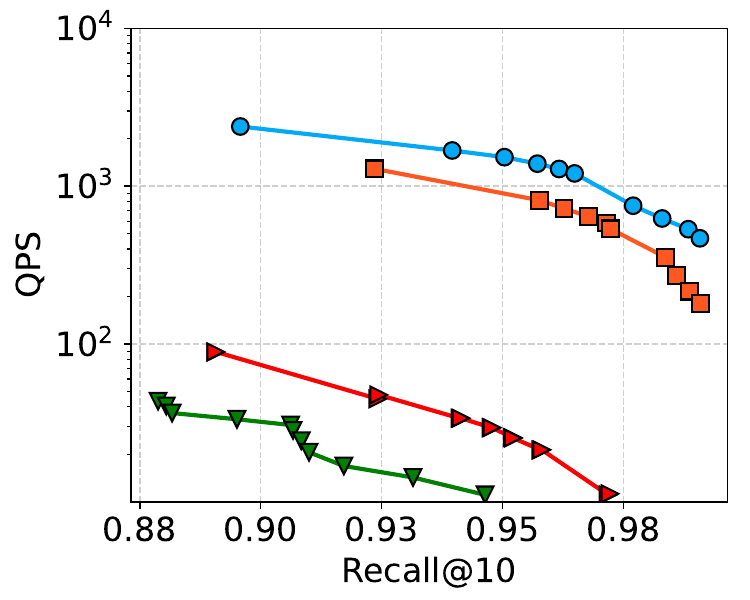}}
		\label{subfig:scale_msarco}}\vspace{-4ex}
	\caption{\small Performance Comparison on Large-Scale Datasets}
	\label{fig:scale_performance}
	\vspace{-2ex}
\end{figure}

As shown in Figure~\ref{fig:scale_performance}, we present the performance comparison of \vsag with four baselines is on two large-scale datasets: ANT-INTERNET($\sim$10M) and MSMARCO($\sim$100M). Here, ANT-INTERNET is an internal text dataset from Ant Group. We made our best effort to construct the well-tuned indexes for the MSMARCO dataset using 256 threads. The performance of hnsw (faiss) is similar to that of hnswlib, while the performance of faiss-ivf is significantly worse than that of faiss-ivfpqfs. Additionally, nndescent failed to complete index construction within two days on the MSMARCO dataset. As a result, we do not include these methods in the results.

For faiss-ivfpqfs, we followed the META-FAISS guidebook to tune its parameters. To strike a balance between construction time and performance, we use 1 million centroids and employed PQ with $256\times4$ bits. For the other algorithms, we use the parameters highlighted in bold in Table~\ref{tab:parameter_settings}.

For \vsag, the construction process took 15.37 hours with a memory footprint of 463GB on the MSMARCO dataset. As the results shown, at Recall@10=99\% \vsag achieves a significant improvement in QPS compared to hnswlib (increased from 180 to 467), which is a $2.59\times$ performance boost. Similarly, on the ANT-INTERNET dataset, at Recall@10=96\% \vsag demonstrates a remarkable enhancement in QPS rising from 659 to 1421, which is $2.15\times$ QPS of hnswlib.

\vspace{-2ex}
\subsection{Ablation Study}
\subsubsection{Cache Miss Analysis}

Table~\ref{tab:cache_opt} conduct ablation tests to investigate the effectiveness of \vsag's strategies (see \S\ref{sec:tech1_cache_miss_opt},\S\ref{sec:tech2_auto_param_opt} and\S\ref{sec:tech3_distance_opt}).
We use GIST1M and SIFT1M datasets, and set $m_c =36 ,\alpha_c =1.0$ for \vsag.
The strategies are incremental, i.e., the strategy $k$ row reports the performance of the baseline with strategies $1..k$.

Strategy 1 uses \textit{quantization methods}, and it improves QPS while ensuring the same recall rates.
Specifically, the QPS on GIST1M increases from 510 to 1272 (149\% growth), and the QPS on SIFT1M increases from 1695 to 2881 (69\% growth).
This is because quantization can significantly reduce the cost of distance computation.

Strategies 2-5 optimize \textit{memory access}.
As a result, the L3 cache miss rate reduces from 93.89\% to 39.23\% on GIST1M, and from 77.88\% to 20.98\% on SIFT1M.
Such a drop in cache miss rate leads to a 70\% QPS increase on GIST1M , and a 74\% QPS increase on SIFT1M.
This is because L3 cache loads are the dominate cost in the search time after using quantization methods.

Strategies 6-7 use \textit{PRS} to balance memory and CPU usage.
When $\delta$ increase, \vsag would allocates more memory, and the L3 cache load of  will decrease.
For example, the L3 cache load decrease by 41\% on GIST1M, and by 30\% on SIFT1M.
On GIST1M dataset, \vsag's QPS increases from 2167 to 2337, as memory pressure is the performance bottleneck.
However, on SIFT1M dataset, \vsag's QPS on GIST1M decreases from 5027 to 4640, because CPU pressure outweighs memory pressure on this dataset.
We can decide whether to use \textit{PRS} based on the workload of memory and CPU.

\begin{figure}[t!]\centering \vspace{-2ex}
	\includegraphics[width=0.5\textwidth]{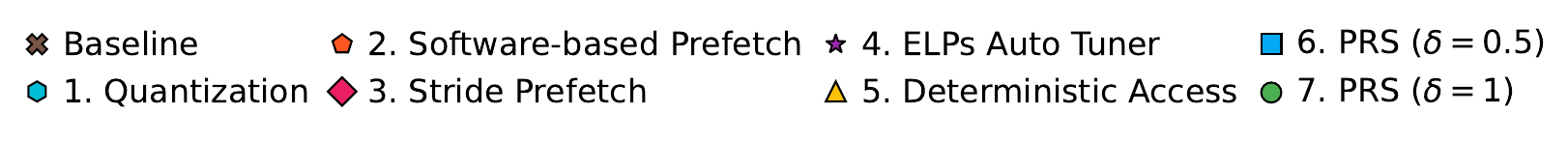}
	\hfill\\\vspace{-3ex}
	\subfigure[][{\scriptsize GIST1M}]{
		\scalebox{0.3}[0.3]{\includegraphics{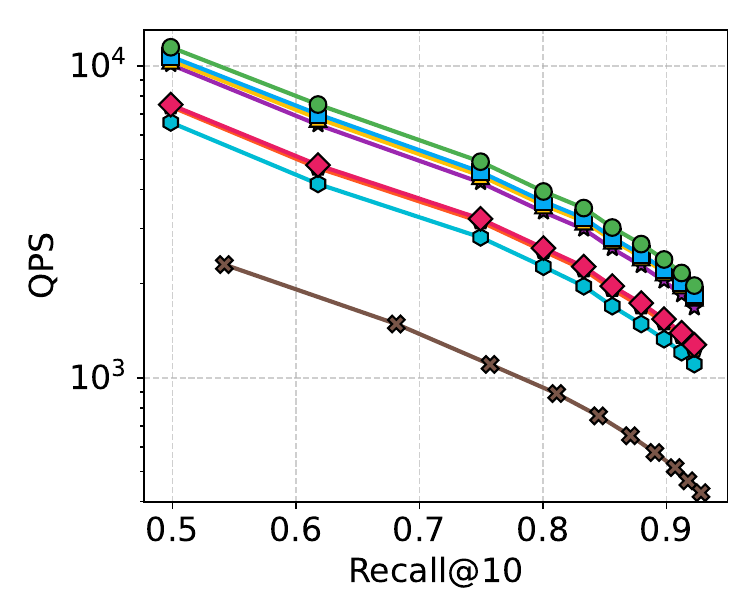}}
		\label{subfig:cache_gist}}
	\subfigure[][{\scriptsize GLOVE-100}]{
		\scalebox{0.3}[0.3]{\includegraphics{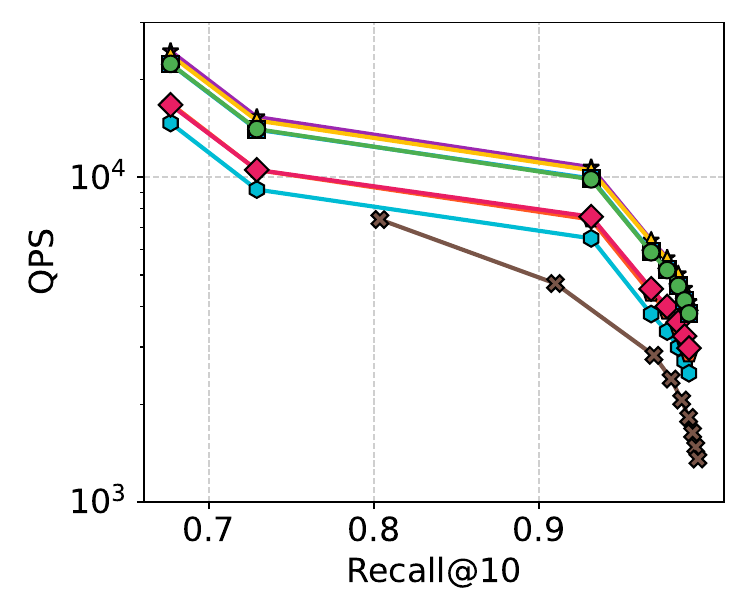}}
		\label{subfig:cache_sift}}\vspace{-4ex}
	\caption{\small Performance of \vsag's Strategies.}\vspace{-2ex}
	\label{fig:cache}
	
	\vspace{-2ex}
\end{figure}

\subsubsection{Performance}
Figure~\ref{fig:cache} illustrates the cumulative performance gaps under varying numbers of optimization strategies. For each strategy, varying $ef_s$ from 10 to 100 yields specific Recall@10 and QPS values. For example, the top-left point in the Baseline corresponds to $ef_s = 10$, while the bottom-right point corresponds to $ef_s = 100$. The strategies contributing most to QPS improvements are 1. Quantization and 4. ELPs Auto Tuner. The former drastically reduces distance computation overhead, while the latter significantly enhances the prefetch effectiveness in 3. Stride Prefetch. However, applying only 3. Stride Prefetch shows minimal difference compared to 2. Software-based Prefetch, as prefetch efficiency heavily depends on environment-level parameters $\omega$ and $\nu$.

\subsection{Evaluation of ILPs Auto-Tuner}
\subsubsection{Tuning Cost}

\begin{table*}[t!]
	\small
	\begin{center} 
		\caption{\small Time Cost Breakdown in Each Phase. }\vspace{-3ex}
		\label{tab:portion_of_cost}
		\begin{tabular}{l|l|c|c|c|c|c|c}\toprule
			\multirow{2}{*}{\bf Dataset} & \multirow{2}{*}{\bf Algorithm} 
			& \multicolumn{4}{c|}{\bf Offline Construction and Tuning Phase} & \multicolumn{2}{c}{\bf Online Search Phase (1,000 queries)} \\ 
			\cline{3-8}
			& & \bf Build Index & \bf ILP Tuning & \bf ELP Tuning & \bf QLP Training& \bf QLP Tuning& \bf Search with Query\\ \midrule
			
			\multirow{2}{*}{SIFT1M} 
			& \vsag & 5998.564s & 68.185s & 10.293s & 50.768s & $\sim$0.001s & 0.217s \\ \cline{2-8}
			& hnswlib & \multicolumn{4}{c|}{30.79 hours ($18\times$)} & - & 0.370s \\ \midrule
			
			\multirow{2}{*}{GIST1M} 
			& \vsag & 11085.626s & 119.813s & 25.397s & 149.467s & $\sim$0.001s & 1.587s \\ \cline{2-8}
			& hnswlib & \multicolumn{4}{c|}{61.64 hours ($19\times$)} & - & 4.807s \\ \bottomrule
			
		\end{tabular}
	\end{center}
	\vspace{-2ex}
\end{table*}

As shown in Table~\ref{tab:portion_of_cost}, we present the time consumption of VSAG across different phases on the SIFT1M and GIST1M datasets, varying index-level parameters $m_c \in(8,16,24,32)$ and $\alpha_c\in(1.0,1.2,1.4,1.6,1.8,2.0)$. Here, we report the total time for 1,000 queries at Recall@10 = 99\%. During the offline phase, the tuning of ILP (i.e., 68s) and ELP (i.e., 10s) takes the longest time, as it involves adjusting a large number of parameters and performing actual searches to select the optimal parameters. However, the total time spent on all tuning (i.e., 128s) processes is still significantly lower than the time required for building index (i.e., 5998s). In contrast, hnswlib requires over 30 hours to repeatedly construct indexes with different ILP during its tuning process. In the online phase, the tuning of QLP introduces only 0.001s overhead because the decision tree only relies on a small number of features. This overhead is almost negligible compared to the search cost (i.e., 0.217s). Even when including the tuning overhead of QLP, \vsag's overall search cost remains lower than that of hnswlib (i.e., 0.37s).

\begin{figure}[t!]\centering \vspace{-2ex}
	\includegraphics[width=0.3\textwidth]{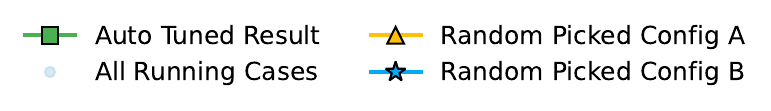}
	\hfill\\\vspace{-2ex}
	\subfigure[][{\scriptsize GIST1M}]{
		\scalebox{0.3}[0.3]{\includegraphics{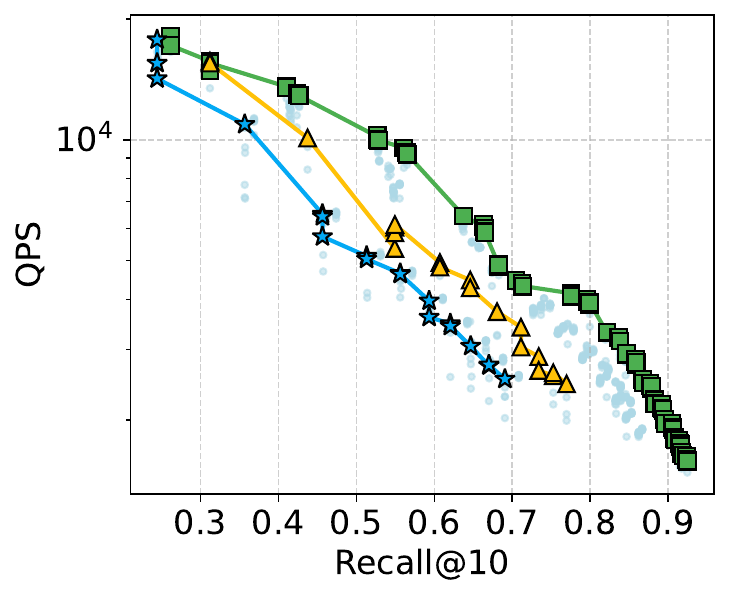}}}
	\subfigure[][{\scriptsize GLOVE-100}]{
		\scalebox{0.3}[0.3]{\includegraphics{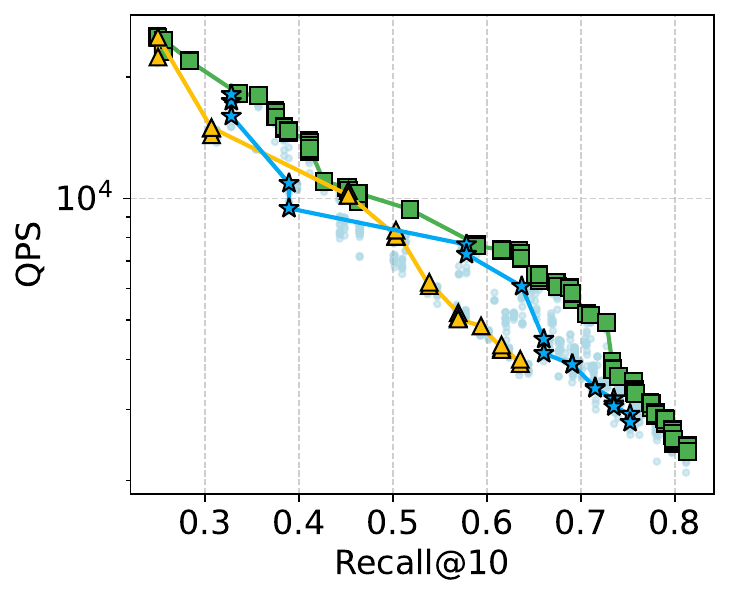}}}\vspace{-4ex}
	\caption{\small Performance of Auto-Tuned ILPs.}%\vspace{-3ex}
	\label{fig:autoILP}
	\vspace{-4ex}
\end{figure}

\subsubsection{Tuning Performance}

Beyond tuning costs, we also demonstrate the performance of the ILPs auto tuned index. As shown in Figure \ref{fig:autoILP}, we randomly selected two index-level parameter configurations (A and B) as baselines and plotted the performance of all parameter combinations in running cases. \vsag exhibits significant performance gains over these baselines. For example, on the GIST1M dataset at a fixed QPS of 2500, the worst-case Recall@10 among all running cases is 62\%, while the tuned index achieves Recall@10=88\%, representing a 26\% (absolute) improvement. At a fixed Recall@10=70\%, the worst-case QPS is 2000, while \vsag achieves 4000 QPS - an improvement of 100\%. Similar trends hold for the GLOVE-100 dataset: maximum Recall@10 improvements exceed 15\% at fixed QPS of 500, and QPS improves from around 4000 to 7000 (over 75\% gain) at a fixed recall rate of 60\%.

\subsection{Evaluation of QLPs Auto-Tuner}
Table~\ref{tab:tuningQLP} evaluates the QLPs tuning result on GIST1M and SIFT1M, under a recall guarantee of 94\% and 97\%. The baseline method (i.e., \textit{FIX}) manually selects the smallest $ef_s$ that ensures the target recall.
In contrast, \vsag employs a decision tree classification approach to divide queries into two categories: (1) Simple queries, which can converge to the target accuracy with a smaller $ef_s$ value and incur lower computational cost. For these queries, an appropriate $ef_s$ can help improve retrieval speed. (2) Complex queries, which require larger $ef_s$ values to achieve the desired accuracy. For these queries, an appropriate $ef_s$ can help improve retrieval precision.
For both types of queries, the QLPs Auto-Tuner can intelligently selects the required $ef_s$, leading to over a 5\% increase in QPS when recall thresholds equal 94\% and 97\%, respectively.

\section{Case Study}

In Ant Group, vector retrieval capabilities are provided through a distributed vector database system, which is designed to support real-time queries with highly availability and high scalability. For storage, the vector database system splits the billion-level dataset into subsets managed by  the LSM-Tree structure~\cite{2022manu}, and each subset is referred to as a segment. Our \vsag library focuses solely on indexing and retrieval performance. Therefore, when \vsag is integrated into a vector database, we separately construct a \vsag index for each segment. Then, for a coming query, it is executed on  indexes on Segments in parallel, and the vector database  system merges the results from segments as the final result. 

For the deployment of the vector database system, we use query nodes to manage the segments. Here, a query node refers to a server that responds to online requests. To minimize the impact of query node failures and avoid prolonged recovery times, the size and number of segments maintained by a single distributed node are limited. As data volume grows, scalability is achieved by adding more nodes horizontally. This approach is also widely adopted in vector databases such as Milvus~\cite{2022manu} and Vald~\cite{vald}.

For example, in an image search scenario involving approximately 10 billion images in Ant Group, each image is embedded into a 512-dimensional vector and stored in a distributed vector database cluster. The cluster is configured such that each segment contains approximately 10 million rows of data. Each query node can host up to 4 segments, and each query node is deployed on an machine instance with a 16-core CPU and 80GB memory. The entire cluster consists of approximately 400 such machine instances. By using the \vsag index, the average latency in this scenario is reduced from 3.0ms to 1.1ms per segment, and the upper limit of QPS throughput is increased by $2.65\times$, compared to using the open-source hnswlib algorithm.

\begin{table}[t!]
	\small\scriptsize
	\centering
	\renewcommand{\arraystretch}{1.2} % 增加行高，避免内容被截断
	\caption{\small Comparison of Tuning Performance of QLPs}\vspace{-3ex} \label{tab:tuningQLP}
	\begin{tabular}{l|c|cc|cc}
		\toprule
		\multirow{2}{*}{\textbf{Method}} & \multirow{2}{*}{\textbf{Metric}} & \multicolumn{2}{c|}{\textbf{GIST1M}} & \multicolumn{2}{c}{\textbf{SIFT1M}} \\
		\cmidrule(lr){3-4} \cmidrule(lr){5-6}
		&   & \textbf{94\%} & \textbf{97\%} & \textbf{94\%} & \textbf{97\%} \\
		\midrule
		\multirow{2}{*}{\centering\textbf{FIX}}  
		& Recall@10   & 94.64\%       & 97.49\%       & 94.54\%      & 97.61\%  \\
		& QPS         & 1469          & 902           & 4027          & 2834    \\
		\midrule
		\multirow{2}{*}{\centering\textbf{\vsag (Ours)}} 
		& Recall@10    & 94.71\%       & 97.58\%       & 94.63\%       & 97.66\% \\
		& QPS          & 1534          & 967           & 4050          & 2912    \\
		\bottomrule
	\end{tabular}\vspace{-2ex}
	
	\vspace{-6ex}
\end{table}

\section{Related Work}
\label{sec:related}

The mainstream methods for vector retrieval can be divided into two categories: space partitioning-based and graph-based.
Graph-based algorithms \cite{HNSW-PAMI-2020,Diskann-NIPS-2019,NSG-VLDB-2019-deng-cai,SSG-PAMI-2022-deng-cai,HVS-VLDB-2021-kejing-lu,tMRNG:journals/pacmmod/PengCCYX23,ELPIS-VLDB-2023,ANNSurvey-TKDE-2020-Wei-Wang,yang2025revisitingindexconstructionproximity} can ensure high recall with practical efficiency, e.g., HNSW~\cite{HNSW-PAMI-2020}, NSG~\cite{2019nsg}, VAMANA~\cite{Diskann-NIPS-2019}, and $\tau$-MNG~\cite{2023taumng}.
These methods build a proximity graph where each node is a base vector and edges connect pairs of nearby vectors.
During a vector search, they greedily move towards the query vector to identify its nearest neighbors.
Our \vsag framework can adapt to graph-based algorithms mentioned above, to improve the performance in production.

Space partitioning-based methods (e.g., IVFADC~\cite{PQ-fast-scan-VLDB-2015,IMI-PAMI-2015,Learning-Hash-survey-2017-PAMI,NeuralLSH-2020-ICLR}) group similar vectors into subspaces with K-means~\cite{PQ-fast-scan-VLDB-2015, johnson2019faiss} or Locality-Sensitive Hashing (LSH)~\cite{LSH-1999,LSH-Pstable-2004,C2LSH-2012-SIGMOD,PMLSH-bolong-2020,SRS-yifang-2014,Binary-LSH-NIPS-2009}.
During the search process, they traverse some vector subspaces to find the nearest neighbors.
These methods can achieve high cache hit rates due to the continuous organization of vectors, but they suffer from a low recall issue.
In comparison, graph-based ANNS algorithms (e.g., \vsag and HNSW) usually achieve a higher QPS under the same recall.
Note that the technique of \vsag cannot be applied to space partitioning-based algorithms.

\section{Conclusion}
\label{sec:conclusion}
In this work, we present \vsag, an open-source framework for ANNS that can be applied to most of the graph-based indexes. \vsag employs software-based prefetch, deterministic access greedy search, and PRS to significantly reduce the cache miss rate. \vsag has a three-level parameter tuning mechanism that automatically adjusts different parameters based on their tuning complexity. \vsag combines quantization and selective re-ranking to integrate low- and high-precision distance computations. Experiments on real-world datasets demonstrate that \vsag outperforms baselines.

\section{acknowledgment}
This work was supported by Ant Group. Heng Tao Shen and Peng Cheng are partially supported by the Fundamental Research Funds for the Central Universities. Xuemin Lin is supported by NSFC U2241211.

\balance
\bibliographystyle{ACM-Reference-Format}
\bibliography{add}


\begin{thebibliography}{61}

%%% the Web address.

\ifx \showCODEN    \undefined \def \showCODEN     #1{\unskip}     \fi
\ifx \showDOI      \undefined \def \showDOI       #1{#1}\fi
\ifx \showISBNx    \undefined \def \showISBNx     #1{\unskip}     \fi
\ifx \showISBNxiii \undefined \def \showISBNxiii  #1{\unskip}     \fi
\ifx \showISSN     \undefined \def \showISSN      #1{\unskip}     \fi
\ifx \showLCCN     \undefined \def \showLCCN      #1{\unskip}     \fi
\ifx \shownote     \undefined \def \shownote      #1{#1}          \fi
\ifx \showarticletitle \undefined \def \showarticletitle #1{#1}   \fi
\ifx \showURL      \undefined \def \showURL       {\relax}        \fi
% invisible to TeX
\providecommand\bibfield[2]{#2}
\providecommand\bibinfo[2]{#2}
\providecommand\natexlab[1]{#1}
\providecommand\showeprint[2][]{arXiv:#2}

\bibitem[\protect\citeauthoryear{Alipay}{Alipay}{2025a}]%
        {antgroup}
\bibfield{author}{\bibinfo{person}{Alipay}.} \bibinfo{year}{2025}\natexlab{a}.
\newblock \bibinfo{title}{\text{Ant Group}}.
\newblock \bibinfo{howpublished}{\url{https://www.antgroup.com}}.
\newblock


\bibitem[\protect\citeauthoryear{Alipay}{Alipay}{2025b}]%
        {alipay}
\bibfield{author}{\bibinfo{person}{Alipay}.} \bibinfo{year}{2025}\natexlab{b}.
\newblock \bibinfo{title}{\text{Face Recognition}}.
\newblock
  \bibinfo{howpublished}{\url{https://open.alipay.com/api/detail?code=I1080300001000043632}}.
\newblock


\bibitem[\protect\citeauthoryear{Amazon}{Amazon}{2025}]%
        {amazon}
\bibfield{author}{\bibinfo{person}{Amazon}.} \bibinfo{year}{2025}\natexlab{}.
\newblock \bibinfo{title}{\text{Product Search}}.
\newblock \bibinfo{howpublished}{\url{https://www.amazon.com/}}.
\newblock


\bibitem[\protect\citeauthoryear{Andr{\'e}, Kermarrec, and
  Le~Scouarnec}{Andr{\'e} et~al\mbox{.}}{2015}]%
        {PQ-fast-scan-VLDB-2015}
\bibfield{author}{\bibinfo{person}{Fabien Andr{\'e}},
  \bibinfo{person}{Anne-Marie Kermarrec}, {and} \bibinfo{person}{Nicolas
  Le~Scouarnec}.} \bibinfo{year}{2015}\natexlab{}.
\newblock \showarticletitle{Cache locality is not enough: High-Performance
  Nearest Neighbor Search with Product Quantization Fast Scan}.
\newblock \bibinfo{journal}{\emph{Proceedings of the VLDB Endowment}}
  \bibinfo{volume}{9}, \bibinfo{number}{4} (\bibinfo{year}{2015}).
\newblock


\bibitem[\protect\citeauthoryear{Aum{\"u}ller, Bernhardsson, and
  Faithfull}{Aum{\"u}ller et~al\mbox{.}}{2020}]%
        {ann-benchmakrs}
\bibfield{author}{\bibinfo{person}{Martin Aum{\"u}ller}, \bibinfo{person}{Erik
  Bernhardsson}, {and} \bibinfo{person}{Alexander Faithfull}.}
  \bibinfo{year}{2020}\natexlab{}.
\newblock \showarticletitle{ANN-Benchmarks: A benchmarking tool for approximate
  nearest neighbor algorithms}.
\newblock \bibinfo{journal}{\emph{Information Systems}}  \bibinfo{volume}{87}
  (\bibinfo{year}{2020}), \bibinfo{pages}{101374}.
\newblock


\bibitem[\protect\citeauthoryear{Ayers, Litz, Kozyrakis, and Ranganathan}{Ayers
  et~al\mbox{.}}{2020}]%
        {2020prefetchAndMemoryAccess}
\bibfield{author}{\bibinfo{person}{Grant Ayers}, \bibinfo{person}{Heiner Litz},
  \bibinfo{person}{Christos Kozyrakis}, {and} \bibinfo{person}{Parthasarathy
  Ranganathan}.} \bibinfo{year}{2020}\natexlab{}.
\newblock \showarticletitle{Classifying Memory Access Patterns for
  Prefetching}. In \bibinfo{booktitle}{\emph{Proceedings of the Twenty-Fifth
  International Conference on Architectural Support for Programming Languages
  and Operating Systems}} (Lausanne, Switzerland)
  \emph{(\bibinfo{series}{ASPLOS '20})}. \bibinfo{publisher}{Association for
  Computing Machinery}, \bibinfo{address}{New York, NY, USA},
  \bibinfo{pages}{513--526}.
\newblock
\showISBNx{9781450371025}
\urldef\tempurl%
\url{https://doi.org/10.1145/3373376.3378498}
\showDOI{\tempurl}


\bibitem[\protect\citeauthoryear{Azizi, Echihabi, and Palpanas}{Azizi
  et~al\mbox{.}}{2023}]%
        {ELPIS-VLDB-2023}
\bibfield{author}{\bibinfo{person}{Ilias Azizi}, \bibinfo{person}{Karima
  Echihabi}, {and} \bibinfo{person}{Themis Palpanas}.}
  \bibinfo{year}{2023}\natexlab{}.
\newblock \showarticletitle{Elpis: Graph-based similarity search for scalable
  data science}.
\newblock \bibinfo{journal}{\emph{Proceedings of the VLDB Endowment}}
  \bibinfo{volume}{16}, \bibinfo{number}{6} (\bibinfo{year}{2023}),
  \bibinfo{pages}{1548--1559}.
\newblock


\bibitem[\protect\citeauthoryear{Babenko and Lempitsky}{Babenko and
  Lempitsky}{2015}]%
        {IMI-PAMI-2015}
\bibfield{author}{\bibinfo{person}{Artem Babenko} {and}
  \bibinfo{person}{Victor~S. Lempitsky}.} \bibinfo{year}{2015}\natexlab{}.
\newblock \showarticletitle{The Inverted Multi-Index}.
\newblock \bibinfo{journal}{\emph{{IEEE} Trans. Pattern Anal. Mach. Intell.}}
  \bibinfo{volume}{37}, \bibinfo{number}{6} (\bibinfo{year}{2015}),
  \bibinfo{pages}{1247--1260}.
\newblock


\bibitem[\protect\citeauthoryear{Braun and Litz}{Braun and Litz}{2019}]%
        {braun2019sma}
\bibfield{author}{\bibinfo{person}{Peter Braun} {and} \bibinfo{person}{Heiner
  Litz}.} \bibinfo{year}{2019}\natexlab{}.
\newblock \showarticletitle{Understanding memory access patterns for
  prefetching}. In \bibinfo{booktitle}{\emph{International Workshop on
  AI-assisted Design for Architecture (AIDArc), held in conjunction with
  ISCA}}.
\newblock


\bibitem[\protect\citeauthoryear{Chen and Guestrin}{Chen and Guestrin}{2016}]%
        {2016xgboost}
\bibfield{author}{\bibinfo{person}{Tianqi Chen} {and} \bibinfo{person}{Carlos
  Guestrin}.} \bibinfo{year}{2016}\natexlab{}.
\newblock \showarticletitle{XGBoost: {A} Scalable Tree Boosting System}. In
  \bibinfo{booktitle}{\emph{Proceedings of the 22nd {ACM} {SIGKDD}
  International Conference on Knowledge Discovery and Data Mining, San
  Francisco, CA, USA, August 13-17, 2016}},
  \bibfield{editor}{\bibinfo{person}{Balaji Krishnapuram},
  \bibinfo{person}{Mohak Shah}, \bibinfo{person}{Alexander~J. Smola},
  \bibinfo{person}{Charu~C. Aggarwal}, \bibinfo{person}{Dou Shen}, {and}
  \bibinfo{person}{Rajeev Rastogi}} (Eds.). \bibinfo{publisher}{{ACM}},
  \bibinfo{pages}{785--794}.
\newblock


\bibitem[\protect\citeauthoryear{Chen and Baer}{Chen and Baer}{1995}]%
        {chen1995hardwarePrefetch}
\bibfield{author}{\bibinfo{person}{Tien-Fu Chen} {and}
  \bibinfo{person}{Jean-Loup Baer}.} \bibinfo{year}{1995}\natexlab{}.
\newblock \showarticletitle{Effective hardware-based data prefetching for
  high-performance processors}.
\newblock \bibinfo{journal}{\emph{IEEE transactions on computers}}
  \bibinfo{volume}{44}, \bibinfo{number}{5} (\bibinfo{year}{1995}),
  \bibinfo{pages}{609--623}.
\newblock


\bibitem[\protect\citeauthoryear{Datar, Immorlica, Indyk, and Mirrokni}{Datar
  et~al\mbox{.}}{2004}]%
        {LSH-Pstable-2004}
\bibfield{author}{\bibinfo{person}{Mayur Datar}, \bibinfo{person}{Nicole
  Immorlica}, \bibinfo{person}{Piotr Indyk}, {and} \bibinfo{person}{Vahab~S
  Mirrokni}.} \bibinfo{year}{2004}\natexlab{}.
\newblock \showarticletitle{Locality-sensitive hashing scheme based on p-stable
  distributions}. In \bibinfo{booktitle}{\emph{Proceedings of the twentieth
  annual symposium on Computational geometry}}. \bibinfo{pages}{253--262}.
\newblock


\bibitem[\protect\citeauthoryear{Dong, Moses, and Li}{Dong
  et~al\mbox{.}}{2011}]%
        {NN-decent-WWW-2011}
\bibfield{author}{\bibinfo{person}{Wei Dong}, \bibinfo{person}{Charikar Moses},
  {and} \bibinfo{person}{Kai Li}.} \bibinfo{year}{2011}\natexlab{}.
\newblock \showarticletitle{Efficient k-nearest neighbor graph construction for
  generic similarity measures}. In \bibinfo{booktitle}{\emph{Proceedings of the
  20th international conference on World wide web}}. \bibinfo{pages}{577--586}.
\newblock


\bibitem[\protect\citeauthoryear{Dong, Indyk, Razenshteyn, and Wagner}{Dong
  et~al\mbox{.}}{2020}]%
        {NeuralLSH-2020-ICLR}
\bibfield{author}{\bibinfo{person}{Yihe Dong}, \bibinfo{person}{Piotr Indyk},
  \bibinfo{person}{Ilya~P Razenshteyn}, {and} \bibinfo{person}{Tal Wagner}.}
  \bibinfo{year}{2020}\natexlab{}.
\newblock \showarticletitle{Learning Space Partitions for Nearest Neighbor
  Search}.
\newblock \bibinfo{journal}{\emph{ICLR}} (\bibinfo{year}{2020}).
\newblock


\bibitem[\protect\citeauthoryear{Fu, Wang, and Cai}{Fu et~al\mbox{.}}{2022}]%
        {SSG-PAMI-2022-deng-cai}
\bibfield{author}{\bibinfo{person}{Cong Fu}, \bibinfo{person}{Changxu Wang},
  {and} \bibinfo{person}{Deng Cai}.} \bibinfo{year}{2022}\natexlab{}.
\newblock \showarticletitle{High Dimensional Similarity Search With Satellite
  System Graph: Efficiency, Scalability, and Unindexed Query Compatibility}.
\newblock \bibinfo{journal}{\emph{{IEEE} Trans. Pattern Anal. Mach. Intell.}}
  \bibinfo{volume}{44}, \bibinfo{number}{8} (\bibinfo{year}{2022}),
  \bibinfo{pages}{4139--4150}.
\newblock


\bibitem[\protect\citeauthoryear{Fu, Xiang, Wang, and Cai}{Fu
  et~al\mbox{.}}{2019a}]%
        {NSG-VLDB-2019-deng-cai}
\bibfield{author}{\bibinfo{person}{Cong Fu}, \bibinfo{person}{Chao Xiang},
  \bibinfo{person}{Changxu Wang}, {and} \bibinfo{person}{Deng Cai}.}
  \bibinfo{year}{2019}\natexlab{a}.
\newblock \showarticletitle{Fast Approximate Nearest Neighbor Search With The
  Navigating Spreading-out Graph}.
\newblock \bibinfo{journal}{\emph{Proc. {VLDB} Endow.}} \bibinfo{volume}{12},
  \bibinfo{number}{5} (\bibinfo{year}{2019}), \bibinfo{pages}{461--474}.
\newblock


\bibitem[\protect\citeauthoryear{Fu, Xiang, Wang, and Cai}{Fu
  et~al\mbox{.}}{2019b}]%
        {2019nsg}
\bibfield{author}{\bibinfo{person}{Cong Fu}, \bibinfo{person}{Chao Xiang},
  \bibinfo{person}{Changxu Wang}, {and} \bibinfo{person}{Deng Cai}.}
  \bibinfo{year}{2019}\natexlab{b}.
\newblock \showarticletitle{Fast approximate nearest neighbor search with the
  navigating spreading-out graph}.
\newblock \bibinfo{journal}{\emph{Proc. VLDB Endow.}} \bibinfo{volume}{12},
  \bibinfo{number}{5} (\bibinfo{date}{Jan.} \bibinfo{year}{2019}),
  \bibinfo{pages}{461–474}.
\newblock
\showISSN{2150-8097}


\bibitem[\protect\citeauthoryear{Gan, Feng, Fang, and Ng}{Gan
  et~al\mbox{.}}{2012}]%
        {C2LSH-2012-SIGMOD}
\bibfield{author}{\bibinfo{person}{Junhao Gan}, \bibinfo{person}{Jianlin Feng},
  \bibinfo{person}{Qiong Fang}, {and} \bibinfo{person}{Wilfred Ng}.}
  \bibinfo{year}{2012}\natexlab{}.
\newblock \showarticletitle{Locality-sensitive hashing scheme based on dynamic
  collision counting}. In \bibinfo{booktitle}{\emph{Proceedings of the 2012 ACM
  SIGMOD international conference on management of data}}.
  \bibinfo{pages}{541--552}.
\newblock


\bibitem[\protect\citeauthoryear{Gao, Gou, Xu, Yang, Long, and Wong}{Gao
  et~al\mbox{.}}{2024a}]%
        {ExRaBitQ-arxiv-2024}
\bibfield{author}{\bibinfo{person}{Jianyang Gao}, \bibinfo{person}{Yutong Gou},
  \bibinfo{person}{Yuexuan Xu}, \bibinfo{person}{Yongyi Yang},
  \bibinfo{person}{Cheng Long}, {and} \bibinfo{person}{Raymond Chi-Wing Wong}.}
  \bibinfo{year}{2024}\natexlab{a}.
\newblock \showarticletitle{Practical and Asymptotically Optimal Quantization
  of High-Dimensional Vectors in Euclidean Space for Approximate Nearest
  Neighbor Search}.
\newblock \bibinfo{journal}{\emph{arXiv preprint arXiv:2409.09913}}
  (\bibinfo{year}{2024}).
\newblock


\bibitem[\protect\citeauthoryear{Gao and Long}{Gao and Long}{2023}]%
        {ADSampling:journals/sigmod/GaoL23}
\bibfield{author}{\bibinfo{person}{Jianyang Gao} {and} \bibinfo{person}{Cheng
  Long}.} \bibinfo{year}{2023}\natexlab{}.
\newblock \showarticletitle{High-Dimensional Approximate Nearest Neighbor
  Search: with Reliable and Efficient Distance Comparison Operations}.
\newblock \bibinfo{journal}{\emph{Proc. {ACM} Manag. Data}}
  \bibinfo{volume}{1}, \bibinfo{number}{2} (\bibinfo{year}{2023}),
  \bibinfo{pages}{137:1--137:27}.
\newblock
\urldef\tempurl%
\url{https://doi.org/10.1145/3589282}
\showDOI{\tempurl}


\bibitem[\protect\citeauthoryear{Gao and Long}{Gao and Long}{2024}]%
        {2024rabitq}
\bibfield{author}{\bibinfo{person}{Jianyang Gao} {and} \bibinfo{person}{Cheng
  Long}.} \bibinfo{year}{2024}\natexlab{}.
\newblock \showarticletitle{RaBitQ: Quantizing High-Dimensional Vectors with a
  Theoretical Error Bound for Approximate Nearest Neighbor Search}.
\newblock \bibinfo{journal}{\emph{Proc. {ACM} Manag. Data}}
  \bibinfo{volume}{2}, \bibinfo{number}{3} (\bibinfo{year}{2024}),
  \bibinfo{pages}{167}.
\newblock


\bibitem[\protect\citeauthoryear{Gao, Xiong, Gao, Jia, Pan, Bi, Dai, Sun, Wang,
  and Wang}{Gao et~al\mbox{.}}{2024b}]%
        {gao2024RAGsurvey}
\bibfield{author}{\bibinfo{person}{Yunfan Gao}, \bibinfo{person}{Yun Xiong},
  \bibinfo{person}{Xinyu Gao}, \bibinfo{person}{Kangxiang Jia},
  \bibinfo{person}{Jinliu Pan}, \bibinfo{person}{Yuxi Bi}, \bibinfo{person}{Yi
  Dai}, \bibinfo{person}{Jiawei Sun}, \bibinfo{person}{Meng Wang}, {and}
  \bibinfo{person}{Haofen Wang}.} \bibinfo{year}{2024}\natexlab{b}.
\newblock \bibinfo{title}{Retrieval-Augmented Generation for Large Language
  Models: A Survey}.
\newblock
\newblock
\showeprint[arxiv]{2312.10997}~[cs.CL]


\bibitem[\protect\citeauthoryear{Gionis, Indyk, Motwani, et~al\mbox{.}}{Gionis
  et~al\mbox{.}}{1999}]%
        {LSH-1999}
\bibfield{author}{\bibinfo{person}{Aristides Gionis}, \bibinfo{person}{Piotr
  Indyk}, \bibinfo{person}{Rajeev Motwani}, {et~al\mbox{.}}}
  \bibinfo{year}{1999}\natexlab{}.
\newblock \showarticletitle{Similarity search in high dimensions via hashing}.
  In \bibinfo{booktitle}{\emph{Vldb}}, Vol.~\bibinfo{volume}{99}.
  \bibinfo{pages}{518--529}.
\newblock


\bibitem[\protect\citeauthoryear{Google}{Google}{2025}]%
        {google}
\bibfield{author}{\bibinfo{person}{Google}.} \bibinfo{year}{2025}\natexlab{}.
\newblock \bibinfo{title}{\text{Search Engine}}.
\newblock \bibinfo{howpublished}{\url{https://www.google.com/}}.
\newblock


\bibitem[\protect\citeauthoryear{Guo, Luan, Xiang, Yan, Yi, Luo, Cheng, Xu,
  Luo, Liu, Cao, Qiao, Wang, Tang, and Xie}{Guo et~al\mbox{.}}{2022}]%
        {2022manu}
\bibfield{author}{\bibinfo{person}{Rentong Guo}, \bibinfo{person}{Xiaofan
  Luan}, \bibinfo{person}{Long Xiang}, \bibinfo{person}{Xiao Yan},
  \bibinfo{person}{Xiaomeng Yi}, \bibinfo{person}{Jigao Luo},
  \bibinfo{person}{Qianya Cheng}, \bibinfo{person}{Weizhi Xu},
  \bibinfo{person}{Jiarui Luo}, \bibinfo{person}{Frank Liu},
  \bibinfo{person}{Zhenshan Cao}, \bibinfo{person}{Yanliang Qiao},
  \bibinfo{person}{Ting Wang}, \bibinfo{person}{Bo Tang}, {and}
  \bibinfo{person}{Charles Xie}.} \bibinfo{year}{2022}\natexlab{}.
\newblock \showarticletitle{Manu: a cloud native vector database management
  system}.
\newblock \bibinfo{journal}{\emph{Proc. VLDB Endow.}} \bibinfo{volume}{15},
  \bibinfo{number}{12} (\bibinfo{date}{Aug.} \bibinfo{year}{2022}),
  \bibinfo{pages}{3548–3561}.
\newblock
\showISSN{2150-8097}


\bibitem[\protect\citeauthoryear{Guo, Sun, Lindgren, Geng, Simcha, Chern, and
  Kumar}{Guo et~al\mbox{.}}{2020}]%
        {Learned-PQ-2020-ICML-guoruiqi}
\bibfield{author}{\bibinfo{person}{Ruiqi Guo}, \bibinfo{person}{Philip Sun},
  \bibinfo{person}{Erik Lindgren}, \bibinfo{person}{Quan Geng},
  \bibinfo{person}{David Simcha}, \bibinfo{person}{Felix Chern}, {and}
  \bibinfo{person}{Sanjiv Kumar}.} \bibinfo{year}{2020}\natexlab{}.
\newblock \showarticletitle{Accelerating large-scale inference with anisotropic
  vector quantization}. In \bibinfo{booktitle}{\emph{International Conference
  on Machine Learning}}. PMLR, \bibinfo{pages}{3887--3896}.
\newblock


\bibitem[\protect\citeauthoryear{Indyk and Motwani}{Indyk and Motwani}{1998}]%
        {Curse-of-dim-1998}
\bibfield{author}{\bibinfo{person}{Piotr Indyk} {and} \bibinfo{person}{Rajeev
  Motwani}.} \bibinfo{year}{1998}\natexlab{}.
\newblock \showarticletitle{Approximate nearest neighbors: towards removing the
  curse of dimensionality}. In \bibinfo{booktitle}{\emph{Proceedings of the
  thirtieth annual ACM symposium on Theory of computing}}.
  \bibinfo{pages}{604--613}.
\newblock


\bibitem[\protect\citeauthoryear{Intel}{Intel}{2025}]%
        {avx512}
\bibfield{author}{\bibinfo{person}{Intel}.} \bibinfo{year}{2025}\natexlab{}.
\newblock \bibinfo{title}{\text{AVX512}}.
\newblock
  \bibinfo{howpublished}{\url{https://www.intel.com/content/www/us/en/products/docs/accelerator-engines/what-is-intel-avx-512.html}}.
\newblock


\bibitem[\protect\citeauthoryear{Jayaram~Subramanya, Devvrit, Simhadri,
  Krishnawamy, and Kadekodi}{Jayaram~Subramanya et~al\mbox{.}}{2019a}]%
        {Diskann-NIPS-2019}
\bibfield{author}{\bibinfo{person}{Suhas Jayaram~Subramanya},
  \bibinfo{person}{Fnu Devvrit}, \bibinfo{person}{Harsha~Vardhan Simhadri},
  \bibinfo{person}{Ravishankar Krishnawamy}, {and} \bibinfo{person}{Rohan
  Kadekodi}.} \bibinfo{year}{2019}\natexlab{a}.
\newblock \showarticletitle{Diskann: Fast accurate billion-point nearest
  neighbor search on a single node}.
\newblock \bibinfo{journal}{\emph{Advances in Neural Information Processing
  Systems}}  \bibinfo{volume}{32} (\bibinfo{year}{2019}).
\newblock


\bibitem[\protect\citeauthoryear{Jayaram~Subramanya, Devvrit, Simhadri,
  Krishnawamy, and Kadekodi}{Jayaram~Subramanya et~al\mbox{.}}{2019b}]%
        {jayaram2019diskann}
\bibfield{author}{\bibinfo{person}{Suhas Jayaram~Subramanya},
  \bibinfo{person}{Fnu Devvrit}, \bibinfo{person}{Harsha~Vardhan Simhadri},
  \bibinfo{person}{Ravishankar Krishnawamy}, {and} \bibinfo{person}{Rohan
  Kadekodi}.} \bibinfo{year}{2019}\natexlab{b}.
\newblock \showarticletitle{Diskann: Fast accurate billion-point nearest
  neighbor search on a single node}.
\newblock \bibinfo{journal}{\emph{Advances in neural information processing
  Systems}}  \bibinfo{volume}{32} (\bibinfo{year}{2019}).
\newblock


\bibitem[\protect\citeauthoryear{J{\'{e}}gou, Douze, and Schmid}{J{\'{e}}gou
  et~al\mbox{.}}{2011a}]%
        {jegou2010pq_ivfadc}
\bibfield{author}{\bibinfo{person}{Herv{\'{e}} J{\'{e}}gou},
  \bibinfo{person}{Matthijs Douze}, {and} \bibinfo{person}{Cordelia Schmid}.}
  \bibinfo{year}{2011}\natexlab{a}.
\newblock \showarticletitle{Product Quantization for Nearest Neighbor Search}.
\newblock \bibinfo{journal}{\emph{{IEEE} Trans. Pattern Anal. Mach. Intell.}}
  \bibinfo{volume}{33}, \bibinfo{number}{1} (\bibinfo{year}{2011}),
  \bibinfo{pages}{117--128}.
\newblock


\bibitem[\protect\citeauthoryear{J{\'{e}}gou, Douze, and Schmid}{J{\'{e}}gou
  et~al\mbox{.}}{2011b}]%
        {PQ-PAMI-2014}
\bibfield{author}{\bibinfo{person}{Herv{\'{e}} J{\'{e}}gou},
  \bibinfo{person}{Matthijs Douze}, {and} \bibinfo{person}{Cordelia Schmid}.}
  \bibinfo{year}{2011}\natexlab{b}.
\newblock \showarticletitle{Product Quantization for Nearest Neighbor Search}.
\newblock \bibinfo{journal}{\emph{{IEEE} Trans. Pattern Anal. Mach. Intell.}}
  \bibinfo{volume}{33}, \bibinfo{number}{1} (\bibinfo{year}{2011}),
  \bibinfo{pages}{117--128}.
\newblock


\bibitem[\protect\citeauthoryear{Ji, Liu, Dai, Yang, Zheng, Wu, Dun, Gu, and
  Yan}{Ji et~al\mbox{.}}{2025}]%
        {ji2025enhancingmultistepreasoning}
\bibfield{author}{\bibinfo{person}{Kaixuan Ji}, \bibinfo{person}{Guanlin Liu},
  \bibinfo{person}{Ning Dai}, \bibinfo{person}{Qingping Yang},
  \bibinfo{person}{Renjie Zheng}, \bibinfo{person}{Zheng Wu},
  \bibinfo{person}{Chen Dun}, \bibinfo{person}{Quanquan Gu}, {and}
  \bibinfo{person}{Lin Yan}.} \bibinfo{year}{2025}\natexlab{}.
\newblock \bibinfo{title}{Enhancing Multi-Step Reasoning Abilities of Language
  Models through Direct Q-Function Optimization}.
\newblock
\newblock
\showeprint[arxiv]{2410.09302}~[cs.LG]


\bibitem[\protect\citeauthoryear{Johnson, Douze, and J{\'e}gou}{Johnson
  et~al\mbox{.}}{2019a}]%
        {faiss:johnson2019billion}
\bibfield{author}{\bibinfo{person}{Jeff Johnson}, \bibinfo{person}{Matthijs
  Douze}, {and} \bibinfo{person}{Herv{\'e} J{\'e}gou}.}
  \bibinfo{year}{2019}\natexlab{a}.
\newblock \showarticletitle{Billion-scale similarity search with {GPUs}}.
\newblock \bibinfo{journal}{\emph{IEEE Transactions on Big Data}}
  \bibinfo{volume}{7}, \bibinfo{number}{3} (\bibinfo{year}{2019}),
  \bibinfo{pages}{535--547}.
\newblock


\bibitem[\protect\citeauthoryear{Johnson, Douze, and J{\'e}gou}{Johnson
  et~al\mbox{.}}{2019b}]%
        {johnson2019faiss}
\bibfield{author}{\bibinfo{person}{Jeff Johnson}, \bibinfo{person}{Matthijs
  Douze}, {and} \bibinfo{person}{Herv{\'e} J{\'e}gou}.}
  \bibinfo{year}{2019}\natexlab{b}.
\newblock \showarticletitle{Billion-scale similarity search with GPUs}.
\newblock \bibinfo{journal}{\emph{IEEE Transactions on Big Data}}
  \bibinfo{volume}{7}, \bibinfo{number}{3} (\bibinfo{year}{2019}),
  \bibinfo{pages}{535--547}.
\newblock


\bibitem[\protect\citeauthoryear{Lewis, Perez, Piktus, Petroni, Karpukhin,
  Goyal, K{\"u}ttler, Lewis, Yih, Rockt{\"a}schel, et~al\mbox{.}}{Lewis
  et~al\mbox{.}}{2020}]%
        {LLM-RAG-NIPS-2020}
\bibfield{author}{\bibinfo{person}{Patrick Lewis}, \bibinfo{person}{Ethan
  Perez}, \bibinfo{person}{Aleksandra Piktus}, \bibinfo{person}{Fabio Petroni},
  \bibinfo{person}{Vladimir Karpukhin}, \bibinfo{person}{Naman Goyal},
  \bibinfo{person}{Heinrich K{\"u}ttler}, \bibinfo{person}{Mike Lewis},
  \bibinfo{person}{Wen-tau Yih}, \bibinfo{person}{Tim Rockt{\"a}schel},
  {et~al\mbox{.}}} \bibinfo{year}{2020}\natexlab{}.
\newblock \showarticletitle{Retrieval-augmented generation for
  knowledge-intensive nlp tasks}.
\newblock \bibinfo{journal}{\emph{Advances in Neural Information Processing
  Systems}}  \bibinfo{volume}{33} (\bibinfo{year}{2020}),
  \bibinfo{pages}{9459--9474}.
\newblock


\bibitem[\protect\citeauthoryear{Li, Zhang, Sun, Wang, Li, Zhang, and Lin}{Li
  et~al\mbox{.}}{2020}]%
        {ANNSurvey-TKDE-2020-Wei-Wang}
\bibfield{author}{\bibinfo{person}{Wen Li}, \bibinfo{person}{Ying Zhang},
  \bibinfo{person}{Yifang Sun}, \bibinfo{person}{Wei Wang},
  \bibinfo{person}{Mingjie Li}, \bibinfo{person}{Wenjie Zhang}, {and}
  \bibinfo{person}{Xuemin Lin}.} \bibinfo{year}{2020}\natexlab{}.
\newblock \showarticletitle{Approximate Nearest Neighbor Search on High
  Dimensional Data - Experiments, Analyses, and Improvement}.
\newblock \bibinfo{journal}{\emph{{IEEE} Trans. Knowl. Data Eng.}}
  \bibinfo{volume}{32}, \bibinfo{number}{8} (\bibinfo{year}{2020}),
  \bibinfo{pages}{1475--1488}.
\newblock


\bibitem[\protect\citeauthoryear{Lu, Kudo, Xiao, and Ishikawa}{Lu
  et~al\mbox{.}}{2021}]%
        {HVS-VLDB-2021-kejing-lu}
\bibfield{author}{\bibinfo{person}{Kejing Lu}, \bibinfo{person}{Mineichi Kudo},
  \bibinfo{person}{Chuan Xiao}, {and} \bibinfo{person}{Yoshiharu Ishikawa}.}
  \bibinfo{year}{2021}\natexlab{}.
\newblock \showarticletitle{{HVS:} Hierarchical Graph Structure Based on
  Voronoi Diagrams for Solving Approximate Nearest Neighbor Search}.
\newblock \bibinfo{journal}{\emph{Proc. {VLDB} Endow.}} \bibinfo{volume}{15},
  \bibinfo{number}{2} (\bibinfo{year}{2021}), \bibinfo{pages}{246--258}.
\newblock


\bibitem[\protect\citeauthoryear{Malkov and Yashunin}{Malkov and
  Yashunin}{[n.d.]}]%
        {hnswlib}
\bibfield{author}{\bibinfo{person}{Yu~A Malkov} {and} \bibinfo{person}{Dmitry~A
  Yashunin}.} \bibinfo{year}{[n.d.]}\natexlab{}.
\newblock \bibinfo{title}{hnswlib}.
\newblock \bibinfo{howpublished}{\url{https://github.com/nmslib/hnswlib}}.
\newblock


\bibitem[\protect\citeauthoryear{Malkov and Yashunin}{Malkov and
  Yashunin}{2020a}]%
        {HNSW-PAMI-2020}
\bibfield{author}{\bibinfo{person}{Yury~A. Malkov} {and}
  \bibinfo{person}{Dmitry~A. Yashunin}.} \bibinfo{year}{2020}\natexlab{a}.
\newblock \showarticletitle{Efficient and Robust Approximate Nearest Neighbor
  Search Using Hierarchical Navigable Small World Graphs}.
\newblock \bibinfo{journal}{\emph{{IEEE} Trans. Pattern Anal. Mach. Intell.}}
  \bibinfo{volume}{42}, \bibinfo{number}{4} (\bibinfo{year}{2020}),
  \bibinfo{pages}{824--836}.
\newblock


\bibitem[\protect\citeauthoryear{Malkov and Yashunin}{Malkov and
  Yashunin}{2020b}]%
        {malkov2018hnsw}
\bibfield{author}{\bibinfo{person}{Yury~A. Malkov} {and}
  \bibinfo{person}{Dmitry~A. Yashunin}.} \bibinfo{year}{2020}\natexlab{b}.
\newblock \showarticletitle{Efficient and Robust Approximate Nearest Neighbor
  Search Using Hierarchical Navigable Small World Graphs}.
\newblock \bibinfo{journal}{\emph{{IEEE} Trans. Pattern Anal. Mach. Intell.}}
  \bibinfo{volume}{42}, \bibinfo{number}{4} (\bibinfo{year}{2020}),
  \bibinfo{pages}{824--836}.
\newblock


\bibitem[\protect\citeauthoryear{Peng, Choi, Chan, Yang, and Xu}{Peng
  et~al\mbox{.}}{2023a}]%
        {tMRNG:journals/pacmmod/PengCCYX23}
\bibfield{author}{\bibinfo{person}{Yun Peng}, \bibinfo{person}{Byron Choi},
  \bibinfo{person}{Tsz~Nam Chan}, \bibinfo{person}{Jianye Yang}, {and}
  \bibinfo{person}{Jianliang Xu}.} \bibinfo{year}{2023}\natexlab{a}.
\newblock \showarticletitle{Efficient Approximate Nearest Neighbor Search in
  Multi-dimensional Databases}.
\newblock \bibinfo{journal}{\emph{Proc. {ACM} Manag. Data}}
  \bibinfo{volume}{1}, \bibinfo{number}{1} (\bibinfo{year}{2023}),
  \bibinfo{pages}{54:1--54:27}.
\newblock
\urldef\tempurl%
\url{https://doi.org/10.1145/3588908}
\showDOI{\tempurl}


\bibitem[\protect\citeauthoryear{Peng, Choi, Chan, Yang, and Xu}{Peng
  et~al\mbox{.}}{2023b}]%
        {2023taumng}
\bibfield{author}{\bibinfo{person}{Yun Peng}, \bibinfo{person}{Byron Choi},
  \bibinfo{person}{Tsz~Nam Chan}, \bibinfo{person}{Jianye Yang}, {and}
  \bibinfo{person}{Jianliang Xu}.} \bibinfo{year}{2023}\natexlab{b}.
\newblock \showarticletitle{Efficient Approximate Nearest Neighbor Search in
  Multi-dimensional Databases}.
\newblock \bibinfo{journal}{\emph{Proc. ACM Manag. Data}} \bibinfo{volume}{1},
  \bibinfo{number}{1}, Article \bibinfo{articleno}{54} (\bibinfo{date}{May}
  \bibinfo{year}{2023}), \bibinfo{numpages}{27}~pages.
\newblock


\bibitem[\protect\citeauthoryear{Raginsky and Lazebnik}{Raginsky and
  Lazebnik}{2009}]%
        {Binary-LSH-NIPS-2009}
\bibfield{author}{\bibinfo{person}{Maxim Raginsky} {and}
  \bibinfo{person}{Svetlana Lazebnik}.} \bibinfo{year}{2009}\natexlab{}.
\newblock \showarticletitle{Locality-sensitive binary codes from
  shift-invariant kernels}.
\newblock \bibinfo{journal}{\emph{Advances in neural information processing
  systems}}  \bibinfo{volume}{22} (\bibinfo{year}{2009}).
\newblock


\bibitem[\protect\citeauthoryear{Shen, Yan, Zhang, Hu, Du, and He}{Shen
  et~al\mbox{.}}{2025}]%
        {shen2025codi}
\bibfield{author}{\bibinfo{person}{Zhenyi Shen}, \bibinfo{person}{Hanqi Yan},
  \bibinfo{person}{Linhai Zhang}, \bibinfo{person}{Zhanghao Hu},
  \bibinfo{person}{Yali Du}, {and} \bibinfo{person}{Yulan He}.}
  \bibinfo{year}{2025}\natexlab{}.
\newblock \bibinfo{title}{CODI: Compressing Chain-of-Thought into Continuous
  Space via Self-Distillation}.
\newblock
\newblock
\showeprint[arxiv]{2502.21074}~[cs.CL]


\bibitem[\protect\citeauthoryear{Sun, Wang, Qin, Zhang, and Lin}{Sun
  et~al\mbox{.}}{2014}]%
        {SRS-yifang-2014}
\bibfield{author}{\bibinfo{person}{Yifang Sun}, \bibinfo{person}{Wei Wang},
  \bibinfo{person}{Jianbin Qin}, \bibinfo{person}{Ying Zhang}, {and}
  \bibinfo{person}{Xuemin Lin}.} \bibinfo{year}{2014}\natexlab{}.
\newblock \showarticletitle{{SRS:} Solving c-Approximate Nearest Neighbor
  Queries in High Dimensional Euclidean Space with a Tiny Index}.
\newblock \bibinfo{journal}{\emph{Proc. {VLDB} Endow.}} \bibinfo{volume}{8},
  \bibinfo{number}{1} (\bibinfo{year}{2014}), \bibinfo{pages}{1--12}.
\newblock


\bibitem[\protect\citeauthoryear{Taobao}{Taobao}{2025}]%
        {taobao}
\bibfield{author}{\bibinfo{person}{Taobao}.} \bibinfo{year}{2025}\natexlab{}.
\newblock \bibinfo{title}{\text{Product Search}}.
\newblock \bibinfo{howpublished}{\url{https://www.taobao.com/}}.
\newblock


\bibitem[\protect\citeauthoryear{{Vald}}{{Vald}}{2021}]%
        {vald}
\bibfield{author}{\bibinfo{person}{{Vald}}.} \bibinfo{year}{2021}\natexlab{}.
\newblock \bibinfo{title}{\text{Vald}}.
\newblock \bibinfo{howpublished}{\url{https://github.com/vdaas/vald}}.
\newblock


\bibitem[\protect\citeauthoryear{Vanderwiel and Lilja}{Vanderwiel and
  Lilja}{2000}]%
        {vanderwiel2000stride}
\bibfield{author}{\bibinfo{person}{Steven~P Vanderwiel} {and}
  \bibinfo{person}{David~J Lilja}.} \bibinfo{year}{2000}\natexlab{}.
\newblock \showarticletitle{Data prefetch mechanisms}.
\newblock \bibinfo{journal}{\emph{ACM Computing Surveys (CSUR)}}
  \bibinfo{volume}{32}, \bibinfo{number}{2} (\bibinfo{year}{2000}),
  \bibinfo{pages}{174--199}.
\newblock


\bibitem[\protect\citeauthoryear{Veidenbaum, Tang, Gupta, Nicolau, and
  Ji}{Veidenbaum et~al\mbox{.}}{1999}]%
        {veidenbaum1999cacheline}
\bibfield{author}{\bibinfo{person}{Alexander~V Veidenbaum},
  \bibinfo{person}{Weiyu Tang}, \bibinfo{person}{Rajesh Gupta},
  \bibinfo{person}{Alexandru Nicolau}, {and} \bibinfo{person}{Xiaomei Ji}.}
  \bibinfo{year}{1999}\natexlab{}.
\newblock \showarticletitle{Adapting cache line size to application behavior}.
  In \bibinfo{booktitle}{\emph{Proceedings of the 13th international conference
  on Supercomputing}}. \bibinfo{pages}{145--154}.
\newblock


\bibitem[\protect\citeauthoryear{Wang, Yi, Guo, Jin, Xu, Li, Wang, Guo, Li, Xu,
  et~al\mbox{.}}{Wang et~al\mbox{.}}{2021}]%
        {Milvus-SIGMOD-2021}
\bibfield{author}{\bibinfo{person}{Jianguo Wang}, \bibinfo{person}{Xiaomeng
  Yi}, \bibinfo{person}{Rentong Guo}, \bibinfo{person}{Hai Jin},
  \bibinfo{person}{Peng Xu}, \bibinfo{person}{Shengjun Li},
  \bibinfo{person}{Xiangyu Wang}, \bibinfo{person}{Xiangzhou Guo},
  \bibinfo{person}{Chengming Li}, \bibinfo{person}{Xiaohai Xu},
  {et~al\mbox{.}}} \bibinfo{year}{2021}\natexlab{}.
\newblock \showarticletitle{Milvus: A purpose-built vector data management
  system}. In \bibinfo{booktitle}{\emph{Proceedings of the 2021 International
  Conference on Management of Data}}. \bibinfo{pages}{2614--2627}.
\newblock


\bibitem[\protect\citeauthoryear{Wang, Zhang, Sebe, Shen, et~al\mbox{.}}{Wang
  et~al\mbox{.}}{2017}]%
        {Learning-Hash-survey-2017-PAMI}
\bibfield{author}{\bibinfo{person}{Jingdong Wang}, \bibinfo{person}{Ting
  Zhang}, \bibinfo{person}{Nicu Sebe}, \bibinfo{person}{Heng~Tao Shen},
  {et~al\mbox{.}}} \bibinfo{year}{2017}\natexlab{}.
\newblock \showarticletitle{A survey on learning to hash}.
\newblock \bibinfo{journal}{\emph{IEEE transactions on pattern analysis and
  machine intelligence}} \bibinfo{volume}{40}, \bibinfo{number}{4}
  (\bibinfo{year}{2017}), \bibinfo{pages}{769--790}.
\newblock


\bibitem[\protect\citeauthoryear{Yang, Li, Jin, Zhong, Wang, Shen, Jia, and
  Wang}{Yang et~al\mbox{.}}{2024c}]%
        {yang2024ddc}
\bibfield{author}{\bibinfo{person}{Mingyu Yang}, \bibinfo{person}{Wentao Li},
  \bibinfo{person}{Jiabao Jin}, \bibinfo{person}{Xiaoyao Zhong},
  \bibinfo{person}{Xiangyu Wang}, \bibinfo{person}{Zhitao Shen},
  \bibinfo{person}{Wei Jia}, {and} \bibinfo{person}{Wei Wang}.}
  \bibinfo{year}{2024}\natexlab{c}.
\newblock \showarticletitle{Effective and General Distance Computation for
  Approximate Nearest Neighbor Search}.
\newblock \bibinfo{journal}{\emph{arXiv preprint arXiv:2404.16322}}
  (\bibinfo{year}{2024}).
\newblock


\bibitem[\protect\citeauthoryear{Yang, Li, and Wang}{Yang
  et~al\mbox{.}}{2024b}]%
        {Fast-Index-arxiv-2024}
\bibfield{author}{\bibinfo{person}{Mingyu Yang}, \bibinfo{person}{Wentao Li},
  {and} \bibinfo{person}{Wei Wang}.} \bibinfo{year}{2024}\natexlab{b}.
\newblock \showarticletitle{Fast High-dimensional Approximate Nearest Neighbor
  Search with Efficient Index Time and Space}.
\newblock \bibinfo{journal}{\emph{arXiv preprint arXiv:2411.06158}}
  (\bibinfo{year}{2024}).
\newblock


\bibitem[\protect\citeauthoryear{Yang, Xie, Liu, Yu, Gao, Wang, Peng, and
  Cui}{Yang et~al\mbox{.}}{2025}]%
        {yang2025revisitingindexconstructionproximity}
\bibfield{author}{\bibinfo{person}{Shuo Yang}, \bibinfo{person}{Jiadong Xie},
  \bibinfo{person}{Yingfan Liu}, \bibinfo{person}{Jeffrey~Xu Yu},
  \bibinfo{person}{Xiyue Gao}, \bibinfo{person}{Qianru Wang},
  \bibinfo{person}{Yanguo Peng}, {and} \bibinfo{person}{Jiangtao Cui}.}
  \bibinfo{year}{2025}\natexlab{}.
\newblock \bibinfo{title}{Revisiting the Index Construction of Proximity
  Graph-Based Approximate Nearest Neighbor Search}.
\newblock
\newblock
\showeprint[arxiv]{2410.01231}~[cs.DB]


\bibitem[\protect\citeauthoryear{Yang, Hu, Peng, Li, Li, Wang, and Liu}{Yang
  et~al\mbox{.}}{2024a}]%
        {2024vdtuner}
\bibfield{author}{\bibinfo{person}{Tiannuo Yang}, \bibinfo{person}{Wen Hu},
  \bibinfo{person}{Wangqi Peng}, \bibinfo{person}{Yusen Li},
  \bibinfo{person}{Jianguo Li}, \bibinfo{person}{Gang Wang}, {and}
  \bibinfo{person}{Xiaoguang Liu}.} \bibinfo{year}{2024}\natexlab{a}.
\newblock \showarticletitle{VDTuner: Automated Performance Tuning for Vector
  Data Management Systems}. In \bibinfo{booktitle}{\emph{2024 IEEE 40th
  International Conference on Data Engineering (ICDE)}}.
  \bibinfo{pages}{4357--4369}.
\newblock


\bibitem[\protect\citeauthoryear{YouTube}{YouTube}{2025}]%
        {youtube}
\bibfield{author}{\bibinfo{person}{YouTube}.} \bibinfo{year}{2025}\natexlab{}.
\newblock \bibinfo{title}{\text{Video Search}}.
\newblock \bibinfo{howpublished}{\url{https://www.youtube.com/}}.
\newblock


\bibitem[\protect\citeauthoryear{Zhang, Yao, Lai, Huang, Fang, Tao, Song, and
  Liu}{Zhang et~al\mbox{.}}{2025}]%
        {zhang2025reasoningreinforced}
\bibfield{author}{\bibinfo{person}{Kongcheng Zhang}, \bibinfo{person}{Qi Yao},
  \bibinfo{person}{Baisheng Lai}, \bibinfo{person}{Jiaxing Huang},
  \bibinfo{person}{Wenkai Fang}, \bibinfo{person}{Dacheng Tao},
  \bibinfo{person}{Mingli Song}, {and} \bibinfo{person}{Shunyu Liu}.}
  \bibinfo{year}{2025}\natexlab{}.
\newblock \bibinfo{title}{Reasoning with Reinforced Functional Token Tuning}.
\newblock
\newblock
\showeprint[arxiv]{2502.13389}~[cs.AI]


\bibitem[\protect\citeauthoryear{Zheng, Zhao, Weng, Hung, Liu, and
  Jensen}{Zheng et~al\mbox{.}}{2020}]%
        {PMLSH-bolong-2020}
\bibfield{author}{\bibinfo{person}{Bolong Zheng}, \bibinfo{person}{Xi Zhao},
  \bibinfo{person}{Lianggui Weng}, \bibinfo{person}{Nguyen Quoc~Viet Hung},
  \bibinfo{person}{Hang Liu}, {and} \bibinfo{person}{Christian~S. Jensen}.}
  \bibinfo{year}{2020}\natexlab{}.
\newblock \showarticletitle{{PM-LSH:} {A} Fast and Accurate {LSH} Framework for
  High-Dimensional Approximate {NN} Search}.
\newblock \bibinfo{journal}{\emph{Proc. {VLDB} Endow.}} \bibinfo{volume}{13},
  \bibinfo{number}{5} (\bibinfo{year}{2020}), \bibinfo{pages}{643--655}.
\newblock


\bibitem[\protect\citeauthoryear{Zhou, Lu, Li, and Tian}{Zhou
  et~al\mbox{.}}{2012}]%
        {2012sq}
\bibfield{author}{\bibinfo{person}{Wengang Zhou}, \bibinfo{person}{Yijuan Lu},
  \bibinfo{person}{Houqiang Li}, {and} \bibinfo{person}{Qi Tian}.}
  \bibinfo{year}{2012}\natexlab{}.
\newblock \showarticletitle{Scalar quantization for large scale image search}.
  In \bibinfo{booktitle}{\emph{Proceedings of the 20th {ACM} Multimedia
  Conference, {MM} '12, Nara, Japan, October 29 - November 02, 2012}},
  \bibfield{editor}{\bibinfo{person}{Noboru Babaguchi},
  \bibinfo{person}{Kiyoharu Aizawa}, \bibinfo{person}{John~R. Smith},
  \bibinfo{person}{Shin'ichi Satoh}, \bibinfo{person}{Thomas Plagemann},
  \bibinfo{person}{Xian{-}Sheng Hua}, {and} \bibinfo{person}{Rong Yan}} (Eds.).
  \bibinfo{publisher}{{ACM}}, \bibinfo{pages}{169--178}.
\newblock


\bibitem[\protect\citeauthoryear{Zitzler and Thiele}{Zitzler and
  Thiele}{1999}]%
        {1999parento}
\bibfield{author}{\bibinfo{person}{E. Zitzler} {and} \bibinfo{person}{L.
  Thiele}.} \bibinfo{year}{1999}\natexlab{}.
\newblock \showarticletitle{Multiobjective evolutionary algorithms: a
  comparative case study and the strength Pareto approach}.
\newblock \bibinfo{journal}{\emph{IEEE Transactions on Evolutionary
  Computation}} \bibinfo{volume}{3}, \bibinfo{number}{4}
  (\bibinfo{year}{1999}), \bibinfo{pages}{257--271}.
\newblock


\end{thebibliography}

\newpage
\appendix
\section{Detail Explanation of Deterministic Access Greedy Search}
The initialization phase establishes candidate and visited sets (Lines 1-3). The candidate set dynamically maintains the nearest points discovered during traversal, while the visited set tracks explored points. The algorithm iteratively extracts the nearest unvisited point from the candidate set and processes its neighbors through four key phases (Lines 4-17).

At each iteration, the nearest unvisited point $x_i$ is retrieved from the candidate set (Line 5). Batch processing then identifies \textit{valid} neighbors through three filtering criteria (Lines 6-10): 1) unvisited ($j \notin V$), 2) label validity ($L_{j} \le \alpha_s$), and 3) degree constraint ($|N| \le m_s$). The latter two conditions implement automatic index-level parameters tuning as detailed in Section~\ref{sec:tech2_auto_param_opt}. It is important to note that we differentiate between the identifier $j$ and vector data $x_j$ - a distinction that impacts cache performance as memory access may incur cache misses while identifier operations do not. 

The prefetching mechanism operates through two parameters: $\omega$ controls the prefetch stride (Lines 11, 14), while $\nu$ specifies the prefetch depth (Lines 12, 15). The batch distances computation of neighbors follows a three-stage pipeline (Lines 13-17): 1) stride prefetch for subsequent points (Lines 14-15), 2) data access(Line 16), and 3) distance calculation with candidate set maintenance (Line 17). When $|C| \ge ef_s$, we pop farthest-point in $C$ to keep $|C| \le ef_s$ (Line 17). When all candidate points are processed, we apply selective re-ranking (see Section~\ref{sec:tech3_distance_opt}) to enhance the distances precision (Line 18). $T$ is the high precision distances between points in $ANN_k(x_q)$ and $x_q$. Finally, the algorithm returns $ANN_k(x_q)$ with related distances $T$ (Line 19).

\section{VSAG Index Construction Algorithm.}
\label{sec:index_construction}
In \vsag, the insertion process for each new point $x_i$ (Lines 3--4) can be divided into three steps. First, a search is conducted on the graph to obtain the approximate nearest neighbors $\mathrm{ANN}_k(x_i)$ and their corresponding distances $T_i$ (Line 5). Next, pruning is performed on $\mathrm{ANN}_k(x_i)$, and the edges along with their labels are added to the graph index (Line 6). The first step, the search process, is similar to that of other graph-based index constructions. However, the second step, the pruning process, differs from other indexes. Instead of performing actual pruning, \vsag assigns labels. Specifically, \vsag assigns a label $L_{i, j}$ to each edge $G_{i, j}$ in graph (see Algorithm \ref{algo:pruning}). \textit{Note that the label is used to indicate whether an edge exists under a certain construction parameter.} Given $\alpha_s$, if $L_{i, j} \leq \alpha_s$, it indicates that the edge $G_{i, j}$ exists in the graph constructed using $\alpha_s$ or larger value as the pruning rate. During the search and pruning processes, only edges satisfying $L_{i, j} \leq \alpha_s$ need to be considered to retrieve the graph corresponding to the given construction parameters, which is proven in Theorem~\ref{theo:alpha}.

Next, for each new neighbor, reverse edges are added, and pruning is performed (Lines 7--14). If the distance $T_{i, j}$ from the current neighbor $x_j$ to $x_i$ is greater than the maximum distance among $x_j$'s neighbors and $x_j$'s neighbor list is already full (Line 8), no reverse edge is added (Line 9). Otherwise, $x_i$ is inserted into $x_j$'s neighbor list at an appropriate position $r$ to maintain $G_{j}$ and $T_j$ in ascending order (Lines 10--13). Subsequently, pruning and labeling are reapplied to the neighbors of $x_j$ located after position $r$. Finally, the algorithm returns the edges $G$ and labels $L$ of the graph index. 

\begin{algorithm}[t!]
	\DontPrintSemicolon
	\KwIn{ 
		dataset $D$, 
		pruning rates $A$ sorted in ascending order, 
		candidate size $ef_c$, 
		maximum degree $m_c$, 
	}
	\KwOut{ directed graph $G$ with out-degree $\le m_c$}
	\SetKwFunction{GreedySearch}{GreedySearch}
	\SetKwFunction{PruneAssignLabel}{PruneAssignLabel}
	
	initialize $G$ to an empty graph \;
	initialize $L$, $T$ to empty lists \;
	
	\For {$0 \le i < |D|$} {
		$x_i \leftarrow D_i$ \;
		\tcp{search ANN}
		$\mathrm{ANN}_k(x_i), T_i \leftarrow $ GreedySearch($G, L, D, \{x_0\},x_i, ef_c, m_c, \max(A), 1$) \;
		
		\tcp{prune}
		$G_i, L_i \leftarrow$ PrunebasedLabeling($x_i,\mathrm{ANN}(x_i),T_i, A, m_c, 0$)\;
		
		\ForEach{ $j \in G_i$} {
			\If{$T_{i, j} \ge T_{j, -1} $ \textbf{\textup{and}}  $|G_j| \ge m_c$} {
				continue;
			} \Else {
				\tcp{insertion maintaining ascending order} 
				$r \leftarrow $ minimum id s.t. $T_{j, r} > T_{i, j}$ \;            
				Insert $x_i$ right before $G_{j, r}$ \;
				Insert $T_{i, j}$ right before $T_{j, r}$\;
				\tcp{add reverse edges and prune}
				$G_j, L_j \leftarrow$ PrunebasedLabeling($x_j,G_{j, r},T_j, A, m_c, r$)\;
			}
		}
	}

	\Return{$G, L$} \;
	
	\caption{\vsag Index Construction}
	\label{algo:construction}
\end{algorithm}

\section{Proof of Theorem}
\subsection{Proof of Theorem \ref{theo:max_degree}}
\label{sec:proof_of_theorem_1}
\begin{proof}
	The only difference in construction under varying maximum degrees lies in the termination timing of the loop. Note that during the edge connection process, existing edges are never deleted. Specifically, when inserting $x_i$, the neighbors edges are continuously added to $G_i$ until $|G_i| == a$ or $|G_i| == b$. Since $a < b$, the condition $|G_i| == a$ will be satisfied first during edge addition. In the construction process with maximum degree $b$, the first $a$ edges added in the loop are identical to those in $G^a_i$. The subsequent $b - a$ edges inserted do not remove the first $a$ edges. Thus, $G^a_i == G^b_{i, 0:a}$.
\end{proof}
\subsection{Proof of Theorem \ref{theo:alpha}}
\label{sec:proof_of_theorem_2}
\begin{proof}
	We prove this by mathematical induction. Consider the insertion process of $G^a_i$ and $G^b_i$, examining each point $x_j \in ANN_a(i)$ sorted by distance in ascending order.
	
	\noindent \textbf{Base Case:} When inserting $x_1$, since both $G^a_i$ and $G^b_i$ are initially empty, no edges can trigger pruning, so $x_1$ is guaranteed to be inserted. Thus, $G^a_i = G^b_i = \{x_1\}$, and we have $G^a_i \subseteq G^b_i$.
	
	\noindent \textbf{First Iteration: Inserting $x_2$:} We analyze the cases:
	\begin{itemize}[leftmargin=*]
		\item \textbf{1: Both sets remain unchanged.} If $b \cdot dis(x_2, x_1) \leq dis(x_i, x_2)$, $x_2$ will not be added to $G^b_i$. Since $a < b$, it follows that $a \cdot dis(x_2, x_1) \leq dis(x_i, x_2)$, so $x_2$ will also not be added to $G^a_i$. In this case, $G^a_i = G^b_i = \{x_1\}$.
		
		\item \textbf{2: Both sets change.} If $a \cdot dis(x_2, x_1) > dis(x_i, x_2)$, since $a < b$, it follows that $b \cdot dis(x_2, x_1) > dis(x_i, x_2)$. In this case, $G^a_i = G^b_i = \{x_1, x_2\}$.
		
		\item \textbf{3: Only $G^b_i$ changes.} If $a \cdot dis(x_2, x_1) \leq dis(x_i, x_2)$ and $b \cdot dis(x_2, x_1) > dis(x_i, x_2)$, then $G^a_i = \{x_1\} \subset \{x_1, x_2\} = G^b_i$.
		
		\item \textbf{4: Only $G^a_i$ changes.} If $a \cdot dis(x_2, x_1) > dis(x_i, x_2)$ and $b \cdot dis(x_2, x_1) \leq dis(x_i, x_2)$, this situation cannot occur because $a < b$. Thus, for the second insertion, we still have $G^a_i \subseteq G^b_i$.
	\end{itemize}
	
	\noindent \textbf{Inductive Step: Inserting $x_{n+1}$:} Assume that before the $n$-th insertion, $G^a_i \subseteq G^b_i$. We analyze the cases:
	\begin{itemize}[leftmargin=*]
		\item \textbf{1: Both sets remain unchanged.} If $\exists x_j \in G^a_i$ such that $b \cdot dis(x_{n+1}, x_j) \leq dis(x_i, x_{n+1})$, then $a \cdot dis(x_{n+1}, x_j) \leq dis(x_i, x_{n+1})$ must also hold, so $G^a_i$ remains unchanged. Since $G^a_i \subseteq G^b_i$, there must exist $x_j \in G^b_i$ such that $b \cdot dis(x_{n+1}, x_j) \leq dis(x_i, x_{n+1})$, so $G^b_i$ also remains unchanged. Thus, $G^a_i \subseteq G^b_i$.
		
		\item \textbf{2: Both sets change.} If $\forall x_j \in G^b_i$, $b \cdot dis(x_{n+1}, x_j) > dis(x_i, x_{n+1})$, and $\forall x_j \in G^a_i$, $a \cdot dis(x_{n+1}, x_j) > dis(x_i, x_{n+1})$, these conditions are compatible and may occur. In this case, $G^a_i = G^a_i \cup \{x_{n+1}\} \subseteq G^b_i \cup \{x_{n+1}\} = G^b_i$.
		
		\item \textbf{3: Only $G^b_i$ changes.} If $\exists x_j \in G^a_i \ s.t.\ a \cdot dis(x_{n+1}, x_j) \leq dis(x_i, x_{n+1})$, and $\forall x_j \in G^b_i$, $b \cdot dis(x_{n+1}, x_j) > dis(x_i, x_{n+1})$, these conditions are compatible and may occur. In this case, $G^a_i \subseteq G^a_i \cup \{x_{n+1}\} = G^b_i$.
		
		\item \textbf{4: Only $G^a_i$ changes.} This case does not exist because Lines 5-6 of Algorithm \ref{algo:pruning} ensure that if a point is not be pruned in $G^a_i$, it will not be pruned again in $G^b_i$. Therefore, if $x_{n+1}$ is added to $G^a_i$, its label is set as $L_{n+1} = a$. During subsequent iterations of $\alpha$, the condition in Lines 5-6 of Algorithm \ref{algo:pruning} ensures that this point is no longer considered. 
	\end{itemize}
\end{proof}

\section{Selecting Optimal ILP Configurations.} 
\label{sec:secectILP}
Under specified optimization objectives for the accuracy and efficiency of retrieval, different configurations of parameter at the index level represent candidate solutions in the design space. When environment-level and query-level parameters remain fixed, each configuration induces unique performance characteristics in terms of target metrics. The complete combinatorial space of index-level parameters constitutes the solution domain for multi-objective optimization. \vsag's automated tuning mechanism aims to identify a solution subset where no configuration in the space is Pareto superior~\cite{1999parento} to any member of this subset. This set of non-dominated solution forms the Pareto Optimal Frontier~\cite{1999parento}, representing optimal trade-offs between competing objectives.

Through systematic evaluation of ILP configurations, \vsag generates performance profiles containing retrieval accuracy and latency measurements, enabling Pareto Frontier derivation. When users specify target accuracy thresholds or latency constraints, the system performs optimal configuration selection by identifying the minimal Pareto-optimal solution that satisfies the specified requirements.

\begin{example}
	Consider a user-defined constraint of $Recall@10 > 90\%$ with three Pareto Frontier candidates: $(A, 91\%, 2000)$, $(B, 90\%, 2100)$, $(C, 89\%, 2200)$, where tuples denote (configuration, Recall@10, queries per second). Configuration B emerges as the optimal selection, satisfying the recall threshold while maximizing query throughput through its superior QPS performance.
\end{example}

\section{Tuning Cost of ILPs}
\label{sec:tuning_cost_ilp}

\begin{table}[t!]
	\small
	\centering%\vspace{-2ex} 
	\caption{Comparison of Tuning Cost of ILPs.}\vspace{-3ex} \label{tab:tuning_cost}
	\begin{tabular}{l|c|c|c|c|c|c}
		\toprule
		\multirow{2}{*}{\textbf{Dataset}} & \multicolumn{2}{c|}{\textbf{FIX}} & \multicolumn{2}{c|}{\textbf{BF}} & \multicolumn{2}{c}{\textbf{\vsag (Ours)}} \\
		\cmidrule(lr){2-3} \cmidrule(lr){4-5} \cmidrule(lr){6-7}
		& \textbf{Mem} & \textbf{Time} & \textbf{Mem} & \textbf{Time} & \textbf{Mem} & \textbf{Time} \\
		\midrule
		GIST1M         & 3.83G     & 3.20H  & 89.87G     & 61.64H    & 4.07G    & 2.92H        \\
		SIFT1M         & 0.73G     & 0.85H   & 15.48G     & 30.79H    & 0.97G    & 1.86H      \\
		GLOVE-100      & 0.74G     & 1.22H   & 15.36G     & 41.94H    & 1.03G    & 2.13H \\
		\bottomrule
	\end{tabular}\vspace{-2ex}
\end{table}

We evaluate the tuning time and memory footprint when varying index-level parameters $m_c \in(8,16,24,32)$ and $\alpha_c\in(1.0,1.2,1.4,$ $1.6,1.8,2.0)$. The compared method including \textit{FIX} (fix parameter $m_c = 32$ and $\alpha_c = 2.0$, building with \textit{hnswlib}) and \textit{BF} (construct indexes of all combination of index-level parameters, building with \textit{hnswlib}).  Table~\ref{tab:tuning_cost} shows results (\textit{Mem} for memory footprint in gigabytes; \textit{Time}  for building time in hours).

Among the four datasets, the method with minimal construction time and memory usage is \textit{FIX}, while \textit{BF} requires the most. This is because baseline methods can only brute-force traverse each parameter configuration and build indexes separately. Additionally, it can dynamically select edges during search using edge labels acquired during construction, enabling parameter adjustment at search time without repeated index construction. This allows \vsag to achieve a $19\times$ tuning time saving compared to \textit{BF} on GIST1M. Moreover, the use of $T$ to cache distances during construction eliminates redundant computations, further accelerating \vsag's tuning. In terms of memory usage, graphs built with different $m_c$ and $a_c$ share significant structural overlap. This enables \vsag to compress these graphs via edge labels, resulting in over $19\times$ memory footprint reduction during tuning compared to \textit{BF}. Compared to \textit{FIX}, \vsag only introduces minimal additional memory costs for storing extra edges and label sets.

\section{Applications}
\label{sec:application}

\stitle{Product Search and Recommendation Systems.}
Vector retrieval has become a cornerstone in product search and recommendation engines for e-commerce. By converting signals such as user behavior, browsing history, and purchase data into vector representations, similarity-based retrieval can pinpoint products that match user preferences, thereby enabling personalization at scale.
E-commerce platforms~\cite{taobao, amazon} typically maintain millions of product and user vectors, which demand high-performance and low-latency retrieval solutions.  
Leveraging \vsag's \textit{memory access optimization} and \textit{distance computation optimization}, these services can achieve the same target QPS with fewer computational resources.  
In a product search scenario at Ant Group, which involves 512-dimensional embeddings, \vsag not only lowered overall resource usage by a factor of 4.6$\times$ but also cut service latency by 1.2$\times$, surpassing contemporary partition-based methods.

In practice, production environments span various \emph{resource classes} (e.g., 4C16G vs.\ 2C8G) and \emph{platforms} (e.g., x86 vs.\ ARM).  
Realizing optimal performance across this heterogeneous infrastructure required extensive parameter tuning, with limited transferability once resource classes changed.  
\vsag’s \emph{automatic parameter tuning} feature relieves engineers from repeated manual fine-tuning, drastically reducing the overhead of deploying and maintaining vector retrieval services.  
This translates into robust performance across diverse platforms and ever-shifting requirements.

\stitle{Large Language Model.} 
In large language model~(LLM) applications~\cite{LLM-RAG-NIPS-2020}, high-dimensional text embeddings are essential for reducing hallucinations and delivering \textit{real-time} interactive responses.  
Modern embedding models, such as  \texttt{text-embedding-3-small}, produce 1{,}536-dimensional vectors, while \texttt{text-embedding-3-large} outputs up to 3{,}072 dimensions—significantly more than the typical 256-dimensional image embeddings.  
Effective retrieval over these large embedding vectors demands both precise distance computation and careful memory management. \vsag’s dedicated \emph{distance computation} and \emph{memory access} optimizations excel at handling these high-dimensional workloads.

\emph{Deep Research} is quickly emerging as a next-generation standard for LLM-based applications, functioning as a specific form of \emph{agentic search} that relies on chain-of-thought reasoning~\cite{shen2025codi} and multi-step retrieval~\cite{ji2025enhancingmultistepreasoning, zhang2025reasoningreinforced}.
In Retrieval-Augmented Generation (RAG)~\cite{gao2024RAGsurvey}, this process repeatedly refines queries and explores multiple reasoning paths while adaptively gathering domain-specific information.
Such iterative lookups impose substantial demands on both system throughput and latency, particularly with high-dimensional vectors (e.g., 1{,}536 dimensions).
\vsag delivers up to a \textbf{5$\times$} improvement in retrieval capacity, critical for maintaining real-time responsiveness in the context-aware exploration that agentic search enables.

\begin{figure}[t!]\centering
	\scalebox{0.5}[0.5]{\includegraphics{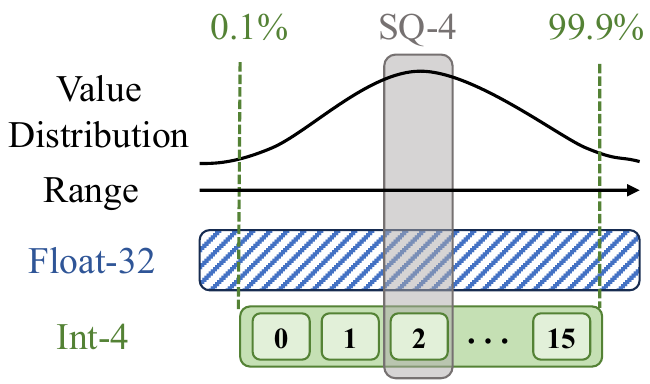}}
	\caption{\small Truncated scalar quantization. }
	\label{fig:tech3-SQ}
\end{figure}

\section{Improving Quantization Precision}
\label{sec:improve_quantization_pre}
Scalar quantization (SQ) compresses floating-point vectors by mapping 32-bit float values to lower-bit integer representations in each dimension. In SQ-b (where b denotes bit-width), the dynamic range is uniformly partitioned into $2^b$ intervals. For SQ4 (4-bit) on vector $x$, this creates 16 intervals over $[\min(x),\max(x)]$. Values of each dimension of $x$ within each interval are encoded as their corresponding integer indices. However, practical implementations face critical range estimation challenges: direct use of observed min/max values proves suboptimal when data outliers existing. 

\begin{example}
	Analysis of the GIST1M dataset reveals this limitation – while 99\% of dimensions exhibit values below 0.3, using the absolute maximum (1.0) would leave 70\% of the quantization intervals (0.3-1.0) underutilized. This interval under-utilization severely degrades quantization precision.
\end{example}

To enhance robustness, we propose \textit{Truncated Scalar Quantization} using the 99th percentile statistics rather than absolute extremes. This approach discards outlier-induced distortions while preserving quantization resolution over the primary data distribution, achieving superior balance between compression efficiency and numerical precision as shown in Figure \ref{fig:tech3-SQ}.

\end{document}